\documentclass{article}%
\usepackage{amsfonts}
\usepackage{amsmath}
\usepackage{fancybox}
\usepackage{appendix}%
\usepackage{amssymb}%
\usepackage[dvips]{graphicx}
\usepackage{graphicx}
\usepackage{psfrag}

\providecommand{\U}[1]{\protect\rule{.1in}{.1in}}
\newtheorem{theorem}{Theorem}

\newtheorem{corollary}[theorem]{Corollary}

\newtheorem{definition}[theorem]{Definition}

\newtheorem{lemma}[theorem]{Lemma}

\newtheorem{proposition}[theorem]{Proposition}

\newenvironment{proof}[1][Proof]{\noindent\textbf{#1.} }{\ \rule{0.5em}{0.5em}}
\begin{document}

\title{The massive Dirac equation in the Kerr-Newman-de
Sitter and Kerr-Newman black hole spacetimes}
\author{G. V. Kraniotis \footnote{email: gkraniot@cc.uoi.gr}\\
University of Ioannina, Physics Department \\ Section of
Theoretical Physics, GR- 451 10, Greece \\
}

 \maketitle

\begin{abstract}
Exact solutions of the Dirac general relativistic equation that describe the dynamics of a massive, electrically charged  particle with half-integer spin in the curved spacetime geometry of an electrically charged, rotating Kerr-Newman-(anti) de Sitter black hole are investigated. We first, derive the Dirac equation in the Kerr-Newman-de Sitter (KNdS) black hole background using a generalised Kinnersley null tetrad in the Newman-Penrose formalism. Subsequently in this frame and in the KNdS black hole spacetime, we prove the separation of the Dirac equation into ordinary differential equations for the radial and angular parts. Under specific transformations of the independent and dependent variables we prove that the transformed radial equation for a massive charged spin $\frac{1}{2}$ fermion in the background KNdS black hole constitutes a highly non-trivial generalisation of Heun's equation since it possess five regular finite singular points. Using a Regge-Wheeler-like independent variable we transform the radial equation in the KNdS background into a Schr\"{o}dinger like differential equation and investigate its asymptotic behaviour near the event and cosmological horizons.
For the case of a massive fermion in the background of a Kerr-Newman (KN) black hole we first prove that the radial and angular equations that result from the separation of Dirac's equation reduce to the generalised Heun differential equation (GHE). The local solutions of such GHE are derived and can be described by holomorphic functions whose power series coefficients are determined by a four-term recurrence relation. In addition using asymptotic analysis we derive the solutions for the massive fermion far away from the KN black hole and the solutions near the event horizon . The determination of the separation constant as an eigenvalue problem in the KN background is investigated. Using the aforementioned four-term recursion formula we prove that in the non-extreme KN geometry there are no bound states with $\omega^2<\mu^2$, where $\omega$ and $\mu$ are the energy and mass of the fermion respectively.

\end{abstract}

\section{Introduction}

The problem of massive perturbations in the strong gravity background of a black hole is fascinating and fundamentally significant problem, as has been demonstrated recently for the case of scalar perturbations \cite{Kraniotis1},\cite{Bezerra},\cite{WU}  which has acquired extra impetus after the discovery of gravitational waves in Nature \cite{GW150914},\cite{GW151226}.

In particular in \cite{Kraniotis1} the Klein-Gordon-Fock equation for a massive charged particle in the background of the Kerr-Newman-(anti) de Sitter black hole was separated. The resulting radial and angular equations, for particular values of the inverse Compton wavelength in terms of the cosmological constant $\Lambda$, were reduced to a Heun form. Then both Heun equations were solved in terms of an infinite series of hypergeometric functions using the idea of augmented convergence .
In this setup, the solution converges inside the ellipse with foci at two of the finite regular singularities and passing through the third finite regular singular point with the possible exception of the line connecting the two foci. This method of constructing a Heun function \footnote{The idea of augmented convergence was first introduced in \cite{Erdelyi}-\cite{Svartholm}.} offers a perspective on the notorious connection problem that we will discuss in more detail in this paper. In particular, the solution obtained for the massive radial Heun equation in KNdS spacetime converges in the ellipse with foci at the event and Cauchy horizons of the black hole (BH)\cite{Kraniotis1}.
We remind the reader that Heun's differential equation (HE)  belongs to the class of Fuchsian differential equations (FDE), since it is the most general linear differential equation with \textit{four regular singular points} \cite{KARLHEUNmunich},\cite{RONVEAUX}.
In \cite{Kraniotis1} the concept of false singularity was discussed, which reduces the solution to hypergeometric functions for certain values of the physical parameters in the equation.
The general case, however, is that the solution of the KGF equation with the method of separation of variables, for a rotating charged cosmological black hole, results in FDE for the radial and angular parts which for most of the parameter space contain more than three finite singularities and thereby generalise the Heun differential equations.
In the absence of the cosmological constant the solutions obtained were given in terms of confluent Heun functions \cite{Kraniotis1}.

It is of importance to extend our knowledge of the issue of massive perturbations in black hole curved backgrounds and go beyond the scalar degrees of freedom by  examining massive particle perturbations with non-zero spin in the strong gravity space-time of a black hole.

More specifically it is the purpose of this paper to investigate the general relativistic Dirac equation, in particular its separation and its analytic solutions, in the important Kerr-Newman-(anti) de Sitter and Kerr-Newman black hole backgrounds. These spacetime backgrounds constitute the most general exact solutions of the theory of general relativity that describe rotating, massive and electrically charged black holes.

The geometrisation of the dynamics of relativistic spin-$\frac{1}{2}$ fermions \footnote{Dirac had presented a special relativistic generalisation of the Schr\"{o}dinger equation in \cite{PMDiracI}.} using the concept of parallel displacement in general relativity was initiated in the work of \cite{V. Fock}. It acquired a new boost after the fundamental discovery of Chandrasekhar that the massive Dirac equation can be separated in Kerr geometry into time and azimuthal angle modes and rewritten in terms of radial wave equations and coupled angular ordinary differential equations \cite{Chandrasekhar}\footnote{See \cite{Saul} for a nice recent account.} . In this framework, the Newman-Penrose formalism was indeed proved to be a fundamental tool for the study of spin fields in curved space-time \cite{NPformalism},\cite{spinors}. The massive Dirac equation in the Kerr-Newman (KN) background was separated by D.Page \cite{page} using the Kinnersley tetrad \cite{WKinnersley}, while the separated radial part of the massive Dirac equation in the KN and Reisser-Nordstr\"om metrics using a Carter tetrad was studied in \cite{Baticr}. Further investigations include the separation of the Dirac equation in the Kerr-de Sitter space-time \cite{Khanal} as well as the study of gravitational perturbations in the KN metric \cite{Bose}. The massive Dirac equation has been investigated in the Hamiltonian formulation in various axisymmetric geometries in \cite{Finster}-\cite{Belgiorno}. Despite all these efforts though, as we mentioned earlier, little is known about the theory of explicit exact solutions of the separated ODEs for the massive fermion in such black hole backgrounds.

Remarkably enough as we shall see in the present paper, the separation of the massive Dirac equation in the Kerr-Newman-(anti) de Sitter and Kerr-Newman black hole backgrounds results in radial and angular differential equations which also generalise the HE. In particular we first prove  that the transformed radial equation for a massive charged spin $\frac{1}{2}$ fermion in the background KNdS black hole possess five regular finite singular points and it thus constitute a highly non-trivial generalisation of Heun's equation.
In the KN black hole background both radial and angular equations lead to the mathematical structure of the \textit{generalised Heun equation} (GHE) in which the singularity at infinity constitutes an irregular singular point. The GHE was introduced and discussed for the first time in the mathematical literature in \cite{RSCHAFKEDSCHMIDT}.
The theory of the solutions of the GHE is of even richer mathematical structure and it is of paramount importance for solving the radial and angular equations for a massive fermion in the curved KN black hole background.

In this respect we note that in the theory of complex differential equations a study of the \textit{global} behaviour of solutions is one of the most interesting and difficult problems. Specifically, such a global problem for linear ordinary differential equations (ODEs)  consists in finding explicit connection relations between the local solutions at two different (regular or singular) points $z_0$ and $z_1$ \cite{OLVER}. This is referred to as a two point connection problem. Interestingly enough, such two point connection problems i) for local solutions at two regular singularities ii) between  a regular singularity and an irregular singularity have been investigated for the important case of the GHE differential equation in the works \cite{RSCHAFKEDSCHMIDT},\cite{RegIrregSIAMJ}.
This is of physical interest as we shall see: for example for the GHE radial ODE that results from separation of the massive Dirac equation in the Kerr-Newman background a two point connection problem between the local solution around the regular singularity at the event horizon and the solution in the vicinity of the irregular singularity at $\infty$ can be set up and solved.

The material of this paper is organised as follows. In section $2$ using a generalised Kinnersley null tetrad in the Newman-Penrose formalism in the KN(a)dS spacetime we calculate the corresponding Ricci-rotation coeffients (\ref{cosmoconsricci}) and the Dirac equations (\ref{PaulDKNdS}), (\ref{DiracKNdSspacetime}). Subsequently, using the separation ansatz (\ref{ansatzseperation}), we prove the separability of the Dirac equation in the KNdS spacetime and the corresponding radial and angular ODEs (\ref{ALGMEINEDIR1})-(\ref{ALGEMEINEDIR4}).
In \textsection \ref{fundamentalKNdS} we prove by applying appropriate transformations of the dependent and independent variables that the separated massive radial equation for the spin half charged fermion in KNdS spacetime is a highly non-trivial generalisation of the Heun equation since it possess five regular finite singularities-see equation (\ref{MasDiracGHrKNdS}).
Furthermore in \ref{AsymptotsLambdaRegge} we derive the asymptotic forms of the radial equation in the KNdS black hole background. Using a Regge-Wheeler-like coordinate, eqns. (\ref{ReggeW}),(\ref{LambdaRWcoord}) we first transform the massive radial equation in the presence of the cosmological constant into eqn.(\ref{tortoiseDELambda}). We then derive the asymptotics of (\ref{tortoiseDELambda}) near the event horizon and near the cosmological horizon, equations (\ref{eventhorizonKNdS}) and (\ref{cosmohorizonLambda}) respectively.
In section \ref{KNDiracSection} we derive in the zero cosmological constant limit the separated radial and angular ODEs. In sections \ref{anguGHE} and \ref{AKTINAGHE} we transform both radial and angular ODEs into the generalised Heun form, eqns.(\ref{GHeunMunchenPDIRACmas}),(\ref{GONGHE}) and (\ref{GHERADIALKN}) respectively.
In section \ref{SolvingGHE} we discuss the analytic local solutions of the GHE. As is shown there, the coefficients in the series expansion around the simple singularities obey a four-term recursion relation. From the various local solutions one can derive various important identities.
In section \ref{AsymptoticsGHEr} making use of the fact that the radial GHE in KN spacetime has an irregular singularity at $\infty$ we derive the asymptotic solutions near infinity i.e. far away from the black hole. In section \ref{CONNEIRRREGGLOBAL} following the work in \cite{RSCHAFKEDSCHMIDT} we investigate the global solution associated with the connection problem between the regular singularity at $1$ and the irregular one at $\infty$ for the GHE which after appropriate transformations of the dependent and the independent variables transforms to a connection relation in which the connection coefficients can be computed. Such connection coefficients involve besides the index parameters and the parameter $\alpha$ of the GHE the expansion coefficients in the transformed local solution. Again these coefficients are proved to obey a four-term recursion relation. This is of importance for constructing a global solution relating the event horizon with $\infty$ in the case of the KN black hole. In section \ref{Gargantua} we derive the near horizon limit solution for the KN black hole using the local solution around the regular singularity associated with the event horizon.
The theme of \textsection \ref{statheradiaxlambda} is the determination of the separation constant $\lambda$ in KN-spacetime that also appears nontrivially in the radial equation (\ref{GHERADIALKN}). Following the approach in \cite{HSchmidt} we write the angular equation as an eigenvalue matrix equation. Using functional analysis techniques and linear operator perturbation theory \cite{KatoLPT} we prove that the $\lambda$-eigenvalues  depend holomorphically on the two physical parameters $\nu,\xi$ defined in \textsection \ref{statheradiaxlambda} and their partial derivatives with respect to them are bounded by trigonometric functions and therefore by $1$ (see eqn.(\ref{estimatesgonia}))\footnote{An essential ingredient in obtaining these estimates is the computation of the operator norm of a matrix. }.
The eigenvalues obey a particular partial differential equation (\ref{eigenvalueeqn}). In  \textsection \ref{pdeeigenl} we prove using Charpit's method that this differential equation reduces to the third Painlev\'{e} transcendent $\rm{P_{III}}$, a  nonlinear ordinary differential equation-see Eqn.(\ref{Painlevetria})\footnote{The mathematical importance of the six Painlev\'{e} transcendents stems from the fact that they belong to the class of ordinary differential equations of the form $\frac{{\rm d}^2 y}{{\rm d}x^2}=F\left(\frac{{\rm d}y}{{\rm d}x},y,x\right)$ where $F$ is a rational function of ${\rm d}y/{\rm d}x$ and $y$, and an analytic function of $x$, which have the property that their solutions are free from $movable\; critical\; points$ \cite{PPainleve},\cite{Garnier}.}.
In \textsection \ref{RicattiPIIIRational} we derive special closed form solutions for the $\rm{P_{III}}$ the angular eigenvalues obey, for particular values of the parameters in terms of Bessel functions. In these special cases, the eigenvalues of the angular KN differential equation are expressible in closed analytic form in terms of Bessel functions.
In \textsection \ref{PIIIasymptotic} we investigate a novel approach in which we derive the asymptotic solutions of Painlev\'{e} $\rm{P_{III}}$ in terms of Jacobian elliptic functions. Therefore, the third Painlev\'{e} transcendent is asymptotically related to the Jacobian elliptic function. This is analogous to the scattering theory of ordinary quantum mechanics in which the Bessel functions have an asymptotic expansion in terms of trigonometric functions.
The angular eigenvalues themselves \footnote{An additional physical significance of the knowledge of the  angular eigenvalues for a rotating charged black hole is that they can be used in the calculations of emission from such curved backgrounds. An analysis for the emission rates from the Reissner-Nordstr\"{o}m black hole can be found in \cite{Donzwei}.   }  in this asymptotic elliptic limit of the transcendent   $\rm{P_{III}}$ are expressed in terms of a reduced form of the Jacobian elliptic functions ${\rm sn},{\rm dn},$ and ${\rm cn}$.
In \textsection \ref{boundfermionicstates} we prove that no fermionic bound states with $\omega^2<\mu^2$, where $\omega$ is the energy of the fermion and $\mu$ its mass, exist in the nonextreme Kerr-Newman geometry. We achieve that by using the fundamental four-term recursion relation the coefficients in the power series expansion of a closed form analytic solution of the radial GHE in the KN geometry satisfy.

In Appendix A we discuss the connection problem for the regular singularities of the GHE.
In Appendix B we explore the representation of the Painlev\'{e} $\rm{P_{III}}$ nonlinear ordinary differential equation (ODE) as the compatibility condition for a Lax pair of first order linear systems as was introduced in the works \cite{Flaschka} and \cite{Jimbo1} \footnote{In this approach, the nonlinear Painlev\'{e} $\rm{P_{III}}$ equation is written as an integrability condition of a linear system. This framework stems from the discovery of Lax of a general principle for associating nonlinear equations with linear operators so that the eigenvalues of the linear operator are integrals of the nonlinear equations \cite{LaxPeter}.}. This isomonodromic deformation method in the theory of Painlev\'{e} appears suitable for both integrating the Painlev\'{e} transcendents as well as for studying their asymptotic behaviour.

Taking into account the contribution from the cosmological
constant $\Lambda,$ the generalisation of the Kerr-Newman solution
is described by the Kerr-Newman de Sitter $($KNdS$)$ metric
element which in Boyer-Lindquist (BL) coordinates is given by
\cite{Stuchlik1},\cite{BCAR},\cite{GrifPod},\cite{ZdeStu} (in units where $G=1$ and $c=1$):
\begin{align}
\mathrm{d}s^{2}  & =\frac{\Delta_{r}^{KN}}{\Xi^{2}\rho^{2}}(\mathrm{d}%
t-a\sin^{2}\theta\mathrm{d}\phi)^{2}-\frac{\rho^{2}}{\Delta_{r}^{KN}%
}\mathrm{d}r^{2}-\frac{\rho^{2}}{\Delta_{\theta}}\mathrm{d}\theta
^{2}\nonumber \\ &-\frac{\Delta_{\theta}\sin^{2}\theta}{\Xi^{2}\rho^{2}}(a\mathrm{d}%
t-(r^{2}+a^{2})\mathrm{d}\phi)^{2}%
\label{KNADSelement}
\end{align}%
\begin{equation}
\Delta_{\theta}:=1+\frac{a^{2}\Lambda}{3}\cos^{2}\theta,
\;\Xi:=1+\frac {a^{2}\Lambda}{3},
\end{equation}

\begin{equation}
\Delta_{r}^{KN}:=\left(  1-\frac{\Lambda}{3}r^{2}\right)  \left(  r^{2}
+a^{2}\right)  -2Mr+e^{2},
\label{DiscrimiL}
\end{equation}

\begin{equation}
\rho^{2}=r^{2}+a^{2}\cos^{2}\theta,
\end{equation}
where $a,M,e,$ denote the Kerr parameter, mass and electric charge
of the black hole, respectively.
The KN(a)dS metric is the most general exact stationary black hole solution of the Einstein-Maxwell system of differential equations.
This
is accompanied by a non-zero electromagnetic field
$F=\mathrm{d}A,$ where the vector potential is
\cite{ZST},\cite{GrifPod}:
\begin{equation}
A=-\frac{er}{\Xi(r^{2}+a^{2}\cos^{2}\theta)}(\mathrm{d}t-a\sin^{2}\theta
\mathrm{d}\phi).
\end{equation}

\section{Null tetrad and the Dirac equation in Kerr-Newman-de Sitter black hole spacetime }\label{LambdaKNTetrad}

The Kerr-Newman-de Sitter geometry (as was the case with the Kerr and Kerr-Newman geometries) can be described in terms of a local Newman-Penrose null tetrad frame that is adapted to the principal null geodesics, i.e. the tetrad coincides with the two principal null directions of the Weyl tensor $C_{\mu\nu\rho\varepsilon}$. In this generalised Kinnersley frame, the null tetrad is constructed directly from the tangent vectors of the principal null geodesics:
\begin{equation}
\dot{t}:=\frac{{\rm d}t}{{\rm d}\lambda}=\frac{\Xi^2 (r^2+a^2)}{\Delta_r^{KN}}E,\;\;\dot{r}=\pm \Xi E,\;\;
\dot{\theta}=0,\;\;\dot{\phi}=\frac{a\Xi^2}{\Delta_r^{KN}}E,
\end{equation}
where the dot denotes differentiation with respect to the affine parameter $\lambda$ and $E$ denotes a constant.

Thus the  generalised Kinnersley null tetrad in the Kerr-Newman-de Sitter spacetime is given by
\begin{align}
l^{\mu}&=\left[\frac{(r^2+a^2)\Xi}{\Delta_r^{KN}},1,0,\frac{a\Xi}{\Delta_r^{KN}}\right],\;
n^{\mu}=\left[\frac{\Xi(r^2+a^2)}{2\rho^2},-\frac{\Delta_r^{KN}}{2\rho^2},0,\frac{a\Xi}{2\rho^2}\right] \nonumber \\
m^{\mu}&=\frac{1}{(r+ia\cos\theta)\sqrt{2\Delta_{\theta}}}\left[ia\Xi\sin\theta,0,\Delta_{\theta},\frac{i\Xi}{\sin\theta}\right]\nonumber \\
\overline{m}^{\mu}&=\frac{-1}{(r-ia\cos\theta)\sqrt{2\Delta_{\theta}}}\left[ia\Xi\sin\theta,0,-\Delta_{\theta},\frac{i\Xi}{\sin\theta}\right]
\label{GLAMBDAKNNULLTETRAD}
\end{align}

\subsubsection{Calculation of the Ricci-rotation coefficients in the Kerr-Newman-de Sitter spacetime}
Using the generalised Kinnersley tetrad we derived in section \ref{LambdaKNTetrad}, we computed the Ricci-rotation coefficients via  the formula for the $\lambda$-symbols given by Chandrasekhar \cite{Chandrasekhar}:
\begin{equation}
\lambda_{(a)(b)(c)}=e_{(b)i,j}[e_{(a)}^{\;i}e_{(c)}^{\;j}-e_{(a)}^{\;j}e_{(c)}^{\;i}].
\end{equation}
This formula has the advantage that one has to calculate ordinary derivatives of the dual co-tetrad.
Computation of the dual co-tetrad in the Kerr-Newman-de Sitter spacetime yields:
\begin{align*}
{\bf l}&=\frac{1}{\Xi}{\rm d}t-\frac{\rho^2}{\Delta_{r}^{KN}}{\rm d}r-\frac{a\sin^2\theta}{\Xi}{\rm d}\phi \\
{\bf n}&=\frac{\Delta_r^{KN}}{2\rho^2\Xi}{\rm d}t+\frac{1}{2}{\rm d}r-\frac{a\sin^2\theta\Delta_r^{KN}}{2\rho^2\Xi}{\rm d}\phi, \\
{\bf m}&=\frac{\Delta{\theta}a\sin\theta i}{\Xi \sqrt{2\Delta_{\theta}}(r+ia\cos\theta)}{\rm d}t-
\frac{\rho^2}{(r+ia\cos\theta)\sqrt{2\Delta_{\theta}}}{\rm d}\theta-\frac{\Delta_{\theta}i\sin\theta(r^2+a^2)}{\Xi(r+ia\cos\theta)\sqrt{2\Delta_{\theta}}}{\rm d}\phi,\\
{\bf \overline{m}}&=\frac{-ia\sin\theta\Delta_{\theta}}{\Xi(r-ia\cos\theta)\sqrt{2\Delta_{\theta}}}{\rm d}t
-\frac{\rho^2}{(r-ia\cos\theta)\sqrt{2\Delta_{\theta}}}{\rm d}\theta+\frac{i\Delta_{\theta}\sin\theta(r^2+a^2)}{\Xi(r-ia\cos\theta)\sqrt{2\Delta_{\theta}}}{\rm d}\phi.
\end{align*}

The non-vanishing $\lambda$-symbols are computed to be:
\begin{align}
\lambda_{213}&=-\sqrt{2\Delta_{\theta}}\frac{a^2\sin\theta\cos\theta}{\rho^2\overline{\rho}},\;\lambda_{324}=
-\frac{ia\cos\theta\Delta_r^{KN}}{\rho^4},\\
\lambda_{243}&=\frac{-\Delta_r^{KN}}{2\rho^2\overline{\rho}},\;\lambda_{234}=-\frac{\Delta_r^{KN}}{2\rho^2(r-ia\cos\theta)},\\
\lambda_{134}&=\frac{1}{r-ia\cos\theta}=\frac{1}{\overline{\rho}^*},\;\lambda_{314}=\frac{-2ia\cos\theta}{\rho^2},\\
\lambda_{122}&=-\frac{1}{2}\frac{{\rm d}\Delta_r^{KN}}{{\rm d}r}\frac{1}{\rho^2}+r\frac{\Delta_r^{KN}}{\rho^4},
\;\lambda_{132}=\frac{\sqrt{2\Delta_{\theta}}ira\sin\theta}{\overline{\rho}\rho^2},\\
\lambda_{334}&=\frac{1}{\sin\theta\sqrt{2}\overline{\rho}}\frac{{\rm d}(\sqrt{\Delta_{\theta}}\sin\theta)}{{\rm d}\theta}+\frac{i\sqrt{\Delta_{\theta}}a\sin\theta}{\sqrt{2}\overline{\rho}^2},\;\lambda_{241}=\frac{ira\sqrt{2\Delta_{\theta}}\sin\theta}{\rho^2\overline{\rho}^*}
,\\
\lambda_{412}&=\frac{\sqrt{2\Delta_{\theta}}a^2 \sin\theta\cos\theta}{\rho^2 (r-ia\cos\theta)},\;
\lambda_{443}=\frac{{\rm d}}{{\rm d}\theta}(\sqrt{\Delta_{\theta}}\sin\theta)\frac{1}{\sqrt{2}\overline{\rho}^*\sin\theta}-\frac{ia\sin\theta\sqrt{\Delta_{\theta}}}{\sqrt{2}\overline{\rho}^{*2}},
\label{lambdasymbolsL}
\end{align}
where $\overline{\rho}=r+ia\cos\theta,\;\rho^2=\overline{\rho}\;\overline{\rho}^*$.
The Ricci rotation coefficients $\gamma_{(a)(b)(c)}$ are expressed through the $\lambda$-coefficients as follows:
\begin{equation}
\gamma_{(a)(b)(c)}=\frac{1}{2}[\lambda_{(a)(b)(c)}+\lambda_{(c)(a)(b)}-\lambda_{(b)(c)(a)}]
\end{equation}
Thus a calculation through the $\lambda$-symbols yields the following non-vanishing Ricci coefficients for the Kerr-Newman-de Sitter spacetime:
\begin{align}
\pi &=\gamma_{241}=\frac{1}{2}\frac{ia\sin\theta\sqrt{2\Delta_{\theta}}}{(\overline{\rho}^*)^2},\;\beta=
\frac{1}{2}(\gamma_{213}+\gamma_{343})=\frac{1}{2\sin\theta\sqrt{2}\overline{\rho}}\frac{\rm d}{{\rm d}\theta}(\sqrt{\Delta_{\theta}}\sin\theta)\label{lambdaricci}\\
\gamma &=\frac{1}{2}(\gamma_{212}+\gamma_{342})=\frac{1}{4\rho^2}\frac{{\rm d}\Delta_r^{KN}}{{\rm d}r}-
\frac{1}{2\rho^2 \overline{\rho}^*}\Delta_r^{KN},\; \varrho=\gamma_{314}=-\frac{1}{\overline{\rho}^*},\\
\mu&=\gamma_{243}=-\frac{\Delta_r^{KN}}{2\rho^2\overline{\rho}^*},\;\tau=\gamma_{312}=\frac{-i\sqrt{2\Delta_{\theta}}a\sin\theta}{2\rho^2},\;\alpha:=\frac{1}{2}(\gamma_{214}+\gamma_{344})=\pi-\beta^{*}.
\label{cosmoconsricci}
\end{align}
we obtain the following 2-spinor form of the Dirac equation
\begin{align}
(\nabla_{A\dot{B}}+iqA_{A\dot{B}})P^A+i\mu_*\overline{Q}_{\dot{B}}&=0\label{1stDIRAC},\\
(\nabla_{A\dot{B}}-iqA_{A\dot{B}})Q^A+i\mu_*\overline{P}_{\dot{B}}&=0,
\end{align}
where $\nabla_{A\dot{B}}=\sigma^{\mu}_{\;A\dot{B}}\nabla_{\mu}$, $q$ is the charge or the coupling constant of the massive Dirac particle to the vector field and $\mu_*$ is the particle mass. Thus equivalently the 2-spinor form of general relativistic Dirac's equation is the following:
\begin{align}
\sigma^{\mu}_{\;A\dot{B}}P^A_{;\mu}+i\mu_*\overline{Q}^{\dot{C}}\varepsilon_{\dot{C}\dot{B}}&=0,\\
\sigma^{\mu}_{\;A\dot{B}}Q^A_{;\mu}+i\mu_*\overline{P}^{\dot{C}}\varepsilon_{\dot{C}\dot{B}}&=0.
\end{align}
The components for $\dot{B}=0,1$ of eqn.(\ref{1stDIRAC}) lead to the following general relativistic Dirac equations in the Newman-Penrose formalism for the Kerr-Newman de Sitter black hole spacetime:
\begin{align}
(D^{\prime}-\gamma+\mu+iqn^{\mu}A_{\mu})P^{(1)}+(\delta-\tau+\beta+iq m^{\mu}A_{\mu})P^{(0)}&=-i\mu_{*}\overline{Q}^{(\dot{0})}\label{PaulDKNdS},\\
(-D+\varrho-\varepsilon-iql^{\mu}A_{\mu})P^{(0)}+(-\delta^{\prime}+\alpha-\pi-iq\overline{m}^{\mu}A_{\mu})P^{(1)}&=-i\mu_{*}\overline{Q}^{(\dot{1})}
\label{DiracKNdSspacetime}
\end{align}
where assuming that the azimuthal and time-dependence of the fields will be of the form $e^{i(m\phi-\omega t)}$ we calculate the directional derivatives to be
\begin{align}
D&=l^{\mu}\partial_{\mu}=\frac{\partial }{\partial r}+\frac{i\Xi}{\Delta_r^{KN}}K\equiv {\mathcal D}_0,\\
D^{\prime}\equiv\Delta&=n^{\mu}\partial_{\mu}=-\frac{\Delta_r^{KN}}{2\rho^2}\left(\frac{\partial}{\partial r}-\frac{i\Xi K}{\Delta_r^{KN}}\right)\equiv-\frac{\Delta_r^{KN}}{2\rho^2}{\mathcal D}_0^{\dagger}\\
\delta&=m^{\mu}\partial_{\mu}=\frac{\sqrt{\Delta_{\theta}}}{\sqrt{2}\overline{\rho}}\left[\frac{\partial }{\partial \theta}+\frac{\Xi H}{\Delta_{\theta}}\right]\equiv \frac{\sqrt{\Delta_{\theta}}}{\sqrt{2}\overline{\rho}}{\mathcal L}_0 ,\\
\delta^{\prime}\equiv\delta^{*}&=\overline{m}^{\mu}\partial_{\mu}=\frac{\sqrt{\Delta_{\theta}}}{\sqrt{2}\overline{\rho}^*}\left[\frac{\partial}{\partial \theta}-\frac{\Xi H}{\Delta_{\theta}}\right]\equiv\frac{\sqrt{\Delta_{\theta}}}{\sqrt{2}\overline{\rho}^*}{\mathcal L}^{\dagger}_0,
\end{align}
where
\begin{equation}
K:=ma-\omega(r^2+a^2),\;\;H:=\omega a \sin\theta-\frac{m}{\sin\theta}
\end{equation}

\subsection{Separation of the Dirac equation in the Kerr-Newman-de Sitter spacetime}
Applying the separation ansatz:
\begin{align}
\overline{Q}^{(\dot{0})}&=-\frac{e^{-i\omega t}e^{im\phi}S^{(+)}(\theta)R^{(-)}(r)}{\sqrt{2}(r+ia\cos\theta)},\;\;\;\overline{Q}^{(\dot{1})}=
\frac{e^{-i\omega t}e^{im\phi}S^{(-)}(\theta)R^{(+)}(r)}{\sqrt{\Delta_r^{KN}}}\\
P^{(0)}&=\frac{e^{-i\omega t}e^{i m\phi}S^{(-)}(\theta)R^{(-)}(r)}{\sqrt{2}(r-ia\cos\theta)},\;\;\;
P^{(1)}=\frac{e^{-i\omega t}e^{im\phi}S^{(+)}(\theta)R^{(+)}(r)}{\sqrt{\Delta_r^{KN}}}.
\label{ansatzseperation}
\end{align}
we obtain the following ordinary differential equations for the radial and angular polar parts:
\begin{align}
&\sqrt{\Delta_r^{KN}}\left\{\frac{{\rm d}R^{(-)}(r)}{{\rm d}r}+\left[\frac{\Xi i(ma-\omega(r^2+a^2))}{\Delta_r^{KN}}\right]R^{(-)}(r)-\frac{iqer R^{(-)}(r)}{\Delta_r^{KN}}\right\}=(\lambda+i\mu r)R^{(+)}(r)\label{ALGMEINEDIR1},\\
&\sqrt{\Delta_r^{KN}}\frac{{\rm d}R^{(+)}(r)}{{\rm d}r}-\frac{i\Xi(ma-\omega(r^2+a^2))R^{(+)}(r)}{\sqrt{\Delta_r^{KN}}}+\frac{ieqr R^{(+)}(r)}{\sqrt{\Delta_r^{KN}}}=(\lambda-i\mu r)R^{(-)}(r)\label{ALgemeinDIR2},\\
&\frac{\Delta_{\theta}}{\sqrt{\Delta_{\theta}}}\frac{{\rm d}S^{(+)}(\theta)}{{\rm d}\theta}+
\frac{\Xi}{\sqrt{\Delta_{\theta}}}\left[-\omega a \sin\theta+\frac{m}{\sin\theta}\right]S^{(+)}(\theta)+
\frac{1}{2\sin\theta}\frac{{\rm d}}{{\rm d}\theta}(\sqrt{\Delta_{\theta}}\sin\theta)S^{(+)}(\theta)\nonumber \\
&=(-\lambda+\mu a \cos\theta)S^{(-)}(\theta),\\
&\frac{\Delta_{\theta}}{\sqrt{\Delta_{\theta}}}\frac{{\rm d}S^{(-)}(\theta)}{{\rm d}\theta}+
\left(\omega a \sin\theta-\frac{m}{\sin\theta}\right)\frac{\Xi}{\sqrt{\Delta_{\theta}}}S^{(-)}(\theta)
+\frac{1}{2\sin\theta}\frac{{\rm d}}{{\rm d}\sin\theta}(\sqrt{\Delta_{\theta}}\sin\theta)S^{(-)}(\theta)\nonumber \\
&=(\lambda+\mu a \cos\theta)S^{(+)}(\theta),\label{ALGEMEINEDIR4}
\end{align}
where $\lambda$ is a separation constant and $\mu_*=\mu/\sqrt{2}$.
The angular equations can be combined into the compact form:
\begin{equation}
\sqrt{\Delta_{\theta}}\mathcal{L}^{\dagger}_{1/2}\sqrt{\Delta_{\theta}}\mathcal{L}_{1/2}S^{(-)}(\theta)+
\frac{\Delta_{\theta}\mu a \sin\theta}{\lambda+\mu a \cos\theta}\mathcal{L}_{1/2}S^{(-)}(\theta)=
(\mu^2a^2\cos^2 \theta+\lambda^2)S^{(-)}(\theta).
\label{gwniaKNdS}
\end{equation}
The operators $\mathcal{L}^{\dagger}_{1/2},\mathcal{L}_{1/2}$ are defined as follows:
\begin{equation}
\mathcal{L}^{\dagger}_{1/2}=\frac{{\rm d}}{{\rm d}\theta}-\frac{\Xi H}{\Delta_{\theta}}+
\frac{1}{2\sin\theta}\frac{1}{\sqrt{\Delta_{\theta}}}\frac{\rm d}{{\rm d}\theta}(\sqrt{\Delta_{\theta}}\sin\theta),
\end{equation}
\begin{equation}
\mathcal{L}_{1/2}=\frac{{\rm d}}{{\rm d}\theta}+\frac{\Xi H}{\Delta_{\theta}}+
\frac{1}{2\sin\theta}\frac{1}{\sqrt{\Delta_{\theta}}}\frac{\rm d}{{\rm d}\theta}(\sqrt{\Delta_{\theta}}\sin\theta).
\end{equation}
Equivalently (\ref{gwniaKNdS}) can be written explicitly as follows:
\begin{align}
&\frac{\Delta_{\theta}}{\sqrt{\Delta_{\theta}}}\frac{[\mu a \omega a \sin^2\theta \Xi/\sqrt{\Delta_{\theta}}S^{(-)}(\theta)-\mu a m\frac{\Xi}{\sqrt{\Delta_{\theta}}}S^{(-)}(\theta)]}{\lambda+\mu a \cos\theta} \nonumber \\
&+\frac{\Delta_{\theta}}{\sqrt{\Delta_{\theta}}}\frac{\frac{\mu a}{2}\frac{\rm d}{{\rm d}\theta}(\sqrt{\Delta_{\theta}}\sin\theta)S^{(-)}}{\lambda+\mu a \cos\theta}+
\frac{\mu a \sin\theta\Delta_{\theta}}{\lambda+\mu a \cos\theta}\frac{{\rm d}S^{(-)}(\theta)}{{\rm d}\theta} \nonumber \\
&+\frac{\Delta_{\theta}}{\sqrt{\Delta_{\theta}}}\Biggl\{\sqrt{\Delta_{\theta}}\frac{{\rm d}^2 S^{(-)}(\theta)}{{\rm d}\theta^2}+\frac{{\rm d}}{{\rm d}\theta}\sqrt{\Delta_{\theta}}\frac{{\rm d}S^{(-)}(\theta)}{{\rm d}\theta}+\left(\omega a \cos\theta+\frac{m\cos\theta}{\sin^2\theta}\right)\frac{\Xi}{\sqrt{\Delta_{\theta}}}S^{(-)}(\theta)\nonumber \\
&+\Xi\left(\omega a \sin\theta-\frac{m}{\sin\theta}\right)\Delta_{\theta}^{-3/2}\frac{a^2\Lambda}{3}\cos\theta\sin\theta S^{(-)}(\theta)\nonumber \\
&+\frac{1}{2\sin\theta}\frac{\rm d}{{\rm d}\theta}(\sqrt{\Delta_{\theta}}\sin\theta)\frac{{\rm d}S^{(-)}(\theta)}{{\rm d}\theta}+\frac{1}{2\sin\theta}\frac{{\rm d}^2}{{\rm d}\theta^2}(\sqrt{\Delta_{\theta}}\sin\theta)S^{(-)}(\theta)\nonumber \\
&-\frac{1}{2}\frac{\cos\theta}{\sin^2\theta}\frac{\rm d}{{\rm d}\theta}(\sqrt{\Delta_{\theta}}\sin\theta)S^{(-)}(\theta)\Biggr\}\nonumber \\
&+\frac{\Xi}{\sqrt{\Delta_{\theta}}}\left[-\omega a \sin\theta+\frac{m}{\sin\theta}\right]\left[\omega a \sin\theta-\frac{m}{\sin\theta}\right]
\frac{\Xi}{\sqrt{\Delta_{\theta}}}S^{(-)}(\theta)\nonumber \\
&+\frac{1}{2\sin\theta}\frac{\rm d}{{\rm d}\theta}(\sqrt{\Delta_{\theta}}\sin\theta)\frac{\Delta_{\theta}}{\sqrt{\Delta_{\theta}}}\frac{{\rm d}S^{(-)}(\theta)}{{\rm d}\theta}+
\left[\frac{1}{2\sin\theta}\frac{\rm d}{{\rm d}\theta}(\sqrt{\Delta_{\theta}}\sin\theta)\right]^2 S^{(-)}(\theta)\nonumber \\
&=(-\lambda^2+\mu^2 a^2\cos^2\theta)S^{(-)}(\theta).
\label{gwniaKNdSGHEUN}
\end{align}
Taking the $\Lambda=0$ limit of (\ref{gwniaKNdSGHEUN}) yields the differential equation:
\begin{align}
&\left[\frac{\mu a^2 \omega \sin^2\theta S^{(-)}(\theta)-\mu a m S^{(-)}(\theta) }{\lambda+\mu a \cos\theta}\right]+\frac{\frac{\mu a}{2}\cos\theta S^{(-)}(\theta)}{\lambda +\mu a \cos\theta}+
\frac{\mu a \sin\theta}{\lambda+\mu a \cos\theta}\frac{{\rm d}S^{(-)}(\theta)}{{\rm d}\theta} \nonumber \\
&\Biggl\{\frac{{\rm d}^2 S^{(-)}(\theta)}{{\rm d}\theta^2}+\left(\omega a \cos\theta+
\frac{m\cos\theta}{\sin^2\theta}\right)S^{(-)}(\theta)\nonumber \\
&+\frac{\cos\theta}{2\sin\theta}\frac{{\rm d}S^{(-)}(\theta)}{{\rm d} \theta}-\frac{1}{2}S^{(-)}(\theta)\nonumber \\
&-\frac{1}{2}\frac{\cos^2\theta}{\sin^2\theta}S^{(-)}(\theta)\Biggr\}\nonumber \\
&+\left[-\omega a \sin\theta +\frac{m}{\sin\theta}\right]\left[\omega a \sin\theta-\frac{m}{\sin\theta}\right]S^{(-)}(\theta)\nonumber \\
&+\frac{\cos\theta}{2\sin\theta}\frac{{\rm d}S^{(-)}(\theta)}{{\rm d}\theta}+\left[\frac{1}{2\sin\theta}\frac{\rm d}{{\rm d}\theta}\sin\theta\right]^2 S^{(-)}(\theta)\nonumber \\
&=(-\lambda^2+\mu^2 a^2 \cos^2 \theta)S^{(-)}(\theta),
\label{LEEPAGE}
\end{align}
which agrees with the results in \cite{page}.

\section{Transforming the radial equation in the KNdS spacetime into a generalisation of Heun's equation}\label{fundamentalKNdS}
Combining the radial equations (\ref{ALGMEINEDIR1})-(\ref{ALgemeinDIR2}) we obtain:
\begin{align}
\Biggl[&\Delta_r^{KN}\frac{{\rm d}^2}{{\rm d}r^2}+\left(\frac{1}{2}\frac{{\rm d}\Delta_r^{KN}}{{\rm d}r}-i\mu\frac{\Delta_r^{KN}}{\lambda+i\mu r}\right)\frac{{\rm d}}{{\rm d}r}\nonumber \\
&+\left(K^2+\frac{iK}{2}\frac{{\rm d}\Delta_r^{KN}}{{\rm d}r}\right)\frac{1}{\Delta_r^{KN}}-\frac{\mu K}{\lambda+i\mu r}\nonumber \\
&-2i\omega r \Xi-ieq-\mu^2 r^2-\lambda^2\Biggr]R^{(-)}(r)=0\label{IntervenKNdSrad},
\end{align}
where $K(r):=\Xi[(r^2+a^2)\omega-ma]+eqr\equiv \mathcal{K}(r)+eqr$.
Applying  to (\ref{IntervenKNdSrad}) the independent variable transformation:
\begin{equation}
r=r_1+z(r_2-r_1),
\end{equation}
so that the quartic polynomial that determines the horizons of the Kerr-Newman-de Sitter black hole factorises as follows \footnote{The quantity $\Delta_r^{KN}$ in terms of the radii of the event and Cauchy horizons $r_+,r_{-}$ and the cosmological horizon $r_{\Lambda}^+$ for positive cosmological constant is written as:
$\Delta_r^{KN}=-\frac{\Lambda}{3}(r-r_+)(r-r_{-})(r-r_{\Lambda}^{+})(r-r_{\Lambda}^{-})$.}:
\begin{equation}
\Delta_r^{KN}=-\frac{\Lambda}{3}(r_2-r_1)^4z(z-1)(z-z_3)(z-z_4),
\end{equation}
with the notation:
\begin{align}
r_1&\equiv r_{+},\;\;r_2\equiv r_{-},\;\;r_3\equiv r_{\Lambda}^+,\;\;r_4 \equiv r_{\Lambda}^-,\;\;r_5 \equiv i\lambda/\mu,\\
z_3&\equiv \frac{r_3-r_{+}}{r_{-}-r_{+}},\;\;z_4\equiv \frac{r_4-r_{+}}{r_{-}-r_{+}},\;\;z_5\equiv \frac{r_5-r_{+}}{r_{-}-r_{+}},
\end{align}
yields the equation:
\begin{align}
&\frac{{\rm d}^2R_{-}(z)}{{\rm d}z^2}+\left[\frac{1}{2}\left(\frac{1}{z}+\frac{1}{z-1}+\frac{1}{z-z_3}+\frac{1}{z-z_4}\right)-\frac{1}{z-z_5}\right]
\frac{{\rm d}R_{-}(z)}{{\rm d}z}\nonumber \\
&+\left[\frac{Q_1}{z^2}+\frac{Q_2}{(z-1)^2}+\frac{Q_3}{(z-z_3)^2}+\frac{Q_4}{(z-z_4)^2}+\frac{L_1}{z}+
\frac{L_2}{z-1}+\frac{L_3}{z-z_3}+\frac{L_4}{z-z_4}+\frac{L_5}{z-z_5}\right]R_{-}(z)\nonumber \\
&=0,
\label{PreGeneHeunKNdS}
\end{align}
with
\begin{align}
Q_1&=\frac{3^2 K^2(r_{+})}{\Lambda^2(r_{-}-r_{+})^6 z_3^2 z_4^2}+\frac{3i}{2\Lambda}\frac{(\mathcal{K}(r_{+})+eqr_{+})}{(r_{-}-r_{+})(r_3-r_{+})(r_4-r_{+})}\nonumber \\
&=\frac{3^2 (\mathcal{K}(r_{+})+eqr_{+})^2}{\Lambda^2(r_{-}-r_{+})^6 z_3^2 z_4^2}+\frac{3i}{2\Lambda}\frac{(\mathcal{K}(r_{+})+eqr_{+})}{(r_{-}-r_{+})(r_3-r_{+})(r_4-r_{+})},\\
Q_2&=\frac{3^2 K^2(r_{-})}{\Lambda^2(r_{-}-r_{+})^6 (z_3-1)^2 (z_4-1)^2}+\left(\frac{-3i}{2\Lambda}\right)\frac{\mathcal{K}(r_{-})+eqr_{-}}{(r_{-}-r_{+})(r_3-r_{-})(r_4-r_{-})}\nonumber \\
&=\frac{3^2 (\mathcal{K}(r_{-})+eqr_{-})^2}{\Lambda^2(r_{-}-r_{+})^6 (z_3-1)^2 (z_4-1)^2}+\left(\frac{-3i}{2\Lambda}\right)\frac{\mathcal{K}(r_{-})+eqr_{-}}{(r_{-}-r_{+})(r_3-r_{-})(r_4-r_{-})},\\
Q_3&=\frac{9\{\mathcal{K}(r_{+})-\omega\Xi(r_3^2-r_{+}^2)\}^2+18eqr_3\mathcal{K}(r_3)+9e^2q^2r_3^2}{
\Lambda^2(r_{-}-r_{+})^6z_3^2(z_3-1)^2(z_3-z_4)^2}+\left(\frac{-3i}{2\Lambda}\right)\frac{\mathcal{K}(r_3)+eqr_3}{(r_3-r_{+})(r_3-r_{-})(r_3-r_4)},\\
Q_4&=\frac{9\{\mathcal{K}(r_{+})-\omega\Xi(r_4^2-r_{+}^2)\}^2+18eqr_4\mathcal{K}(r_4)+9e^2q^2r_4^2}{
\Lambda^2(r_{-}-r_{+})^6z_4^2(z_4-1)^2(z_3-z_4)^2}+\left(\frac{3i}{2\Lambda}\right)\frac{\mathcal{K}(r_4)+eqr_4}{(r_4-r_{+})(r_4-r_{-})(r_3-r_4)},
\end{align}
and
\begin{align}
L_1&=\left(\frac{3i}{2\Lambda}\right)\left(\frac{2\Xi\omega r_{+}+eq}{(r_4-r_{+})(r_3-r_{+})}\right)
+\left(-\frac{3\lambda^2}{\Lambda(r_4-r_{+})(r_3-r_{+})}\right)\nonumber \\
&-\frac{3i}{\Lambda}\frac{\mathcal{K}(r_{+})+eqr_{+}}{(r_4-r_{+})(r_3-r_{+})(r_5-r_{+})}
+\left(\frac{-3\mu^2 r_{+}^2}{\Lambda(r_4-r_{+})(r_3-r_{+})}\right)\nonumber \\
&+\left(\frac{-3ieq}{\Lambda (r_4-r_{+})(r_3-r_{+}) }\right)+\left(\frac{-6i\Xi\omega r_{+}}{\Lambda (r_4-r_{+})(r_3-r_{+})}\right)\nonumber\\
&+\Biggl\{18(z_3+z_4)[\mathcal{K}(r_{+})+eqr_{+}]^2+36\Xi \omega r_{-}r_{+}z_3z_4\mathcal{K}(r_{+})+18\Xi^2\omega^2z_3z_4(a^4-r_{+}^4)\nonumber \\
&+18 \Xi^2 z_3z_4(a^2m^2-2am \omega a^2)+18\Xi eqz_3z_4(r_{+}+r_{-})[\omega a^2-am]\nonumber \\
& -18\Xi e q \omega r_{-}r_{+}(r_{+}^2-3r_{+})z_3z_4+18e^2 q^2 r_{-}r_{+}z_3z_4\Biggr\}\frac{1}{\Lambda^2(r_{-}-r_{+})^6z_3^3z_4^3}\nonumber \\
&+\left(\frac{-3i}{2\Lambda}\right)\frac{\mathcal{K}(r_{+})+eqr_{+}}{(r_3-r_{+})^2 (r_4-r_{+})}+
\left(\frac{-3i}{2\Lambda}\right)\frac{\mathcal{K}(r_{+})+eqr_{+}}{(r_4-r_{+})^2(r_3-r_{+})}\nonumber\\
&+\frac{3i}{2\Lambda}\frac{(r_{-}-r_{+})}{(r_3-r_{+})^2(r_4-r_{+})^2}\left\{(\mathcal{K}(r_{-})+eqr_{-})z_3z_4-
\Xi\omega z_3 z_4 (r_{-}-r_{+})^2+[\mathcal{K}(r_{+})+eqr_{+}](z_3+z_4)\right\}\nonumber \\
&+\left(\frac{-3i}{2\Lambda}\right)\frac{\mathcal{K}(r_{+})+eqr_{+}}{(r_{-}-r_{+})(r_3-r_{+})(r_4-r_{+})},\\
L_2&=\left(\frac{-3i}{2\Lambda}\right)\left(\frac{2\Xi\omega r_{-}+eq}{(r_3-r_{-})(r_4-r_{-})}\right)+\frac{3\lambda^2}{\Lambda(r_3-r_{-})(r_4-r_{-})}\nonumber \\
&+\frac{3i}{\Lambda}\frac{\mathcal{K}(r_{-})+eqr_{-}}{(r_3-r_{-})(r_4-r_{-})(r_5-r_{-})}+\frac{3\mu^2r_{-}^2}{\Lambda(r_3-r_{-})(r_4-r_{-})}\nonumber \\
&+\frac{3ieq}{\Lambda(r_3-r_{-})(r_4-r_{-})}+\frac{6i\Xi\omega r_{-}}{\Lambda(r_3-r_{-})(r_4-r_{-})}\nonumber \\
&+\frac{1}{\Lambda^2 (r_3-r_{-})^3 (r_4-r_{-})^3}\Biggl[(36\mathcal{K}(r_{-}r_{+})+18eq(r_{-}+r_{+}))(\mathcal{K}(r_{-})+eqr_{-})(z_3+z_4-1)\nonumber \\
&-18(\mathcal{K}(r_{-})+eqr_{-})^2+z_3z_4(\mathcal{K}(r_{-})+eqr_{-})(-36\Xi r_{-}r_{+}\omega-18eqr_{+})\nonumber \\
&+18\Xi\omega r_{-}^2 z_3z_4(eqr_{-}+\Xi\omega r_{-}^2)-18\Xi a^2\omega z_3z_4(\Xi\omega a^2+eqr_{-}-2am\Xi)\nonumber \\
&+18\Xi amz_3z_4(eqr_{-}-\Xi am)\Biggr]\nonumber \\
&+\left(\frac{3i}{2\Lambda}\frac{(r_{-}-r_{+})}{(r_3-r_{-})^2(r_4-r_{-})^2}\right)\Biggl\{
z_3z_4(\mathcal{K}(r_{+})+eqr_{+})-\Xi\omega z_3z_4(r_{-}-r_{+})^2+3(\mathcal{K}(r_{-})+eqr_{-})-\nonumber \\
&2\mathcal{K}(r_{-}r_{+})(z_3+z_4)+2\Xi\omega r_{-}(r_{+}-r_{-})-eq(r_{+}+r_{-})(z_3+z_4)+eq(r_{+}-r_{-})\Biggr\}\nonumber \\
&-\left(\frac{3i}{2\Lambda}\right)\frac{\mathcal{K}(r_{-})+eqr_{-}}{(r_{-}-r_{+})(r_3-r_{-})(r_4-r_{-})}+
\left(\frac{3i}{2\Lambda}\right)\frac{
\mathcal{K}(r_{-})+eqr_{-}}{(r_3-r_{-})^2(r_4-r_{-})}\nonumber \\
&+\left(\frac{3i}{2\Lambda}\right)\frac{\mathcal{K}(r_{-})+eqr_{-}}{(r_4-r_{-})^2(r_3-r_{-})},
\end{align}

\begin{align}
L_3&=\frac{1}{\Lambda^2 (r_{-}-r_{+})^6 z_3^3 (z_3-1)^3(z_3-z_4)^3}\Biggl[-72\Xi(r_{-}-r_{+})^2z_3^4\omega[\mathcal{K}(r_{+})+eqr_{+}]\nonumber \\&-4\;36\Xi \omega z_3^4(r_{-}-r_{+})^2[\omega \Xi r_{+}^2+eqr_{+}]
-18\Xi eq \omega z_3^4(r_{-}-r_{+})^2(r_{+}-r_{-})\nonumber \\
&-36\Xi^2\omega^2z_3^4r_{+}(r_{-}-r_{+})^2(r_{+}-r_{-})+108\Xi\omega z_3^2r_{+}(r_{-}-r_{+})[\mathcal{K}(r_{+})+eqr_{+}]\nonumber \\
&-108eqr_{+}z_3^2[\mathcal{K}(r_{+})+eqr_{+}]-54eqz_3^2(r_{+}-r_{-})\mathcal{K}(r_{+})-54z_3^2\mathcal{K}^2(r_{+})+54e^2q^2r_{-}r_{+}z_3^2
\nonumber \\
&-180z_3^3\Xi\omega r_{+}(r_{-}-r_{+})\{\mathcal{K}(r_{+})+eqr_{+}\}+
36z_3^3\Xi\omega(r_{-}-r_{+})^2\{\mathcal{K}(r_{+})+eqr_{+}\}\nonumber \\
&+72\Xi\omega z_3^3(r_{-}-r_{+})^2[eqr_{+}+\omega\Xi r_{+}^2]+90eqz_3^3(r_{+}-r_{-})[\mathcal{K}(r_{+})+eqr_{+}]\nonumber \\
&+18e^2q^2z_3^3(r_{-}-r_{+})^2-36z_3^4e^2q^2(r_{-}-r_{+})^2\nonumber \\
&+36 \Xi \omega z_3^3z_4(r_{-}-r_{+})^2[\mathcal{K}(r_{+})+eqr_{+}]+
36\Xi^2\omega^2z_3^3z_4(r_{-}-r_{+})^2 r_{+}(r_{+}+r_{-})\nonumber \\
&+18e^2q^2z_3^3z_4(r_{-}-r_{+})^2+18\Xi eq\omega(r_{-}-r_{+})^2[r_{-}+3r_{+}]z^3_3z_4\nonumber \\
&-18\Xi^2\omega^2z_3^6(r_{-}-r_{+})^4-108\Xi^2\omega^2z_3^5r_{+}(r_{-}-r_{+})^3-54\Xi e q\omega z_3^5 (r_{-}-r_{+})^3\nonumber \\
&+18\Xi eq \omega z_3^4 z_4(r_{-}-r_{+})^3+18\Xi^2\omega^2z_3^4 z_4(r_{-}-r_{+})^3 (r_{-}+r_{+})\nonumber \\
&+108\Xi\omega z_3^2 z_4 r_{+}(r_{-}-r_{+})[\mathcal{K}(r_{+})+eqr_{+}]+
54 eq z_3^2 z_4(r_{-}-r_{+})[\mathcal{K}(r_{+})+eqr_{+}]\nonumber \\
&+36 z_3 z_4[\mathcal{K}(r_{+})+eqr_{+}]^2-18eqz_3z_4(r_{-}-r_{+})[\mathcal{K}(r_{+})+eqr_{+}]\nonumber \\
&-36\Xi\omega r_{+}z_3z_4(r_{-}-r_{+})[\mathcal{K}(r_{+})+eqr_{+}]\Biggr]\nonumber \\
&\left(\frac{-3 i}{2\Lambda}\right)\left(\frac{2\Xi\omega r_3+eq}{(r_{3}-r_{-})(r_3-r_4)(r_3-r_{+})}\right)(r_{-}-r_{+})\nonumber \\
&\left(\frac{-3 i}{\Lambda}\right)\left[\frac{\Xi\omega (r_3-r_{+})(r_3+r_{+})+eqr_{3}+\mathcal{K}(r_{+})}{(r_3-r_{+})(r_3-r_{-})(r_3-r_4)(r_3-r_5)}\right](r_{-}-r_{+})\nonumber \\
&+\frac{3\lambda^2 (r_{-}-r_{+})}{\Lambda(r_3-r_{+})(r_3-r_{-})(r_3-r_4)}+\frac{3\mu^2}{\Lambda}\frac{(r_{-}-r_{+})r_3^2}{(r_3-r_{+})(r_3-r_{-})(r_3-r_4)}\nonumber \\
&+\frac{6i\omega \Xi r_3(r_{-}-r_{+})}{\Lambda (r_{3}-r_{+})(r_3-r_{-})(r_3-r_4)}+\frac{3ieq (r_{-}-r_{+})}{(r_3-r_{-})(r_3-r_{+})(r_3-r_4)\Lambda}\nonumber \\
&\left(\frac{-3i}{2\Lambda}\right)\frac{(r_{-}-r_{+})[\mathcal{K}(r_3)+eqr_3]}{(r_3-r_{+})(r_3-r_4)^2(r_3-r_{-})}+
\left(\frac{-3i}{2\Lambda}\right)\frac{(r_{-}-r_{+})[\mathcal{K}(r_3)+eqr_3]}{(r_3-r_{+})(r_3-r_{-})^2(r_3-r_4)}\nonumber \\
&+\left(\frac{-3i}{2\Lambda}\right)\frac{(r_{-}-r_{+})[\mathcal{K}(r_3)+eqr_3]}{(r_3-r_{+})^2(r_3-r_{-})(r_3-r_4)}\nonumber \\
&+\left(\frac{3i}{2\Lambda}\right)\frac{(r_{-}-r_{+})^3}{(r_3-r_{+})^2(r_3-r_{-})^2(r_3-r_4)^2}\Biggl\{\Xi\omega z_3^4(r_{-}-r_{+})^2+4\Xi\omega r_{+}z_3^3 (r_{-}-r_{+})\nonumber\\
&+2 eq z_3^3(r_{-}-r_{+})+3z_3^2(\mathcal{K}(r_+)+eqr_+)-2\Xi\omega r_{+}z_3^2(r_{-}-r_{+})-eq(r_{-}-r_{+})z_3^2\nonumber \\
&-2z_3z_4(\mathcal{K}(r_+)+eqr_+)-\Xi\omega z_3^2z_4(r_{-}-r_{+})(r_{-}+r_{+})-eq z_3^2 z_4(r_{-}-r_{+})\nonumber\\
&-2 z_3(\mathcal{K}(r_+)+eqr_+)+z_4(\mathcal{K}(r_+)+eqr_+)\Biggr\},
\end{align}

\begin{align}
L_4&=\frac{1}{\Lambda^2 (r_{-}-r_{+})^6 z_4^3(z_4-1)^3(z_3-z_4)^3}\Biggl[
72\Xi(r_{-}-r_{+})^2 z_4^4\omega[\mathcal{K}(r_{+})+eqr_{+}]\nonumber \\
&+36\Xi^2\omega^2z_4^4(r_{-}-r_{+})^2r_{+}(r_{+}-r_{-})+18\Xi\omega eq z_4^4(r_{-}-r_{+})^2(r_{+}-r_{-})\nonumber \\
&+4\; 36 \Xi\omega z_4^4(r_{-}-r_{+})^2[\Xi\omega r_{+}^2+eqr_{+}]+36e^2 q^2 z_4^4(r_{-}-r_{+})^2\nonumber \\
&+108 \Xi\omega r_{+}(r_{+}-r_{-})z_4^2[\mathcal{K}(r_{+})+eqr_{+}]+108eqr_{+}z_4^2[\mathcal{K}(r_{+})+eqr_{+}]\nonumber \\
&+54 e q z_4^2 (r_{+}-r_{-})\mathcal{K}(r_{+})-54e^2 q^2 r_{-}r_{+}z_4^2+54z_4^2\mathcal{K}^2(r_{+})\nonumber \\
&+180 z_4^3 \Xi\omega r_{+}(r_{-}-r_{+})[\mathcal{K}(r_{+})+eqr_{+}]-36z_4^3\Xi\omega (r_{-}-r_{+})^2[\mathcal{K}(r_{+})+eqr_{+}]\nonumber \\
&-72\Xi\omega z_4^3(r_{-}-r_{+})^2 [eqr_{+}+\omega\Xi r_{+}^2]+90eqz_4^3(r_{-}-r_{+})[\mathcal{K}(r_{+})+eqr_{+}]\nonumber \\
&-18e^2q^2z_4^3(r_{-}-r_{+})^2\nonumber \\
&+18\Xi^2 \omega^2 z_4^6(r_{-}-r_{+})^4+108\Xi^2\omega^2 z_4^5 r_{+}(r_{-}-r_{+})^3+54\Xi e q\omega z_4^5(r_{-}-r_{+})^3\nonumber \\
&-18\Xi\omega z_4^4 z_3(r_{-}-r_{+})^3[eq+\omega\Xi(r_{+}+r_{-})]\nonumber \\
&-36\Xi\omega z_3 z_4^3 (r_{-}-r_{+})^2[\mathcal{K}(r_{+})+eqr_{+}]-36\Xi^2\omega^2
z_3 z_4^3 (r_{-}-r_{+})^2(r_{+}+r_{-})r_{+}\nonumber \\
&-18 \Xi e q \omega z_3 z_4^3(r_{-}-r_{+})^2(r_{-}+3r_{+})-18e^2 q^2 z_3 z_4^3(r_{-}-r_{+})^2\nonumber \\
&+18 eq z_3z_4(r_{-}-r_{+})[\mathcal{K}(r_{+})+eqr_{+}]+36\Xi\omega r_{+}z_3z_4(r_{-}-r_{+})[\mathcal{K}(r_{+})+eqr_{+}]\nonumber \\
&+36z_3z_4[\mathcal{K}(r_{+})+eqr_{+}]^2\nonumber \\
&-108\Xi\omega r_{+}(r_{-}-r_{+})z_3z_4^2[\mathcal{K}(r_{+})+eqr_{+}]
-54eqz_4^2 z_3 (r_{-}-r_{+})[\mathcal{K}(r_{+})+eqr_{+}]\Biggr]\nonumber \\
&-\frac{3\mu^2 r_4^2 (r_{-}-r_{+})}{\Lambda (r_4-r_{+})(r_4-r_{-})(r_3-r_4)}-\frac{3\lambda^2 (r_{-}-r_{+})}{\Lambda (r_4-r_{+})(r_4-r_{-})(r_3-r_4)}\nonumber \\
&-\frac{3ieq (r_{-}-r_{+})}{\Lambda (r_3-r_4)(r_{4}-r_{-})(r_4-r_{+})}-\frac{6i\omega \Xi r_4(r_{-}-r_{+})}{\Lambda (r_4-r_{+})(r_4-r_{-})(r_3-r_4)}\nonumber \\
&+\frac{3i}{\Lambda}\frac{r_{-}-r_{+}}{(r_4-r_{+})(r_4-r_{-})(r_3-r_4)(r_4-r_5)}[\Xi \omega (r_4-r_{+})(r_4+r_{+})+\mathcal{K}(r_{+})+eqr_4]\nonumber \\
&+\left(\frac{3i}{2\Lambda}\right)\left(\frac{2\Xi\omega r_4+eq}{(r_4-r_{-})(r_3-r_4)(r_4-r_{+})}\right)(r_{-}-r_{+})\nonumber \\
&+\left(\frac{-3i}{2\Lambda}\right)\frac{(r_{-}-r_{+})[\mathcal{K}(r_4)+eqr_{4}]}{(r_4-r_{+})(r_3-r_4)^2(r_4-r_{-})}
+\left(\frac{3i}{2\Lambda}\right)\frac{(r_{-}-r_{+})[\mathcal{K}(r_4)+eqr_{4}]}{(r_4-r_{+})^2(r_4-r_{-})(r_3-r_4)}\nonumber\\
&\left(\frac{3i}{2\Lambda}\right)\frac{(r_{-}-r_{+})[\mathcal{K}(r_4)+eqr_{4}]}{(r_4-r_{+})(r_4-r_{-})^2(r_3-r_4)}\nonumber \\
&+\left(\frac{-3i}{2\Lambda}\right)\frac{(r_{-}-r_{+})^3}{(r_4-r_{+})^2(r_4-r_{-})^2(r_3-r_4)^2}\Biggl\{
-\Xi\omega z_4^4 (r_{-}-r_{+})^2-4\Xi\omega r_{+}z_4^3(r_{-}-r_{+})\nonumber \\
&-2eq z_4^3 (r_{-}-r_{+})-3z_4^2(\mathcal{K}(r_+)+eqr_{+})+2\Xi\omega r_+ z_4^2(r_{-}-r_{+})+eqz_4^2 (r_{-}-r_{+})\nonumber \\
&+ \Xi\omega z_3z_4^2(r_{-}-r_{+})(r_{-}+r_{+})+eqz_3z_4^2(r_{-}-r_{+})+2z_3z_4(\mathcal{K}(r_+)+eqr_{+}) \nonumber \\
&-z_3 (\mathcal{K}(r_+)+eqr_{+})+2z_4(\mathcal{K}(r_+)+eqr_{+})\Biggr\},
\end{align}

\begin{align}
L_5=-\frac{3i}{\Lambda}\frac{r_{-}-r_{+}}{(r_5-r_{+})(r_5-r_{-})(r_3-r_5)(r_4-r_5)}[\Xi\omega (r_5-r_{+})(r_5+r_{+})+\mathcal{K}(r_{+})+eqr_5].
\end{align}
In eqn.(\ref{PreGeneHeunKNdS}), we use the notation $R_{-}(z)\equiv R^{(-)}(r_1+z(r_2-r_1))$.

Let us calculate the exponents of the five singular points in (\ref{PreGeneHeunKNdS}). For the singularity at $z=0$ the indicial equation becomes
\begin{equation}
F(s)=s(s-1)+\frac{1}{2}s+Q_1=s^2-\frac{1}{2}s+Q_1,
\end{equation}
with roots
\begin{equation}
s_{1,2}^{z=0}=\frac{\frac{1}{2}\pm \sqrt{(-\frac{1}{2})^2-4Q_1}}{2}=:\mu_1.
\end{equation}
Likewise the exponents at the singular points $z=1,z=z_3,z=z_4$ read as follows:
\begin{align}
s_{1,2}^{z=1}&=\frac{\frac{1}{2}\pm \sqrt{(-\frac{1}{2})^2-4Q_2}}{2}=:\mu_2,\\
s_{1,2}^{z=z_3}&=\frac{\frac{1}{2}\pm \sqrt{(-\frac{1}{2})^2-4Q_3}}{2}=:\mu_3,\\
s_{1,2}^{z=z_4}&=\frac{\frac{1}{2}\pm \sqrt{(-\frac{1}{2})^2-4Q_4}}{2}=:\mu_4.
\end{align}
On the other hand the exponents at the fifth singular point $z=z_5$ are computed to be $\{0,2\}$, i.e. they are both integers.
Now by applying the $F-$ homotopic transformation of the dependent variable:
\begin{equation}
R_{-}(z)=z^{\mu_1}(z-1)^{\mu_2}(z-z_3)^{\mu_3}(z-z_4)^{\mu_4}\overline{R}_{-}(z),
\end{equation}
equation (\ref{PreGeneHeunKNdS}) transforms into:
\begin{align}
&\Biggl\{ \frac{{\rm d}^2}{{\rm d}z^2}+\left(\frac{1+4\mu_1}{2z}+
\frac{1+4\mu_2}{2(z-1)}+\frac{1+4\mu_3}{2(z-z_3)}+\frac{1+4\mu_4}{2(z-z_4)}-\frac{1}{z-z_5}\right)\frac{{\rm d}}{{\rm d}z}\nonumber \\
&+\frac{L_1^{\prime}}{z}+\frac{L_2^{\prime}}{z-1}+\frac{L_3^{\prime}}{z-z_3}+
\frac{L_4^{\prime}}{z-z_4}+\frac{L_5^{\prime}}{z-z_5}\Biggr\}\overline{R}_{-}(z)=0,
\label{MasDiracGHrKNdS}
\end{align}
where
\begin{align}
L_1^{\prime}&=L_1+\left[\frac{-4\mu_1\mu_2-\mu_2-\mu_1}{2}\right]+\left[\frac{-4\mu_1\mu_3-\mu_3-\mu_1}{2z_3}\right]
+\left[\frac{-4\mu_1\mu_4-\mu_4-\mu_1}{2z_4}\right]+\frac{\mu_1}{z_5},\nonumber\\
L_2^{\prime}&=L_2+\left[\frac{4\mu_1\mu_2+\mu_1+\mu_2}{2}+\frac{4\mu_2\mu_3+\mu_2+\mu_3}{2(1-z_3)}+
\frac{4\mu_2\mu_4+\mu_2+\mu_4}{2(1-z_4)}+\frac{\mu_2}{z_5-1}\right],\nonumber \\
L_3^{\prime}&=L_3+\frac{4\mu_1\mu_3+\mu_3+\mu_1}{2z_3}+\frac{4\mu_2\mu_3+\mu_3+\mu_2}{2(z_3-1)}+
\frac{-4\mu_3\mu_4-\mu_4-\mu_3}{2(z_4-z_3)}+\frac{\mu_3}{z_5-z_3},\nonumber \\
L_4^{\prime}&=L_4+\frac{4\mu_2\mu_4+\mu_2+\mu_4}{2(z_4-1)}+\frac{4\mu_3\mu_4+\mu_3+\mu_4}{2(z_4-z_3)}+
\frac{4\mu_1\mu_4+\mu_4+\mu_1}{2z_4}-\frac{\mu_4}{z_4-z_5},\nonumber \\
L_5^{\prime}&=L_5-\frac{\mu_1}{z_5}-\frac{\mu_2}{z_5-1}-\frac{\mu_3}{z_5-z_3}-\frac{\mu_4}{z_5-z_4}.
\end{align}
We observe that Eqn.(\ref{MasDiracGHrKNdS}) possess five regular finite singularities and therefore constitutes a highly non-trivial  generalisation of the Heun differential equation-the latter has three finite regular singular points.

As we shall see in section \ref{KNDiracSection} for zero cosmological constant the corresponding  radial equation-see eqn.(\ref{GHERADIALKN})-also leads to a generalisation of Heun's equation, in particular it has the specific mathematical structure of a GHE, however with fewer finite singularities than eqn (\ref{MasDiracGHrKNdS}).

\subsection{Asymptotic forms of the radial equation in KNdS spacetime}\label{AsymptotsLambdaRegge}

The investigation of the asymptotic forms of the radial equation (\ref{IntervenKNdSrad}) can be facilitated if (\ref{IntervenKNdSrad}) is transformed by writing
\begin{equation}
R^{(-)}=(\Delta_r^{KN})^{1/4}(r^2+a^2)^{-1/2}(\lambda+i\mu r)^{1/2}\mathcal{R}_{\Lambda},
\end{equation}
and using the Regge-Wheeler-like (or ''tortoise'') coordinate
\begin{equation}
{\rm d}r^*=\frac{r^2+a^2}{\Delta_r^{KN}}{\rm d}r.
\label{ReggeW}
\end{equation}
The radial equation becomes then
\begin{align}
\frac{{\rm d}^2\mathcal{R}_{\Lambda}}{{\rm d}r^{*2}}&+\Biggl\{\frac{\left(K^2+\frac{iK}{2}\frac{{\rm d}\Delta_r^{KN}}{{\rm d}r}+\Delta_r^{KN}\left[\frac{-\mu K}{\lambda+i\mu r}-2i\omega r\Xi-ieq-\mu^2 r^2-\lambda^2\right]\right)}{(r^2+a^2)^2}\nonumber \\
&-\mathcal{G}_{\Lambda}^2-\frac{{\rm d}}{{\rm d}r^*}\mathcal{G}_{\Lambda}+\frac{1}{4}\frac{i\mu\Delta_r^{KN}}{(\lambda+i\mu r)(r^2+a^2)^2}\left[\frac{{\rm d}\Delta_r^{KN}}{{\rm d}r}-\frac{3i\mu\Delta_r^{KN}}{\lambda+i\mu r}\right]\Biggr\}\mathcal{R}_{\Lambda}=0,
\label{tortoiseDELambda}
\end{align}
where we define:
\begin{equation}
\mathcal{G}_{\Lambda}:=-\frac{1}{4}\frac{{\rm d}\Delta_r^{KN}/{\rm d}r}{r^2+a^2}+\frac{r\Delta_r^{KN}}{(r^2+a^2)^2}.
\end{equation}
Integrating Eqn.(\ref{ReggeW}) we obtain the $r^*(r)$ relation:
\begin{align}
r^*&=\int\frac{r^2+a^2}{\Delta_r^{KN}}{\rm d}r=\frac{3(a^2+r_+^2)}{(r_{\Lambda^-}-r_+)(r_+-r_{\Lambda^+})(r_+-r_-)\Lambda}\log|r-r_+|\nonumber\\
&+\frac{3(a^2+r_-^2)}{(r_{\Lambda^-}-r_)(r_{-}-r_{\Lambda^+})(r_{-}-r_+)\Lambda}\log|r-r_-|
+\frac{3(a^2+r_{\Lambda^+}^2)}{(r_{\Lambda^-}-r_{\Lambda^+})(r_{\Lambda^+}-r_-)(r_{\Lambda^+}-r_+)\Lambda}\log|r-r_{\Lambda^+}|\nonumber\\
&+\frac{(-3(a^2+r_{\Lambda^-})^2)}{(r_{\Lambda^-}-r_{\Lambda^+})(r_{\Lambda^-}-r_-)(r_{\Lambda^-}-r_+)\Lambda}
\log|r-r_{\Lambda^-}|+C.
\label{LambdaRWcoord}
\end{align}
\subsubsection{The near event horizon limit $r\rightarrow r_+$}

In the near event horizon limit $r\rightarrow r_+$ ($r^*\rightarrow -\infty$), equation (\ref{tortoiseDELambda}) takes the form:
\begin{equation}
\frac{{\rm d}^2 \mathcal{R}_{\Lambda}}{{\rm d}r^{*2}}+\left((\Xi \omega-\Xi m \omega_{+}+\frac{eqr_+}{r_+^2+a^2})+\frac{i}{4}\frac{1}{r_+^2+a^2}\frac{{\rm d}\Delta_r^{KN}}{{\rm d}r}\Big|_{r=r_{+}}\right)^2\mathcal{R}_{\Lambda}\sim 0,
\label{eventhorizonKNdS}
\end{equation}
where we define $\omega_+:=\frac{a}{r_+^2+a^2}$.

\subsubsection{The near cosmological horizon limit $r\rightarrow r_{\Lambda}^+$}
In the near cosmological horizon limit $r\rightarrow r_{\Lambda}^+$ ($r^*\rightarrow \infty$) equation (\ref{tortoiseDELambda}) takes the form:
\begin{equation}
\frac{{\rm d}^2\mathcal{R}_{\Lambda}}{{\rm d}r^{*2}}+\left((\Xi \omega-\Xi m\omega_{\Lambda^+}+\frac{eqr_{\Lambda^+}}{r_{\Lambda^+}^2+a^2})+\frac{i}{4}\frac{1}{r_{\Lambda^+}^2+a^2}
\frac{{\rm d}\Delta_r^{KN}}{{\rm d}r}\Big|_{r=r_{\Lambda^+}}\right)^2\mathcal{R}_{\Lambda}\sim0,
\label{cosmohorizonLambda}
\end{equation}
where we define $\omega_{\Lambda^+}=\frac{a}{r_{\Lambda^+}^2+a^2}$.
\section{The general relativistic Dirac equation in the Kerr-Newman spacetime}\label{KNDiracSection}

In the case of the Kerr-Newman spacetime (KN) the Kinnersley null tetrad is a special case of (\ref{GLAMBDAKNNULLTETRAD}) for $\Lambda=0$ and takes the form:
\begin{align*}
l^{\mu}&=\left(\frac{r^2+a^2}{\Delta^{KN}},1,0,\frac{a}{\Delta^{KN}}\right),\;
n^{\mu}=\left(\frac{r^2+a^2}{2\rho^2},-\frac{\Delta^{KN}}{2\rho^2},0,\frac{a}{2\rho^2}\right),\\
m^{\mu}&=\frac{1}{\sqrt{2}(r+ia\cos\theta)}\left(i a\sin\theta,0,1,\frac{i}{\sin\theta}\right).
\end{align*}
The quantity $\Delta^{KN}$ is given by (\ref{DiscrimiL}) by setting $\Lambda=0$, i.e. $\Delta^{KN}:=r^2+a^2+e^2-2 M r$.
In the KN spacetime the non-vanishing $\lambda$-symbols are computed by setting $\Lambda=0$ in the expressions (\ref{lambdasymbolsL}), a procedure that yields:
\begin{align}
\lambda_{213}&=-\sqrt{2}\frac{a^2\sin\theta\cos\theta}{\rho^2\overline{\rho}},\;\lambda_{324}=
-\frac{ia\cos\theta\Delta^{KN}}{\rho^4},\\
\lambda_{243}&=\frac{-\Delta^{KN}}{2\rho^2\overline{\rho}},\;\lambda_{234}=-\frac{\Delta^{KN}}{2\rho^2(r-ia\cos\theta)},\\
\lambda_{134}&=\frac{1}{r-ia\cos\theta}=\frac{1}{\overline{\rho}^*},\;\lambda_{314}=\frac{-2ia\cos\theta}{\rho^2},\\
\lambda_{122}&=-\frac{1}{2}\frac{{\rm d}\Delta^{KN}}{{\rm d}r}\frac{1}{\rho^2}+r\frac{\Delta^{KN}}{\rho^4},
\;\lambda_{132}=\frac{\sqrt{2}ira\sin\theta}{\overline{\rho}\rho^2},\\
\lambda_{334}&=\frac{1}{\sin\theta\sqrt{2}\overline{\rho}}\frac{{\rm d}(\sin\theta)}{{\rm d}\theta}+\frac{ia\sin\theta}{\sqrt{2}\overline{\rho}^2}=\frac{r\cos\theta+ia}{\sqrt{2}\overline{\rho}^2\sin\theta},\;\lambda_{241}=\frac{ira\sqrt{2}\sin\theta}{\rho^2\overline{\rho}^*}
,\\
\lambda_{412}&=\frac{\sqrt{2}a^2 \sin\theta\cos\theta}{\rho^2 (r-ia\cos\theta)},\;
\lambda_{443}=\frac{{\rm d}}{{\rm d}\theta}(\sin\theta)\frac{1}{\sqrt{2}\overline{\rho}^*\sin\theta}-\frac{ia\sin\theta}{\sqrt{2}\overline{\rho}^{*2}},
\label{lambdasymbolsL=0}
\end{align}
while the Dirac equations in the curved background of the KN black hole spacetime have the same general form as in the more general KNdS  case:
\begin{align}
(D^{\prime}-\gamma+\mu+iqn^{\mu}A_{\mu})P^{(1)}+(\delta-\tau+\beta+iq m^{\mu}A_{\mu})P^{(0)}&=-i\mu_{*}\overline{Q}^{(\dot{0})},\label{KNBHDIR1}\\
(-D+\varrho-\varepsilon-iql^{\mu}A_{\mu})P^{(0)}+(-\delta^{\prime}+\alpha-\pi-iq\overline{m}^{\mu}A_{\mu})P^{(1)}&=-i\mu_{*}\overline{Q}^{(\dot{1})}
\label{KNBHDIRAC2}
\end{align}
however with different Ricci coefficients and differential operators for the directional derivatives calculated from those of the KNdS case by setting $\Lambda=0$.
In equations (\ref{KNBHDIR1}),(\ref{KNBHDIRAC2}):
\begin{align}
D&=\nabla_{0\dot{0}}=o^A o^{\dot{A}}\nabla_{A\dot{A}}=l^{\mu}\partial_{\mu}\\
\delta&=\nabla_{0\dot{1}}=o^A \iota^{\dot{A}}\nabla_{A\dot{A}}=m^{\mu}\partial_{\mu}\\
\delta^{\prime}\equiv\overline{\delta}&=\nabla_{1\dot{0}}=\iota^{A}o^{\dot{A}}\nabla_{A\dot{A}}=\overline{m}^{\mu}\partial_{\mu}\\
D^{\prime}&=\nabla_{1\dot{1}}=\iota^{A}\iota^{\dot{A}}\nabla_{A\dot{A}}=n^{\mu}\partial_{\mu}
\end{align}
where $o^A,\iota^A$ are basis spinors with $o_A \iota^A=1$. The quantities $\varrho,\varepsilon,$ etc are the spin coefficients (\ref{lambdaricci})-(\ref{cosmoconsricci}), after setting $\Lambda=0$.
Also $A_{B\dot{B}}=\sigma^{\mu}_{B\dot{B}}A_{\mu}$ and
\begin{align}
\sigma_{0\dot{0}}^{\mu}&=\iota^1\iota^{\dot{1}}l^{\mu}-o^{\dot{1}}\iota^1 m^{\mu}-o^1 \iota^{\dot{1}}\overline{m}^{\mu}+o^1 o^{\dot{1}}n^{\mu}, \\
\sigma_{1\dot{0}}^{\mu}&=-\iota^0 \iota^{\dot{1}}l^{\mu}+\iota^0 o^{\dot{1}}m^{\mu}+o^0\iota^{\dot{1}}\overline{m}^{\mu}-o^0o^{\dot{1}}n^{\mu},\\
\sigma^{\mu}_{0\dot{1}}&=-\iota^1\iota^{\dot{0}}l^{\mu}+o^{\dot{0}}\iota^1 m^{\mu}+o^1\iota^{\dot{0}}\overline{m}^{\mu}-o^1o^{\dot{0}}n^{\mu},\\
\sigma^{\mu}_{1\dot{1}}&=\iota^{\dot{0}}\iota^{0}l^{\mu}-\iota^0 o^{\dot{0}}m^{\mu}-o^0\iota^{\dot{0}}\overline{m}^{\mu}+o^0o^{\dot{0}}n^{\mu}.
\end{align}
We have also made use of the tetrad formalism to define associated local spinor components
\begin{equation}
y^{(k)}=\zeta^{(k)}_{\;A}y^A,\;\;y_{(k)}=\zeta_{(k)}^A y_A,\;\;\;k\in\{0,1\},
\end{equation}
for the original 2-spinor components $y^A,y_A\in \mathbb{C}^2$ through the spinor dyad a la Penrose
\begin{equation}
\zeta_{(k)}^A=(o^A,\iota^A).
\end{equation}
The spin coefficients are defined generically as
\begin{equation}
\Gamma_{\mathbf{A}\mathbf{\dot{A}}(c)}^{(b)}:=\zeta_A^{(b)}\nabla_{\mathbf{A\dot{A}}}\zeta_{(c)}^A.
\end{equation}

\subsection{Separation of the general relativistic Dirac equation in the Kerr-Newman space-time}
Applying the separation ansatz:
\begin{align}
P^{(0)}&=\frac{e^{-i\omega t} e^{im\phi}S^{(-)}(\theta)R^{(-)}(r)}{\sqrt{2}(r-ia\cos\theta)},\;
P^{(1)}=e^{-i\omega t}e^{im\phi}\frac{S^{(+)}(\theta)R^{(+)}(r)}{\sqrt{\Delta^{KN}}},\\
\overline{Q}^{(\dot{0})}&=\frac{-e^{-i\omega t}e^{im\phi}S^{(+)}(\theta)R^{(-)}(r)}{\sqrt{2}(r+ia\cos\theta)},\;
\overline{Q}^{(\dot{1})}=\frac{e^{-i\omega t}e^{im\phi}}{\sqrt{\Delta^{KN}}}S^{(-)}(\theta)R^{(+)}(r).
\end{align}
we obtain the following ordinary radial and angular differential equations
\begin{align}
&\sqrt{\Delta^{KN}}\left[\frac{{\rm d}R^{(-)}(r)}{{\rm d}r}+\left\{\frac{ima-i\omega(r^2+a^2)}{\Delta^{KN}}\right\}R^{(-)}(r)-\frac{iqerR^{(-)}(r)}{\Delta^{KN}}\right]=(\lambda+i\mu r)R^{(+)}(r),\\
&\sqrt{\Delta^{KN}}\frac{{\rm d}R^{(+)}(r)}{{\rm d}r}-\frac{i(ma-\omega(r^2+a^2))}{\sqrt{\Delta^{KN}}}R^{(+)}(r)+\frac{ieqr}{\sqrt{\Delta^{KN}}}R^{(+)}(r)=
(\lambda-i\mu r)R^{(-)}(r),\\
&\frac{{\rm d}S^{(+)}(\theta)}{{\rm d}\theta}+\left[\frac{m}{\sin\theta}-\omega a\sin\theta+\frac{\cot\theta}{2}\right]
S^{(+)}(\theta)=(-\lambda+\mu a \cos\theta)S^{(-)}(\theta)
\label{cornerKN1},\\
&\frac{{\rm d}S^{(-)}(\theta)}{{\rm d}\theta}+\left[\omega a\sin\theta-\frac{m}{\sin\theta}+\frac{\cot\theta}{2}\right]S^{(-)}(\theta)=(\lambda+\mu a \cos\theta)S^{(+)}(\theta)\label{corner2KNbh},
\end{align}
where $\mu_{*}=\mu/\sqrt{2}$.

The radial equation for the $R^{(-)}$ mode can be written as follows \cite{page}:
\begin{align}
\Biggl[\Delta^{KN}\frac{{\rm d}^2}{{\rm d}r^2}&+\left(r-M-\frac{i\mu\Delta^{KN}}{\lambda+i\mu r}\right)\frac{\rm d}{{\rm d}r}+\frac{K^2+i(r-M)K}{\Delta^{KN}}\nonumber \\
&-2i\omega r-ieq-\frac{\mu K}{\lambda+i\mu r}-\mu^2 r^2-\lambda^2\Biggr]R_{-\frac{1}{2}}(r)=0,
\label{raddira}
\end{align}
where $K\equiv(r^2+a^2)\omega-ma+eqr$ and $R_{-\frac{1}{2}}(r)\equiv R^{(-)}(r)$.

\subsection{Transforming the angular KN equation into a generalised Heun form}\label{anguGHE}
Using the variable $x:=\cos\theta$ the previous  angular differential equation (\ref{LEEPAGE}) is written:
\begin{align}
&\Biggl\{(1-x^2)\frac{{\rm d}^2}{{\rm d}x^2}-2x\frac{\rm d}{{\rm d}x}-\frac{a\mu(1-x^2)}{\lambda+a\mu x}\frac{\rm d}{{\rm d}x}\nonumber \\
&+\frac{a\mu(\frac{x}{2}+a\omega (1-x^2)-m)}{\lambda+a\mu x}+a^2(\omega^2-\mu^2)x^2+a\omega x-\frac{1}{4} \nonumber \\
&+\lambda^2+2am\omega-a^2\omega^2+\frac{-m^2+mx-\frac{1}{4}}{1-x^2}\Biggr\}S^{-}(x)=0,
\label{DiracMas}
\end{align}
where $S^{-}(x)\equiv S^{(-)}(\arccos x)$.
Equation (\ref{DiracMas}) possess three finite singularities at the points $x=\pm 1, x=-\lambda/a\mu$ which we denote using the triple:$(a_1,a_2,a_3)=(-1,+1,-\lambda/a\mu)$. Applying the transformation of the independent variable:
\begin{equation}
z=\frac{x-a_1}{a_2-a_1}=\frac{x+1}{2},
\end{equation}
maps $(a_1,a_2)$ to $(0,1)$ while the remaining singularity $a_3$ is mapped to $z=z_3$:
\begin{equation}
z_3=\frac{a_3-a_1}{a_2-a_1}=\frac{-\lambda/a\mu+1}{2}.
\end{equation}
In terms of the new variable $z$ the differential equation (\ref{DiracMas}) becomes:
\begin{align}
& \Biggl\{ \frac{\mathrm{d}^2}{\mathrm{d}z^2}+\left[\frac{1}{z}+\frac{1}{z-1}
-\frac{1}{z-z_3}\right]\frac{\mathrm{d}}{\mathrm{d}z}+4a^2(\mu^2-\omega^2)+\frac{a^2(\mu^2-\omega^2)}{z-1}-\frac{a^2(\mu^2-\omega^2)}{z} \nonumber \\
&+\frac{1}{16}\frac{-4m^2-4m-1}{z^2}+\frac{1}{8}\frac{4m^2+1}{z-1}+\frac{1}{16}\frac{-4m^2+4m-1}{(z-1)^2}
+\frac{1}{8}\frac{-4m^2-1}{z} \nonumber \\
&+\frac{1}{4}\frac{8a\omega z_3^2-8a\omega z_3+2m-2z_3+1}{z_3(z_3-1)(z-z_3)}+\frac{-2m+1}{(z-1)(-4+4z_3)}+\frac{1}{4}\frac{2m+1}{z_3 z}\nonumber \\
&+\frac{1}{4}\frac{4a^2\omega^2 -8am\omega-4a\omega-4\lambda^2+1}{z-1}+
\frac{1}{4}\frac{-4a^2\omega^2+8am\omega-4a\omega+4\lambda^2-1}{z}
\Biggr\}S(z)=0,
\label{preGHEUN}
\end{align}
where $S(z)\equiv S^-(2z-1)$.
Equation (\ref{preGHEUN}) is a particular case of the following generic form:
\begin{align}
S^{\prime\prime}(z)&+\left(\sum_{i=1}^2 \frac{A_i}{z-z_i}+\frac{-1}{z-z_3}+E_0\right)S^{\prime}(z)\nonumber \\
&+\left(\sum_{i=1}^3\frac{C_i}{z-z_i}+\sum_{i=1}^2\frac{B_i}{(z-z_i)^2}+D_0\right)S(z)=0.
\end{align}
with $z_1=0$ and $z_2=1$.
Let us calculate the exponents of the singularities. The indicial equation for the $z=0$ singularity is:
\begin{equation}
F(s)=s(s-1)+s-\frac{1}{4}\left(m+\frac{1}{2}\right)^2=0,
\end{equation}
with roots: $s_{1,2}^{z=0}=\pm\frac{1}{2}|m+1/2|$. Likewise the exponents at the singularities $z=1,z=z_3$ are computed to be: $\{\frac{|m-\frac{1}{2}|}{2},-\frac{|m-\frac{1}{2}|}{2}\}, \{0,2\}$ respectively. Applying the index transformation for the dependent variable $S$:
\begin{equation}
S(z)=e^{\nu z}z^{\alpha_1}(z-1)^{\alpha_2}(z-z_3)^{\alpha_3}\bar{S}(z),
\label{fhomotopyDIRAC}
\end{equation}
where $\nu=\pm i2a\sqrt{\mu^2-\omega^2},\alpha_1=\frac{1}{2}|m+1/2|,\alpha_2=\frac{1}{2}|m-1/2|$ and $\alpha_3=0$ yields \footnote{This calculation of ours rectifies our previous erroneous result of subsection $5.2.2$ in \cite{Kraniotis1} that the angular equation for a massive spin-half particle in the KN background reduces to Heun's form.} a generalised Heun equation (GHE)\footnote{See section \ref{SolvingGHE} for the mathematical background on the generalised Heun differential equation (GHE)  and in particular eqn. (\ref{SCHASCHMI}) for its generic mathematical representation.}:
\begin{equation}
\Biggl\{\frac{\mathrm{d}^2}{\mathrm{d}z^2}+\left[\frac{2\alpha_1+1}{z}+
\frac{2\alpha_2+1}{z-1}+\frac{-1}{z-z_3}\pm 4a\sqrt{\omega^2-\mu^2}\right]\frac{\mathrm{d}}{\mathrm{d}z}+
\sum_{i=1}^3\frac{C_i^{\prime}}{z-z_i}\Biggr\}\bar{S}(z)=0\Leftrightarrow,
\label{GHeunMunchenPDIRACmas}
\end{equation}
\begin{align}
\Biggl\{\frac{\mathrm{d}^2}{\mathrm{d}z^2}&+\left[\frac{2\alpha_1+1}{z}+
\frac{2\alpha_2+1}{z-1}+\frac{-1}{z-z_3}\pm 4a\sqrt{\omega^2-\mu^2}\right]\frac{\mathrm{d}}{\mathrm{d}z}\nonumber \\
&+\frac{\beta_0+\beta_1 z+\beta_2 z^2}{z(z-1)(z-z_3)}\Biggr\}\bar{S}(z)=0,
\label{GONGHE}
\end{align}
where
\begin{equation}
\beta_2=C_1^{\prime}+C_2^{\prime}+C_3^{\prime}=2\nu(\alpha_1+\alpha_2)+\nu(A_1+A_2-1)=2\nu(\alpha_1+\alpha_2)+\nu.
\end{equation}

\subsection{Transforming the radial KN equation into a generalised Heun form}\label{AKTINAGHE}
We will apply the following transformation of variables to the radial equation (\ref{raddira})
\begin{equation}
z=\frac{r-r_{-}}{r_{+}-r_{-}},
\end{equation}
which transforms the radii of the event and Cauchy horizons to the points $z=1$ and $z=0$ respectively and the singularity $r_3=\frac{i\lambda}{\mu}$ to the point $z_3=\frac{r_3-r_{-}}{r_{+}-r_{-}}$.
The radial equation (\ref{raddira}) then takes the form:
\begin{align}
&\Biggl\{\frac{{\rm d}^2}{{\rm d}z^2}+\left[\left(\frac{r_{+}-M}{r_{+}-r_{-}}\right)\frac{1}{z-1}-\left(\frac{r_{-}-M}{r_{+}-r_{-}}\right)\frac{1}{z}
-\frac{1}{z-z_3}\right]\frac{\rm d}{{\rm d}z}\nonumber \\
&+(\omega^2-\mu^2)(r_{+}-r_{-})^2-\frac{\mu^2 r_{+}^2}{z-1}+\frac{\mu^2 r_{-}^2}{z}+
\frac{1}{z}(ieq+2i\omega r_{-}+\lambda^2)+\frac{1}{z-1}(-2i\omega r_{+}-ieq-\lambda^2) \nonumber \\
&-\frac{1}{z}\frac{i\mathcal{H}_{-}}{(r_{-}-r_{+})z_3}
+\frac{1}{z-1}\frac{i\mathcal{H}_{+}}{(r_{-}-r_{+})(z_3-1)} \nonumber \\
&-\frac{1}{z-z_3}\frac{1}{(r_{-}-r_{+})z_3(z_3-1)}\left(i(\omega(r_{-}-r_{+})^2z_3^2+eq(r_{+}-r_{-})z_3-
2\omega r_{-}z_3(r_{-}-r_{+})+\mathcal{H}_{-})\right)+\nonumber \\
&\frac{1}{z}\frac{1}{(r_{-}-r_{+})^2}\Biggl[i\Biggl(a^2\omega(r_{+}+r_{-})-2M\omega(a^2+r_{+}r_{-})+\omega r_{-}^2(3r_{+}-r_{-})\nonumber \\
&+M(2 am-eq (r_{+}+r_{-}))-am(r_{+}+r_{-})+2eq r_+ r_-\Biggr)\nonumber\\
&+2eqa(r_{+}+r_{-})(a\omega-m)+2eq\omega r_{-}^2(3r_{+}-r_{-})+2e^2q^2r_{+}r_{-}\nonumber \\
&+4\omega^2r_{+}r_{-}(r_{-}^2+a^2)-4am\omega(a^2+r_{+}r_{-})+2\omega^2(a^4-r_{-}^4)+2a^2m^2\Biggr]\nonumber \\
&+\frac{1}{z-1}\frac{1}{(r_{-}-r_{+})^2}\Biggl[i\Biggl(-a^2\omega(r_{+}+r_{-})+2M\omega(a^2+r_{+}r_{-})+\omega r_{+}^2(r_{+}-3r_{-})\nonumber \\
&+am(r_{+} +r_{-}-2M)+eqM(r_++r_-)-2eqr_- r_+\Biggr)\nonumber \\
&+2eqa(r_{+}+r_{-})(m-a\omega)+2eq\omega r_{+}^2(r_{+}-3r_{-})-2e^2q^2r_{+}r_{-}\nonumber \\
&-4\omega^2r_{-}r_{+}(a^2+r_{+}^2)+4am\omega(a^2+r_{-}r_{+})+2\omega^2(r_{+}^4-a^4)-2m^2a^2\Biggr]\nonumber \\
&+\frac{1}{(z-1)^2}\frac{1}{(r_{-}-r_{+})^2}\Biggl[i\omega(a^2+r_{+}^2)(r_{+}-M)
+i(am-eqr_+)(M-r_+)\nonumber \\
&+eq r_{+}(2a(a\omega-m)+2\omega r_{+}^2)+e^2q^2r_{+}^2\nonumber \\
&+2a\omega r_{+}^2(a\omega -m)+a^2m^2+\omega^2(r_{+}^4+a^4)-2a^3m\omega\Biggr]\nonumber\\
&+\frac{1}{z^2}\frac{1}{(r_{-}-r_{+})^2}\Biggl[i\omega(a^2+r_{-}^2)(r_{-}-M)+i(M-r_-)(am-eqr_-)\nonumber \\
&+eq r_{-}(2a(a\omega-m)+2\omega r_{-}^2)+e^2q^2r_{-}^2\nonumber \\
&+2a\omega r_{-}^2(a\omega -m)+a^2m^2+\omega^2(r_{-}^4+a^4)-2a^3m\omega\Biggr]\Biggr\}R(z)=0,
\label{radialKNbef}
\end{align}
where $R(z)\equiv R_{-\frac{1}{2}}(r(z))$ and
\begin{equation}
\mathcal{H}_{\pm}:=a^2\omega+eqr_{\pm}+\omega r_{\pm}^2-am.
\end{equation}
The indicial equation for the singularity at $z=0$ takes the form:
\begin{equation}
F(s)=s(s-1)-\left(\frac{r_{-}-M}{r_{+}-r_{-}}\right)s+{\mathcal B}_0,
\label{indicialrad0}
\end{equation}
where ${\mathcal B}_0$ is the coefficient of the term $\frac{1}{z^2}$. The roots of equation (\ref{indicialrad0}) are:
\begin{equation}
s_{1,2}^{z=0}=\frac{-\left[\frac{M-r_{+}}{r_{+}-r_{-}}\right]\pm\sqrt{\Delta_0}}{2}=:\boldsymbol{\mu}_1,
\end{equation}
where $\Delta_0=\left(\frac{M-r_{+}}{r_{+}-r_{-}}\right)^2-4{\mathcal B}_0$. Likewise, the indicial equation for the singularity at $z=1$ reads:
\begin{equation}
F(s)=s(s-1)+\left(\frac{r_{+}-M}{r_{+}-r_{-}}\right)s+{\mathcal B}_1,
\label{indicialrad1}
\end{equation}
where ${\mathcal B}_1$ is the coefficient of the term $\frac{1}{(z-1)^2}$ in (\ref{radialKNbef}). The roots of (\ref{indicialrad1}) are:
\begin{equation}
s_{1,2}^{z=1}=\frac{\frac{M-r_{-}}{r_{+}-r_{-}}\pm\sqrt{\Delta_1}}{2}=:\boldsymbol{\mu_2},
\end{equation}
with the discriminant $\Delta_1=\left(\frac{M-r_{-}}{r_{+}-r_{-}}\right)^2-4{\mathcal B}_1$.
Equation (\ref{radialKNbef}) is of the following generic form:
\begin{align}
R^{\prime\prime}(z)&+\left(\sum_{i=1}^2 \frac{A_i}{z-z_i}+\frac{-1}{z-z_3}+E_0\right)R^{\prime}(z)\nonumber \\
&+\left(\sum_{i=1}^3\frac{C_i}{z-z_i}+\sum_{i=1}^2\frac{B_i}{(z-z_i)^2}+D_0\right)R(z)=0.
\end{align}
Applying now the $F-$ homotopic transformation of the dependent variable $R$
\begin{equation}
R(z)=e^{\nu z}z^{\boldsymbol{\mu_1}}(z-1)^{\boldsymbol{\mu_2}}\bar{R}(z)
\end{equation}
where $\nu=\pm i\sqrt{\omega^2-\mu^2}(r_{+}-r_{-})$
transforms the radial part of the Dirac equation in the curved spacetime of the KN black hole into a generalised Heun differential equation:
\begin{align}
\Biggl\{\frac{{\rm d}^2}{{\rm d}z^2}&+\left[\frac{2\boldsymbol{\mu_1}+\frac{M-r_{-}}{r_{+}-r_{-}}}{z}+
\frac{2\boldsymbol{\mu_2}+\frac{r_{+}-M}{r_{+}-r_{-}}}{z-1}-\frac{1}{z-z_3}\pm 2i\sqrt{\omega^2-\mu^2}(r_{+}-r_{-})\right]\frac{{\rm d}}{{\rm d}z}\nonumber \\
&+\sum_{i=1}^3\frac{C_i^{\prime}}{z-z_i}\Biggr\}\bar{R}(z)=0,
\label{GHERADIALKN}
\end{align}
where
\begin{align}
C^{\prime}_1&=C_1+2i\sqrt{\omega^2-\mu^2}(r_{+}-r_{-})\boldsymbol{\mu_1}+A_1 i\sqrt{\omega^2-\mu^2}(r_{+}-r_{-})+\frac{A_1\boldsymbol{\mu_2}}{z_1-z_2}+\frac{A_2\boldsymbol{\mu_1}}{z_1-z_2}\nonumber \\ &+\frac{2\boldsymbol{\mu_1}\boldsymbol{\mu_2}}{z_1-z_2}-
\frac{\boldsymbol{\mu_1}}{z_1-z_3},\\
C^{\prime}_2&=C_2+2i\sqrt{\omega^2-\mu^2}(r_{+}-r_{-})\boldsymbol{\mu_2}+A_2 i\sqrt{\omega^2-\mu^2}(r_{+}-r_{-})+\frac{A_1\boldsymbol{\mu_2}}{z_2-z_1}+
\frac{A_2\boldsymbol{\mu_1}}{z_2-z_1}\nonumber \\
&+\frac{2\boldsymbol{\mu_1}\boldsymbol{\mu_2}}{z_2-z_1}-\frac{\boldsymbol{\mu_2}}{z_2-z_3},\\
C^{\prime}_3&=C_3+\frac{\boldsymbol{\mu_1}}{z_1-z_3}+\frac{\boldsymbol{\mu_2}}{z_2-z_3}-i\sqrt{\omega^2-\mu^2}(r_{+}-r_{-}).
\end{align}
The fact that the radial equation in the KN spacetime has the mathematical structure of a GHE was also noted in \cite{Baticr}, on the basis of a  computation using a Carter tetrad.

\section{Analytic solutions of the generalised Heun equation}\label{SolvingGHE}
In the previous section we managed to reduce the radial and angular parts of the  massive Dirac equation in the background of the Kerr-Newman black hole to the generalised Heun differential equation. In this section we will investigate its solutions.
The general form of the generalised Heun equation in the mathematical literature has the form \cite{RSCHAFKEDSCHMIDT}:
\begin{align}
y^{\prime\prime}(z)&+\left(\frac{1-\mu_0}{z}+\frac{1-\mu_1}{z-1}+\frac{1-\mu_2}{z-a}+\alpha\right)y^{\prime}(z)\nonumber \\
&+\frac{\beta_0+\beta_1 z+\beta_2 z^2}{z(z-1)(z-a)}y(z)=0.
\label{SCHASCHMI}
\end{align}
where $a\in\mathbb{C}\setminus\{0,1\}$ and $\alpha\not=0$ and $\mu_j,\beta_j$ are complex parameters. This equation with three regular singular points and one irregular singular point at $0,1,a$ and $\infty$ has been discussed in \cite{RSCHAFKEDSCHMIDT}.
 In this form the exponents at the singularities $z=0,1,a$ are respectively $\{0,\mu_0\},\;\{0,\mu_1\},\;\{0,\mu_2\}$. Because of the symmetry of (\ref{SCHASCHMI}) in the parameters $\mu:=(\mu_0,\mu_1,\mu_2)$ under certain index or $F-$ homotopic transformations, one allows the coefficient of the $y(z)$ in (\ref{SCHASCHMI}) to have the form \footnote{We note that when  $\alpha=0$ and $\beta_2(=\lambda_0+\lambda_1+\lambda_2)=0$, $\infty$ is also a regular singularity and (\ref{SCHASCHMI}) reduces to Heun's equation.}:
 \begin{equation}
 \frac{\beta_0+\beta_1 z+\beta_2 z^2}{z(z-1)(z-a)}=\sum_{\substack {\sigma,\rho=0\\ \sigma\not =\rho}}^2\frac{1}{2}\left(
 \frac{1-\mu_{\sigma}}{z-z_{\sigma}}\right)\left(\frac{1-\mu_{\rho}}{z-z_{\rho}}\right)+
 \sum_{k=0}^2 \frac{(\alpha/2)(1-\mu_k)+\lambda_k}{z-z_{k}}
 \end{equation}
with $z_0=0,z_1=1,z_2=a\in\mathbb{C}\setminus\{0,1\}$, and arbitrary parameters $\lambda:=(\lambda_0,\lambda_1,\lambda_2)\in \mathbb{C}^3$. The following  was proven in \cite{RSCHAFKEDSCHMIDT}:
\begin{theorem}\label{eta}
Let $a\in\mathbb{C}\setminus\{0,1\}$ be fixed. Then there exists a unique function $\eta=\eta(\cdot;a)$ holomorphic with respect to $(z,\mu,\alpha,\lambda)\in\{z\in\mathbb{C}:
|z|<\min(1,|a|)\}\times \mathbb{C}^7,$ such that, for each $(\mu,\alpha,\lambda),\eta(\cdot,\mu,\alpha,\lambda;a)$ is a solution of (\ref{SCHASCHMI}) satisfying $\eta(0,\mu,\alpha,\lambda;a)=\frac{1}{\Gamma(1-\mu_0)}$. The function $\eta$ can be expanded in a power series
\begin{equation}
\eta(z,\mu,\alpha,\lambda;a)=\sum_{k=0}^{\infty}\frac{T_k(\mu,\alpha,\lambda;a)}{\Gamma(k+1-\mu_0)\Gamma(k+1)}z^k,
\label{powerseriesGHE}
\end{equation}
where the (unique) coefficients $T_k$ are holomorphic in $(\mu,\alpha,\lambda)$. In particular $T_0(\mu,\alpha,\lambda;a)=1$.
\end{theorem}
Substituting (\ref{powerseriesGHE}) into (\ref{SCHASCHMI}) we obtain recursion relations for the coefficients $T_k$:
\begin{align}
&\frac{\beta_0T_0}{\Gamma(1-\mu_0)\Gamma(1)}+\frac{T_1(1-\mu_0)a}{\Gamma(2-\mu_0)\Gamma(2)}=0,\Leftrightarrow
T_1=-\frac{\beta_0}{a}T_0,\\
&T_2=\left[\frac{2-\mu_0-\mu_2}{a}+(2-\mu_0-\mu_1)-\alpha-\frac{\beta_0}{a}\right]T_1-\frac{\beta_1(1-\mu_0)}{a}T_0,\\
&T_3=\left[2(3-\mu_0-\mu_1)+\frac{2}{a}(3-\mu_0-\mu_2)-2\alpha-\frac{\beta_0}{a}\right]T_2\nonumber \\
&+\left[-\frac{2}{a}(2-\mu_0)(3-\mu_0-\mu_1-\mu_2)+2\frac{(2-\mu_0)}{a}\alpha(a+1)-\frac{2(2-\mu_0)\beta_1}{a}\right]T_1\nonumber \\
&-\frac{\beta_2 2(2-\mu_0)(1-\mu_0)T_0}{a}\\
&T_4=\left[3(4-\mu_0-\mu_1)+\frac{3}{a}(4-\mu_0-\mu_2)-3\alpha-\frac{\beta_0}{a}\right]T_3\nonumber \\
&+3(3-\mu_0)\left[\frac{2}{a}(-4+\mu_0+\mu_1+\mu_2)+2\alpha\frac{a+1}{a}-\frac{\beta_1}{a}\right]T_2\nonumber \\
&+2\;3(2-\mu_0)(3-\mu_0)\left(-\frac{1}{a}\right)[\alpha+\beta_2]T_1,\\
&\cdots \nonumber
\end{align}
which are summarised in the four-term recurrence relation for the $T_k$
\begin{equation}
T_k=\varphi_1(k-1)T_{k-1}-\varphi_2(k-2)T_{k-2}+\varphi_3(k-3)T_{k-3}, \;\;\;(k\in\mathbb{N}),
\end{equation}
where $T_{-1}=T_{-2}=0$ and
\begin{align}
\varphi_1(\xi)&=\xi(\xi+1-\mu_0-\mu_1)+\frac{1}{a}\xi(\xi+1-\mu_0-\mu_2)-\alpha\xi-\frac{1}{a}\beta_0,\\
\varphi_2(\xi)&=(\xi+1)(\xi+1-\mu_0)\left(\frac{1}{a}\xi(\xi+2-\mu_0-\mu_1-\mu_2)-\left(1+\frac{1}{a}\right)\alpha\xi+
\frac{1}{a}\beta_1\right),\\
\varphi_3(\xi)&=(\xi+1)(\xi+2)(\xi+1-\mu_0)(\xi+2-\mu_0)\left(-\frac{1}{a}\right)(\alpha\xi+\beta_2).
\end{align}

The automorphism group of the generalised Heun equation has been investigated in \cite{RSCHAFKEDSCHMIDT} by studying the index transformations:
\begin{equation}
y(z)=z^{\sigma_0}(z-1)^{\sigma_1}(z-a)^{\sigma_a}\tilde{y}(z),
\end{equation}
with $\sigma_j\in\{0,\mu_j\},j=(0,1,2) $
and linear transformations of the independent variable
\begin{equation}
\tilde{z}=\varepsilon z+\delta,\;\;\;(\varepsilon,\delta\in \mathbb{C},\varepsilon \not =0 )
\label{lineartransf}
\end{equation}
which map the simple singularities $(0,1,a)$ into the simple singularities $(0,1,\tilde{a})$ and keep the irregular singularity $\infty$ fixed. For a summary of the results see Table 1, appendix A.
Using these results for the automorphism group one can define for each $j=0,1,2$ a set of two Floquet solutions $y_{j1},y_{j2}$ at $z_j$ in terms of the $\eta$ function of theorem \ref{eta} by
\begin{align}
y_{01}(z,\mu,\alpha,\lambda)&:=\eta(z,\mu_0,\mu_1,\mu_2,\alpha,\lambda;a),\\
y_{02}(z,\mu,\alpha,\lambda)&:=z^{\mu_0}\eta(z,-\mu_0,\mu_1,\mu_2,\alpha,\lambda;a),
\end{align}
for $|z|<\min(1,|a|),$
\begin{align}
y_{11}(z,\mu,\alpha,\lambda)&:=\eta(1-z,\mu_1,\mu_0,\mu_2,-\alpha,-\lambda;1-a),\\
y_{12}(z,\mu,\alpha,\lambda)&:=(1-z)^{\mu_1}\eta(1-z,-\mu_1,\mu_0,\mu_2,-\alpha,-\lambda;1-a),
\end{align}
for $|z-1|<\min(1,|a-1|),$
\begin{align}
y_{21}(z,\mu,\alpha,\lambda)&:=\eta\left(1-\frac{z}{a},\mu_1,\mu_2,\mu_0,-a\alpha,-a\lambda;1-\frac{1}{a}\right),\\
y_{22}(z,\mu,\alpha,\lambda)&:=\left(1-\frac{z}{a}\right)^{\mu_2}\eta\left(1-\frac{z}{a},\mu_1,-\mu_2,\mu_0,-a\alpha,-a\lambda;1-\frac{1}{a}\right),
\end{align}
for $|z-a|<\min(|a|,|1-a|)$.
The local solution in the vicinity of $1$ has the expansion:
\begin{align}
y_{11}(z,\mu,\alpha,\lambda)&:=\eta(1-z,\mu_1,\mu_0,\mu_2,-\alpha,-\lambda;1-a),\nonumber \\
&=\sum_{k=0}^{\infty}\frac{T_k(\mu_1,\mu_0,\mu_2,-\alpha,-\lambda;1-a)}{k!\Gamma(1-\mu_1+k)}(1-z)^k,\nonumber \\&=\frac{1}{\Gamma(1-\mu_1)}+\frac{-\tilde{\beta}_0/\tilde{a}}{\Gamma(2-\mu_1)}T_0(1-z)\nonumber \\
&+\Biggl\{\left[2-\mu_1-\mu_0+\frac{2-\mu_1-\mu_2}{\tilde{a}}+\alpha-\frac{\tilde{\beta}_0}{\tilde{a}}\right]T_1\nonumber \\
&-\frac{\tilde{\beta}_1(1-\mu_1)}{\tilde{a}}T_0\Biggr\}\frac{(1-z)^2}{2!\Gamma(3-\mu_1)}+\cdots,
\end{align}
with $T_0(\mu_1,\mu_0,\mu_2,-\alpha,-\lambda;1-a)=1$.

\section{Asymptotic solutions at infinity $r\rightarrow \infty$}\label{AsymptoticsGHEr}
We can also obtain the far horizon limit of our closed form analytic radial solutions as follows. The radial GHE (\ref{GHERADIALKN}) is a differential equation with an irregular singularity at infinity. Following the discovery of Thom\'{e} that such a differential equation can be satisfied in the neighbourhood of an irregular singularity by a series of the form \cite{OLVER}
\begin{equation}
Y=e^{\lambda \zeta}\zeta^{\mu}\sum_{s=0}^{\infty}\frac{a_s}{\zeta^s}
\label{asymptotic}
\end{equation}
we are able to determine the exponential parameters $\lambda,\mu$. Let us start by first inserting the formal solution at infinity eqn(\ref{asymptotic}), into the generic form of the GHE (\ref{SCHASCHMI}). We obtain:
\begin{align}
\lambda_1^{\infty}&=0,\;\mu_1^{\infty}=-\frac{\beta_2}{\alpha},\\
\lambda_2^{\infty}&=-\alpha,\;\mu_2^{\infty}=\mu_0+\mu_1+\mu_2-3+\frac{\beta_2}{\alpha}.
\label{ekthetika}
\end{align}
Applying this method of asymptotic analysis to the radial GHE equation (\ref{GHERADIALKN}) yields the results:
\[GH_{\infty} \equiv \eta^{\infty}(z,\mu,\alpha,\lambda;a)\sim\left\{\begin{array}{c}
z^{-\frac{\beta_2}{\alpha}}=z^{-(\boldsymbol{\mu_1}+\boldsymbol{\mu_2})}z^{-\left(\frac{C_1+C_2+C_3}{\alpha}\right)}\\
e^{-\alpha z}z^{-(\boldsymbol{\mu_1}+\boldsymbol{\mu_2})}z^{\left(\frac{C_1+C_2+C_3}{\alpha}\right)}\end{array}\right.
\]
and in terms of the original variables:
\[R(r)\sim\left\{\begin{array}{c}
e^{\pm i\sqrt{\omega^2 -\mu^2}(r_{+}-r_{-})z}\left(\frac{r-r_{-}}{r_{+}-r_{-}}\right)^{-\frac{C_1+C_2+C_3}{\pm 2i \sqrt{\omega^2 -\mu^2}(r_{+}-r_{-)}}}\\
e^{-(\pm i\sqrt{\omega^2 -\mu^2}(r_{+}-r_{-})z)}\left(\frac{r-r_{-}}{r_{+}-r_{-}}\right)^{\frac{C_1+C_2+C_3}{\pm 2i \sqrt{\omega^2 -\mu^2}(r_{+}-r_{-)}}}\end{array}\right.\]
\subsection{The connection problem for a regular and the irregular singular point at $\infty$ for the radial generalised Heun equation}\label{CONNEIRRREGGLOBAL}

We start this section by writing the formal solutions at $\infty$ eqns (\ref{asymptotic}),(\ref{ekthetika}), of the GHE , in the form
\begin{align}
\tilde{y}_1(z)&=z^{\mu^{\infty}_1}\left(1+\sum_{s=1}^{\infty}a_1(s)z^{-s}\right), \\
\tilde{y}_2(z)&=e^{-\alpha z}z^{\mu^{\infty }_2}\left(1+\sum_{s=1}^{\infty}a_2(s)z^{-s}\right).
\label{krikossyndesi}
\end{align}

More exactly then in general we know for second order equations such as (\ref{SCHASCHMI}) that there exist uniquely determined solutions $y_1(z)$ and $y_2(z)$ defined in \cite{OLVER} (\textsection{7.2})
\begin{align}
S_1&=\Biggl\{z\Bigg||z|>\max(1,|a|),\arg(-\alpha)-\frac{3}{2}\pi<\arg z<\arg(-\alpha)+\frac{3}{2}\pi\Biggr\},\\
S_2&=\Biggl\{z\Bigg||z|>\max(1,|a|),\arg\alpha-\frac{3}{2}\pi<\arg z<\arg\alpha+\frac{3}{2}\pi\Biggr\},
\end{align}
where $\arg(\pm\alpha)$ are chosen in $[-\pi,\pi[$, such that
\begin{equation}
y_j(z)\sim \tilde{y}_j(z)\;\;\;(S_j\ni z\rightarrow \infty,j=1,2).
\end{equation}

From \cite{RSCHAFKEDSCHMIDT} we know that (\ref{SCHASCHMI}) and (\ref{radialKNbef}) have one solution $y(z)$ that is holomorphic near $1$ with $y(z)=1$ if $\mu_1\not\in\mathbb{N}$ (in the form (\ref{SCHASCHMI})), which can be analytically continued to a neighbourhood in $[1,\infty[.$
The corresponding connection problem  \cite{RegIrregSIAMJ} is to decompose the solution $y(z)$ in the form:
\begin{equation}
y(z)=\gamma_1 y_1(z)+\gamma_2 y_2(z),\;\;\;(z\in]1,\infty[,z\;\; \text{sufficiently large}).
\label{problemconnection}
\end{equation}

Following \cite{RegIrregSIAMJ} we shall  transform the connection problem between $1$ and $\infty$ by defining $z=1/(1-t)$.
Indeed, we shall apply the combined transformation
\begin{equation}
z^{-\mu^{\infty}_ 1}y(z)=v(t),\;\;\;z=\frac{1}{1-t}.
\end{equation}
Starting from (\ref{SCHASCHMI}), one proves that the new dependent variable $v(t)$ satisfies the differential equation:
\begin{align}
v^{\prime\prime}(t)&+\left[\frac{\alpha}{(t-1)^2}+\frac{\mu^{\infty }_2-\mu^{\infty}_ 1+2}{t-1}+
\frac{1-\mu_1}{t}+\frac{1-\mu_2}{t-\tilde{a}}\right]v^{\prime}(t)\nonumber \\
&+\frac{\tilde{\beta}_0+\tilde{\beta}_1 t+\tilde{\beta}_2 t^2}{at(t-\tilde{a})(t-1)^2}v(t)=0,
\label{transfGHE}
\end{align}
with
\begin{align}
\tilde{\beta}_0&=(1-a)(1-\mu_1)\mu^{\infty}_ 1+\beta_0+\beta_1+\beta_2, \\
\tilde{\beta}_1&=-\beta_0+\beta_2 a+\mu^{\infty}_1[a(1-\mu_1)+(1-\mu_2)+(1-a)(\mu^{\infty}_ 1-\mu_0)],\\
\tilde{\beta}_2&=a\mu^{\infty}_1(\mu^{\infty}_1-\mu_0),
\end{align}
and $\tilde{a}=1-1/a$. The singular points $z=0,1,a,\infty$ are mapped to $t=\infty,0,1-1/a,1$ respectively.
$y(z)$ is transformed into the solution $v_0(t)$ of (\ref{transfGHE}) holomorphic near $1$ with $v_0(1)=1,$ which is written as
\begin{equation}
v_0(t)=\sum_{s=0}^{\infty}d_s t^s,\;\;\;d_0=1.
\label{connectioninfinity}
\end{equation}
Inserting (\ref{connectioninfinity}) into (\ref{transfGHE}) and equating parts with equal powers of $t$ we obtain recursion relations for the coefficients $d_k$:
\begin{align}
&a\tilde{a}(1-\mu_1)d_1=\tilde{\beta}_0d_0,\\
&2\tilde{a}ad_2(2-\mu_1)=a\Biggl[-\alpha\left(1-\frac{1}{a}\right)+\left(1-\frac{1}{a}\right)[\mu^{\infty} _2-\mu^{\infty}_1+2]+\left(3-\frac{2}{a}\right)(1-\mu_1)\nonumber \\
&+(1-\mu_2)\Biggr]d_1+\tilde{\beta}_0d_1+\tilde{\beta}_1d_0,\\
&3\left(1-\frac{1}{a}\right)\left(3-\mu_1\right)a d_3=2a\left(3-\frac{2}{a}\right)\left(2-\mu_1\right)d_2 \nonumber \\
&+2d_2 a\left\{\left(1-\frac{1}{a}\right)[\mu^{\infty}_2-\mu^{\infty} _1+2-\alpha]+1-\mu_2\right\}+\tilde{\beta}_0d_2\nonumber \\
&+d_1a\Biggl\{-\left(2-\frac{1}{a}\right)(\mu^{\infty}_2-\mu^{\infty} _1+2)+\alpha-\left(3-\frac{1}{a}\right)(1-\mu_1)-2(1-\mu_2)\Biggr\}+\tilde{\beta}_1d_1\nonumber \\
&+\tilde{\beta}_2d_0\\
&4\left(1-\frac{1}{a}\right)\left(4-\mu_1\right)d_4=3\left[
\left(3-\frac{2}{a}\right)(3-\mu_1)+\left(1-\frac{1}{a}\right)[\mu^{\infty} _2-\mu^{\infty}_1+2-\alpha]+(1-\mu_2)\right]d_3\nonumber \\
&+\frac{\tilde{\beta}_0}{a}d_3\nonumber \\
&+2d_2\left[\left(-3+\frac{1}{a}\right)(2-\mu_1)+\left(\frac{1}{a}-2\right)(\mu^{\infty }_2-\mu^{\infty}_1+2)+\alpha-2(1-\mu_2)\right]\nonumber \\
&+\frac{\tilde{\beta}_1}{a}d_2\nonumber \\
&+(1-\mu^{\infty}_1)(\mu_0-\mu^{\infty}_1+1)d_1,\\
&\cdots\nonumber
\end{align}
which can be summarised in the four-term recursion relation for the coefficients $d_s$:
\begin{align}
&\left(1-\frac{1}{a}\right)(s+1)(s+1-\mu_1)d_{s+1}\nonumber \\
&=\varphi_0(s)d_s+\varphi_1(s-1)d_{s-1}+\varphi_2(s-2)d_{s-2},\;\;\;(s\in\mathbb{N}),\\
&d_0=1,\;\;\;d_{-1}=d_{-2}=0,
\label{vier}
\end{align}
where
\begin{align}
\varphi_0(s)&=\left[\left(3-\frac{2}{a}\right)(s-\mu_1)+\left(1-\frac{1}{a}\right)(\mu^{\infty} _2-\mu^{\infty}_1+2-\alpha)+1-\mu_2\right]s+\tilde{\beta}_0/a,\\
\varphi_1(s)&=\left[\left(\frac{1}{a}-3\right)(s-\mu_1)+\alpha+\left(\frac{1}{a}-2\right)
(\mu^{\infty}_2-\mu^{\infty}_1+2)-2(1-\mu_2)\right]s+\tilde{\beta}_1/a,\\
\varphi_2(s)&=(s+\mu_0-\mu^{\infty}_1)(s-\mu^{\infty}_1).
\end{align}
Then $y_1(z)$ and $y_2(z)$ pass into local solutions $v_1^{+}(t)$ and $v_2^{+}(t)$ of (\ref{transfGHE}) at the irregular singular point 1:
\begin{align}
v_1^{+}(t)&\sim 1+\sum_{\nu=1}^{\infty}a_1(\nu)(1-t)^{\nu},\;\;\left|\arg(1-t)+\arg(-\alpha)\right|<\frac{3}{2}\pi,\\
v_2^{+}(t)&\sim e^{-\frac{\alpha}{1-t}}(1-t)^{\mu_{\infty 1}-\mu_{\infty 2}}
\left(1+\sum_{\nu=1}^{\infty}a_2(\nu)(1-t)^{\nu}\right)\nonumber \\
&\left|\arg(1-t)+\arg\alpha\right|<\frac{3}{2}\pi,\;\;t\rightarrow 1
\end{align}

Thus, the connection relation (\ref{problemconnection}) corresponds to:
\begin{equation}
v_0(t)=\gamma_1v_1^{+}(z)+\gamma_2 v_2^{+}(z)\;\;\;(z\in]0,1[).
\end{equation}
The following theorem, which was proven in \cite{RegIrregSIAMJ}, supplies a limit formula for the connection coefficient $\gamma_2$:
\begin{theorem}
Let $\gamma_2$ be the connection coefficient for the connection problem (\ref{krikossyndesi})-(\ref{problemconnection}) . Assume the coefficient $\mu_1$ in (\ref{SCHASCHMI}) is not a positive integer e.g. that $\mu_1\not\in\{1,2,3,\cdots\}$, $\alpha\not\in[0,\infty[$ and $\Re a<\frac{1}{2}$. Let the sequence $d_k$ be determined by the four-term recursion formula (\ref{vier}). Then
\begin{equation}
\gamma_2=2\sqrt{\pi}e^{\alpha/2}(-\alpha)^{(\mu_{\infty 2}-\mu_{\infty 1}-1)/2}\lim_{k\rightarrow\infty}
exp(-2\sqrt{-\alpha}\sqrt{k})k^{3/4+(\mu_{\infty 2}-\mu_{\infty 1})/2}d_k.
\end{equation}
Moreover the above sequence with limit $\gamma_2$ has an asymptotic expansion, involving powers of $1/\sqrt{k}$.
\end{theorem}
As it is explained in \cite{RegIrregSIAMJ} all connection coefficients between $1$ and $\infty$ can be computed once a formula for $\gamma_2$ is known. The necessity of the conditions on the GHE parameters in the above theorem is also discussed in \cite{RegIrregSIAMJ}.

\section{The near event horizon limit $r\rightarrow r_{+}$.}\label{Gargantua}

From theorem \ref{eta} and using $\zeta=\frac{r-r_{+}}{r_{-}-r_{+}}$ as the independent variable we know the expansion of the local solution near the event horizon limit where $r\rightarrow r_{+}\Leftrightarrow \zeta\rightarrow 0$:
\begin{align}
\eta(\zeta,\mu,\alpha,\lambda;a)&=\sum_{k=0}^{\infty}\frac{T_k(\mu,\alpha,\lambda;a)}{\Gamma(k+1-\mu_0)k!}\zeta^k\nonumber \\
&=\frac{1}{\Gamma(1-\mu_0)}+\frac{-\beta_0/a}{\Gamma(2-\mu_0)}T_0\zeta\nonumber \\
&+\Biggl\{\left[\frac{2-\mu_0-\mu_2}{a}+(2-\mu_0-\mu_1)-\alpha-\frac{\beta_0}{a}\right]T_1\nonumber \\
&-\frac{\beta_1(1-\mu_0)}{a}T_0\Biggr\}\frac{\zeta^2}{\Gamma(3-\mu_0)2!}+\cdots
\end{align}
From the index transformation:
\begin{equation}
R(\zeta)=e^{\nu\left(\frac{r-r_{+}}{r_{-}-r_{+}}\right)}\zeta^{\boldsymbol{\mu_1}^{\prime}}(\zeta-1)^{\boldsymbol{\mu_2}^{\prime}}\eta(\mu^{\prime},\alpha,\lambda^{\prime};\zeta),
\label{RadialSpinordecay}
\end{equation}
where $\nu=\pm i\sqrt{\omega^2-\mu^2}(r_{+}-r_{-}),\;\alpha=2\nu$, we conclude that in the near event horizon limit in terms of the original variables:
\begin{equation}
R(r)\sim \left(\frac{r-r_{+}}{r_{-}-r_{+}}\right)^{\boldsymbol{\mu_1}^{\prime}}\frac{1}{\Gamma(1-\mu_0)}.
\label{eventhorizonorio}
\end{equation}

\section{Asymptotic behaviour using the Regge-Wheeler coordinate}
The asymptotic forms of the radial equation (\ref{raddira}) may also be seen if equation (\ref{raddira}) is transformed by writing
\begin{equation}
R=(\Delta^{KN})^{1/4}(r^2+a^2)^{-1/2}(\lambda+i\mu r)^{1/2}\mathcal{R},
\end{equation}
and if one uses a Regge-Wheeler-like radial variable $r*$ \cite{RegWheeTortoise}, defined by
\begin{equation}
\frac{{\rm d}r^*}{{\rm d}r}=\frac{(r^2+a^2)}{\Delta^{KN}}.
\end{equation}
The radial equation is then:
\begin{align}
\frac{{\rm d}^2 \mathcal{R}}{{\rm d}r^{*2}}&+\Biggl[\frac{(K^2+iK(r-M)+\Delta^{KN}(+\frac{-\mu K}{\lambda+i\mu r}-2i\omega r-ieq-\mu^2 r^2-\lambda^2))}{(r^2+a^2)^2}\nonumber\\
&-\mathcal{G}^2-\frac{{\rm d}\mathcal{G}}{{\rm d}r^*}+\frac{i\mu\Delta^{KN}}{2(\lambda+i\mu r)(r^2+a^2)^2}[r-M-\frac{3}{2}\frac{i\mu\Delta^{KN}}{\lambda+i\mu r}]\Biggr]\mathcal{R}=0,
\label{ReggeWheRDIRAC}
\end{align}
with
\begin{equation}
\mathcal{G}:=-\frac{1}{2}\frac{(r-M)}{r^2+a^2}+\frac{r\Delta^{KN}}{(r^2+a^2)^2}.
\end{equation}
\subsubsection{The near event horizon limit $r\rightarrow r_{+}$ }
Asymptotically, in the near horizon limit $r\rightarrow r_{+}\;(r^*\rightarrow -\infty)$ equation (\ref{ReggeWheRDIRAC}), takes the form
\begin{equation}
\frac{{\rm d}^2\mathcal{R}}{{\rm d}r^{*2}}+\left[\left(\omega-m\omega_{+}+\frac{eqr_{+}}{r_{+}^2+a^2}\right)+\frac{i}{2}\frac{r_{+}-M}{r_+^2+a^2}\right]^2\mathcal{R}\sim 0,
\end{equation}
where $\omega_+:=\frac{a}{r_+^2+a^2}$ and $r_+$ is the event horizon of the KN black hole $r_+=M+\sqrt{M^2-a^2-e^2}$.
The solution for a purely ingoing wave at the event horizon of the KN black hole is
\begin{equation}
\mathcal{R}\sim e^{-ir^*\left(\omega-m\omega_{+}+\frac{eqr_+}{r_+^2+a^2}\right)}
\end{equation}
\subsection{Asymptotic solutions at spatial infinity $r\rightarrow \infty$}
At spatial infinity $r\rightarrow \infty \;(r^* \rightarrow \infty)$ the asymptotic form of (\ref{ReggeWheRDIRAC}) is:
\begin{equation}
\frac{{\rm d}^2 \mathcal{R}}{{\rm d} r^{*2}}+\left(\omega^2-\mu^2+\frac{2eq\omega+2\mu^2 M}{r}\right)\mathcal{R}\sim 0
\end{equation}

\section{Determining the separation constant $\lambda$ of the massive Dirac equation in the Kerr-Newman spacetime}
\label{statheradiaxlambda}

In \textsection{\ref{anguGHE}} we computed the exact  eigenfunctions of the angular components of the Dirac equation in the KN curved background in terms of the generalised Heun functions.
We now discuss how the separation constant $\lambda$ (i.e the angular eigenvalues)\footnote{Some notable works on the determination of the angular eigenvalues are: \cite{HSchmidt},\cite{MonikaDB},\cite{Neznamov},\cite{SRDolan}.} can be determined. For this purpose we will make original use of functional analysis techniques, of properties of special functions, integration of partial differential equations, novel asymptotic analysis and integration of the non-linear Painlev\'{e} transcendents involved in this fascinating problem.
Equations (\ref{cornerKN1})-(\ref{corner2KNbh}) in matrix form read as follows:
\begin{equation}
\left[\begin{array}{cc}
\frac{{\rm d}}{{\rm d}\theta}+\frac{m}{\sin\theta}-\omega a\sin\theta+\frac{\cot\theta}{2} & \lambda -\mu a \cos\theta\\
\lambda+a\mu\cos\theta & -\frac{{\rm d}}{{\rm d}\theta}-\frac{\cot\theta}{2} -\omega a\sin\theta+\frac{m}{\sin\theta}\end{array}\right]\left[\begin{array}{c}
S^{(+)}(\theta) \\
S^{(-)}(\theta)\end{array}\right]
=0.
\end{equation}
Using the matrix $V=\left[\begin{array}{cc} 0&1\\
-1&0\end{array}\right]$ we can write eqns (\ref{cornerKN1})-(\ref{corner2KNbh}) as an eigenvalue equation
\begin{gather}
\frac{1}{2\sin\theta}\Biggl\{2\left[\begin{array}{cc} 0&\sin\theta\\
-\sin\theta&0\end{array}\right]\left[\begin{array}{c} {\rm d}S^{(+)}/{{\rm d}\theta}\\
 {\rm d}S^{(-)}/{{\rm d}\theta}\end{array}\right]+\left[\begin{array}{cc} 0&\cos\theta\\
-\cos\theta&0\end{array}\right]\left[\begin{array}{c} S^{(+)}(\theta)\\
 S^{(-)}(\theta)\end{array}\right]
 \nonumber \\
 \left[\begin{array}{cc} -2\nu \cos\theta\sin\theta&-2m+2\xi\sin^2\theta\\
-2m+2\xi\sin^2\theta&2\nu\cos\theta\sin\theta\end{array}\right]\left[\begin{array}{c} S^{(+)}(\theta)\\
 S^{(-)}(\theta)\end{array}\right]\Biggr\}=\lambda\left[\begin{array}{c} S^{(+)}(\theta)\\
 S^{(-)}(\theta)\end{array}\right],
 \label{HamiltonAngleKN1}
 \end{gather}
 where we define $\xi:=a\omega,\;\nu:=a\mu$.

 For fixed values of $\xi$ and $\nu$, the differential operator $A$ generated by the left hand side of (\ref{HamiltonAngleKN1}) is a self-adjoint operator acting on the Hilbert $L^2((0,\pi),2\sin\theta)^2$ of square integrable vector functions with respect to the weight function $2\sin\theta$.

 The operator $A=A(\xi,\nu)$ as well as its eigenvalues $\lambda_j=\lambda_j(m,\xi,\nu)$ depend holomorphically on $\xi,\nu$.  The partial derivatives of the differential operator $A$ are given by:
 \begin{equation}
 \frac{\partial A}{\partial \nu}=\left[\begin{array}{cc} -\cos\theta&0\\
0&\cos\theta\end{array}\right],\;\;\;\frac{\partial A}{\partial \xi}=\left[\begin{array}{cc} 0&\sin\theta\\
\sin\theta&0\end{array}\right].
\end{equation}
Using analytic perturbation theory, in particular Theorem 3.6, VII, \textsection 3 in \cite{KatoLPT}, we obtain for the eigensystem (\ref{HamiltonAngleKN1}) the following estimates:
\begin{equation}
\Bigl|\frac{\partial \lambda_j}{\partial \nu}\Bigr|\leq \left\Vert \frac{\partial A}{\partial \nu}\right\Vert=\vert\cos\theta\vert\leq 1
\;\;\text{and}\;\; \Bigl|\frac{\partial \lambda_j}{\partial \xi}\Bigr|\leq  \left\Vert \frac{\partial A}{\partial \xi}\right\Vert=\sin\theta\leq 1
\label{estimatesgonia}
\end{equation}
In (\ref{estimatesgonia}) the function $\left\Vert\cdot\right\Vert$ denotes the \textit{operator} or (\textit{spectral}) norm of a matrix:
\begin{definition}
Let $B(X,Y)$ denote the space of continuous linear operators $X\rightarrow Y$, where $X$ and $Y$ are normed spaces. Then $B(X,Y)$ is a vector space with a norm defined by \cite{Conway}: \footnote{The norm is well defined in the sense that if $\mathcal{T}$ is an operator, then $\left\Vert \mathcal{T} x\right\Vert / \left\Vert x\right\Vert\leq b$ for all non-zero $x\in X$, and the supremum $\left\Vert \mathcal{T}\right\Vert$ of such upper bounds $b$ exists. In fact, a linear map belongs to $B(X,Y)$ if, and only if, $\left\Vert \mathcal{T}\right\Vert<\infty$, in which case $\left\Vert \mathcal{T} x \right\Vert\leq \left\Vert \mathcal{T}\right\Vert \left\Vert x\right\Vert$.  }
\begin{equation}
\left\Vert A\right\Vert:=\sup_{x\not =0}\frac{\left\Vert Ax \right\Vert_Y}{\left\Vert x\right\Vert_X}.
\end{equation}
\end{definition}
It can be shown that $B(X,Y)$ is complete when $Y$ is complete.
For a matrix $A$ the operator norm is related to the spectral radius:
\begin{theorem}
\begin{equation}
\left\Vert A\right\Vert=\rho\sqrt{(A^*A)}
\end{equation}
\end{theorem}
Thus for a real $2\times2$ matrix $A$ we compute:
\begin{corollary}
\begin{equation}
\left\Vert\left[\begin{array}{cc} a&b\\
c&d\end{array}\right] \right\Vert=\sqrt{\frac{a^2+b^2+c^2+d^2+\sqrt{((a-d)^2+(b+c)^2)((a+d)^2+(b-c)^2)}}{2}}.
\end{equation}
\end{corollary}

 Alternatively we can express the angular equations in the following manner:
 \begin{equation}
 \left[\begin{array}{c} {\rm d}\mathcal{S}^{(+)}(\theta)/{\rm d}\theta\\
 -{\rm d}\mathcal{S}^{(-)}(\theta)/{\rm d}\theta\end{array}\right]+\left[\begin{array}{cc} -\nu \cos\theta&-\frac{m}{\sin\theta}+\xi\sin\theta\\
-\frac{m}{\sin\theta}+\xi\sin\theta&\nu\cos\theta\end{array}\right]\left[\begin{array}{c} \mathcal{S}^{(+)}(\theta)\\
 \mathcal{S}^{(-)}(\theta)\end{array}\right]
=\lambda\left[\begin{array}{c} \mathcal{S}^{(+)}(\theta)\\
 \mathcal{S}^{(-)}(\theta)\end{array}\right]
 \label{traseigenl}
 \end{equation}
 where $\mathcal{S}(\theta):=\sqrt{\sin\theta}\left[\begin{array}{c} S^{(+)}(\theta)\\
 S^{(-)}(\theta)\end{array}\right]$.
 We will first determine the eigenvalues of (\ref{traseigenl}) for $a=0$ (non-rotating black hole). For this we will need the following lemmas concerning Jacobi polynomials.
 \begin{lemma}\label{JacOne}
 \begin{equation}
 \frac{1+x}{2}P_{\nu-1}^{(\alpha+1,\beta+1)}(x)=\frac{\nu}{2\nu+1+\alpha+\beta}P_{\nu}^{(\alpha+1,\beta)}(x)+
 \frac{\beta+\nu}{2\nu+1+\alpha+\beta}P_{\nu-1}^{(\alpha+1,\beta)}(x)
 \end{equation}
 \begin{proof}
 We will prove it using the relation of Jacobi polynomials to hypergeometric functions. Indeed we have:
 \begin{equation}
 P_{\nu}^{(\alpha,\beta)}(\cos\theta)=\frac{(\alpha+1)_{\nu}}{\nu!}F(\alpha+\beta+1+\nu,-\nu,\alpha+1,\sin^2\frac{\theta}{2})
 \label{defJacGauss}
 \end{equation}
 Now the Gau$\ss$ hypergeometric function $F(a,b,c,x)$ obeys the recurrence  relation:
 \begin{equation}
 (b-a)(1-x)F=(c-a)F(a-)-(c-b)F(b-),
 \label{contiguous1}
 \end{equation}
 where we use the notation $F\equiv F(a,b,c,x),\;\;F(a\pm)=F(a\pm 1,b,c,x)$ and likewise for $F(b\pm),F(c\pm)$. Using (\ref{defJacGauss}),(\ref{contiguous1}), the lemma is proved.
\end{proof}
 \end{lemma}
 \begin{lemma}\label{Jac2}
 \begin{equation}
 \frac{\nu(\alpha+\beta+1+\nu)}{2\nu+1+\alpha+\beta}P_{\nu}^{(\alpha+1,\beta)}(x)=-\frac{\nu(\alpha+\nu+1)}{2\nu+1+\alpha+\beta}P_{\nu-1}^{(\alpha+1,\beta)}(x)+
\nu P_{\nu}^{(\alpha+1,\beta-1)}(x)
\end{equation}
\end{lemma}
\begin{proof}
Using the recurrence relation for the Gau$\ss$ hypergeometric function:
\begin{equation}
(a-b)F=aF(a+)-bF(b+),
\end{equation}
the lemma is proved.
\end{proof}
\begin{lemma}\label{Jac3}
\begin{equation}
-P_{\nu-1}^{(\alpha+1,\beta)}(x)-P_{\nu}^{(\alpha,\beta)}(x)=-P_{\nu}^{(\alpha+1,\beta-1)}(x)
\end{equation}
\end{lemma}
\begin{proof}
Using (\ref{defJacGauss}) and the following recurrence relation for the Gau$\ss$ hypergeometric function:
\begin{equation}
b[F(b+)-F]=(c-1)[F(c-)-F]
\end{equation}
the lemma is proved.
\end{proof}

The result follows directly from Lemmas \ref{Jac2} and \ref{Jac3}:
\begin{proposition}\label{JacGustavJacob}
\begin{equation}
\frac{\nu+\alpha+\beta+1}{2\nu+1+\alpha+\beta}[\nu P_{\nu}^{(\alpha+1,\beta)}(x)+(\beta+\nu)P_{\nu-1}^{(\alpha+1,\beta)}(x)]=
(\nu+\beta)P_{\nu}^{(\alpha+1,\beta-1)}(x)-\beta P_{\nu}^{(\alpha,\beta)}(x).
\end{equation}
\end{proposition}

\begin{proposition}
For $a=0$ the eigenvalues of the separation constant $\lambda$ in Eq.(\ref{HamiltonAngleKN1}) are given by the formula:
\begin{equation}
\lambda_{\nu}=\pm[\nu+\frac{1}{2}+m]
\end{equation}
\end{proposition}
Applying the transformation:
\begin{equation}
\mathcal{S}(\theta)=\sin^{m+1/2}\theta\left[\begin{array}{c}\sqrt{\tan\frac{\theta}{2}}u(\cos\theta)\\
\sqrt{\cot\frac{\theta}{2}}v(\cos\theta)\end{array}\right],
\end{equation}
yields the equations:
\begin{align}
(1-x)\frac{{\rm d}u(x)}{{\rm d}x}&=(m+\frac{1}{2})u(x)+\lambda v(x)\label{uJac1},\\
(1+x)\frac{{\rm d}v(x)}{{\rm d}x}&=-(m+\frac{1}{2})v(x)-\lambda u(x)\label{vJac2}
\end{align}
It can be shown that $\lambda=0$ is not an eigenvalue because it does not satisfy the normalization conditions:
\begin{equation}
\int_{-1}^1(1-x)^{m+1/2}(1+x)^{m-1/2}u^2(x){\rm d}x<\infty,\;\;\int_{-1}^{+1}(1+x)^{m+1/2}(1-x)^{m-1/2}v^2(x){\rm d}x<\infty.
\end{equation}
Thus, we can write (\ref{vJac2}) as follows:
\begin{equation}
u(x)=\frac{-(m+1/2)v(x)-(1+x)v^{\prime}(x)}{\lambda},
\label{endiamensoJ}
\end{equation}
and substituting for $u(x)$ in (\ref{uJac1}) yields
\begin{equation}
(1-x^2)v^{\prime\prime}(x)+v(x)\left[\lambda^2-(m+\frac{1}{2})^2\right]+v^{\prime}(x)[1-2(m+1)x]=0.
\end{equation}
If we compare the last equation with the equation for the Jacobi polynomials:
\begin{equation}
(1-x^2)y^{\prime\prime}(x)+[\beta-\alpha-(\alpha+\beta+2)x]y^{\prime}(x)+\nu(\nu+\alpha+\beta+1)y(x)=0,
\end{equation}
we are lead to the identification:
$\beta=m+\frac{1}{2},\;\alpha=m-\frac{1}{2}$ and the order of the Jacobi polynomial is determined by the equation:
\begin{equation}
\nu^2+\nu(2m+1)=\lambda^2-\left(m+\frac{1}{2}\right)^2,
\end{equation}
which yields:
\begin{equation}
\nu=\lambda-(m+1/2),
\end{equation}
or equivalently $\lambda_{\nu}=\pm\left[\nu+\frac{1}{2}+m\right]$. Defining $v(x)=-P_{(\lambda-m-1/2)}^{(m-\frac{1}{2},m+\frac{1}{2})}(x)$ and using the derivative of the Jacobi polynomial:
\begin{equation}
(P_{\nu}^{(\alpha,\beta)}(x))^{\prime}=\frac{1}{2}(\nu+\alpha+\beta+1)P_{\nu-1}^{(\alpha+1,\beta+1)}(x),
\end{equation}
in equation (\ref{endiamensoJ}) we obtain:
\begin{equation}
\lambda_{\nu}u(x)=\beta P_{\nu}^{(\alpha,\beta)}(x)+\frac{1+x}{2}(\nu+\alpha+\beta+1)P_{\nu-1}^{(\alpha+1,\beta+1)}(x).
\end{equation}
Then using Lemmas \ref{JacOne}-\ref{Jac3} and Proposition \ref{JacGustavJacob} we obtain the result:
\begin{proposition}
\begin{align}
\lambda_{\nu}u(x)&=(\nu+\beta)P_{\nu}^{(\alpha+1,\beta-1)}(x)=|\lambda_{\nu}|P_{\nu}^{(\alpha+1,\beta-1)}(x)\nonumber \\
&\Leftrightarrow u(x)=\pm P_{\nu}^{(\alpha+1,\beta-1)}(x),\;\;\;x\in(-1,1).
\end{align}
\end{proposition}

As we shall see in \textsection \ref{pdeeigenl} the eigenvalues of the angular equation in KN spacetime for different generic values of the parameters  $\nu,\xi$ obey a partial differential equation whose solution using the method of Charpit
is reduced to the solution of the nonlinear ODE of Painlev\'{e} ${\rm P_{III}}$.
The theory of solutions of the Painlev\'{e} ${\rm P_{III}}$ constitutes a very active field of research in mathematical analysis \cite{Clarkson}-\cite{Milne}. In \textsection{\ref{RicattiPIIIRational}} we obtain new exact solutions of the angular Painlev\'{e} ${\rm P_{III}}$  in terms of Bessels functions.

In  \textsection \ref{PIIIasymptotic} we present a new asymptotic analysis for the Painlev\'{e} ${\rm P_{III}}$ which describes the leading order behaviour of eigenvalues of the angular equation in KN spacetime for generic values of the physical parameters $\nu$ and $\xi$. In particular we shall derive in the large limit of the independent variable a closed form solution for the eigenvalues of the angular equation in terms of Jacobian elliptic functions \footnote{In \cite{solitons}, \textsection 7.5, the authors  discuss the asymptotics of solutions for the first two Painlev\'{e} transcendents ${\rm P_{I}},{\rm P_{II}}$. As is mentioned there,  Boutroux initiated such a research programme by studying the asymptotics of solutions of the first Painlev\'{e} differential equation ${\rm P_{I}}:\;\frac{{\rm d}^2 y}{{\rm d}x^2}=6y^2+x,$ in two memoirs and he obtained for the first Painlev\'{e} transcendent $y(x)\sim x^{1/2} \wp\left(\frac{4}{5}x^{5/4};-2,g_3\right), \text{as} |x|\rightarrow\; \infty$ where $\wp(x;g_2,g_3)$ denotes the Weierstra$\ss$ elliptic function \cite{Boutroux}.}.

A very important property of non-linear differential equations such as  $\rm{P_{III}}$ is that they can be considered in the framework of the inverse problem and can be represented as compatibility conditions for a system of auxiliary linear problems.
A particular linear representation of the nonlinear ODE ${\rm P_{III}}$ by a Lax pair was provided in \cite{Jimbo1} and is discussed in Appendix \ref{RiemannHilbertRep} (see eqns.(\ref{Lax1})-(\ref{Lax2})). In the same appendix we briefly discuss how this coupled linear system can be solved by the Inverse Monodromy Tranform and obtain rational solutions for the Painlev\'{e} ${\rm P_{III}}$.

\subsection{Reducing the partial differential equation satisfied by the eigenvalues to  ${\rm P_{III}}$ }\label{pdeeigenl}

In appendix \ref{pdePIIICharpit}, it is proved that the equation  the eigenvalues obey for different parameters $\nu,\xi$ is:
\begin{equation}
(\nu+2\lambda \xi)\frac{\partial \lambda}{\partial \nu}+(\xi+2\lambda \nu)\frac{\partial \lambda}{\partial \xi}+
2m\nu-2\xi\nu=0
\label{eigenvalueeqn}
\end{equation}
Introducing as in \cite{MonikaDB} new coordinates:
\begin{equation}
\nu(t,v)=\frac{t}{2}\left(v+\frac{\epsilon}{v}\right),\;\;\xi(t,v)=\frac{t}{2}\left(v-\frac{\epsilon}{v}\right),\;\;\epsilon\in
\left\{-1,+1\right\},
\end{equation}
and also $w(t,v)=\lambda(\nu,\xi)$ we compute:
\begin{align}
\frac{\partial w}{\partial t}&=\frac{\partial \nu}{\partial t}\frac{\partial \lambda}{\partial \nu}+\frac{\partial \xi}{\partial t}\frac{\partial \lambda}{\partial \xi} \nonumber \\
&=\frac{1}{2}\left(v+\frac{\epsilon}{v}\right)\frac{\partial \lambda}{\partial \nu}+\frac{1}{2}\left(v-\frac{\epsilon}{v}\right)\frac{\partial \lambda}{\partial \xi}\nonumber \\
&=\frac{1}{t}\left(\nu\frac{\partial \lambda}{\partial \nu}+\xi\frac{\partial \lambda}{\partial \xi}\right)
\end{align}
and
\begin{align}
\frac{\partial w}{\partial v}&=\frac{\partial \nu}{\partial v}\frac{\partial \lambda}{\partial \nu}+
\frac{\partial \xi}{\partial v}\frac{\partial \lambda}{\partial \xi} \nonumber \\
&=\frac{t}{2}\left(1-\frac{\epsilon}{v^2}\right)\frac{\partial \lambda}{\partial \nu}+\frac{t}{2}\left(1+\frac{\epsilon}{v^2}\right)\frac{\partial \lambda}{\partial \xi}\nonumber \\
&=\frac{1}{v}\left(\xi\frac{\partial \lambda}{\partial \nu}+\nu\frac{\partial \lambda}{\partial \xi}\right)
\end{align}
In the new coordinates equation (\ref{eigenvalueeqn}) becomes:
\begin{equation}
\frac{\partial w}{\partial t}+\frac{2vw}{t}\frac{\partial w}{\partial v}+m\left(v+\frac{\epsilon}{v}\right)-\frac{t}{2}\left(v^2-\frac{\epsilon^2}{v^2}\right)=0.
\label{partialcharpit}
\end{equation}
In order to solve equation (\ref{partialcharpit}) we will use the method of Charpit \footnote{Paul Charpit read his paper: \textit{M\'{e}moir sur l'int\'{e}grattion des \'{e}quations aux diff\'{e}rences partielles} to the Acad\'{e}mie des Sciences on Wednesday, June 30,1784. However, it appears the paper was never printed \cite{GrattamGuiness}.}. For a pde in the form:
\begin{equation}
F(x,y,u,p,q)=0,\;\;p=\frac{\partial u}{\partial x},\;q=\frac{\partial u}{\partial y},
\label{pdeCharachter}
\end{equation}
where $u(x,y)$ is a general solution of (\ref{pdeCharachter}), we have the following system of differentials:
\begin{equation}
\frac{{\rm d}x}{F_p}=\frac{{\rm d}y}{F_q}=\frac{{\rm d}u}{pF_p+qF_q}=\frac{{\rm d}p}{-F_x-pF_u}=\frac{{\rm d}q}{-f_y-qF_u}
\end{equation}
Applying this method to (\ref{partialcharpit}) we obtain:
\begin{align}
\dot{t}&=\frac{\partial F}{\partial(\partial w/\partial t)}=1,\\
\dot{v}&=\frac{\partial F}{\partial(\frac{\partial w}{\partial v})}=\frac{2vw}{t}\label{eigenasymptot}\\
\dot{w}&=\frac{\partial w}{\partial t}1+\frac{2vw}{t}\frac{\partial w}{\partial v}=-m\left(v+\frac{\epsilon}{v}\right)
+\frac{t}{2}\left(v^2-\frac{1}{v^2}\right).
\end{align}
The procedure yields:
\begin{equation}
\frac{\ddot{v}t}{2v}-\frac{\dot{v}^2t}{2v^2}+\frac{\dot{v}}{2v}=-m\left(v+\frac{\epsilon}{v}\right)
+\frac{t}{2}\left(v^2-\frac{1}{v^2}\right)
\end{equation}
Equivalently:
\begin{equation}
\ddot{v}-\frac{\dot{v}^2}{v}+\frac{\dot{v}}{t}+\frac{2m}{t}(v^2+\epsilon)-\left(v^3-\frac{1}{v}\right)=0.
\label{Painlevetria}
\end{equation}
This is Painlev\'{e}'s third equation, $\rm{P_{III}}$ (see eqn.(\ref{PainleveDREI}) in Appendix \ref{RiemannHilbertRep}) with parameters $\alpha=-2m,\beta=-2m\epsilon,\gamma=1,\delta=-1$.

\subsection{Rational solutions of the Painlev\'{e} ${\rm P_{III}}$ and one parameter solutions in terms of Bessels functions}\label{RicattiPIIIRational}

It is known from the work of Lukashevich \cite{Lukashevich} that the Painlev\'{e} ${\rm P_{III}}$ can be reduced to the Ricatti equation:
\begin{equation}
\frac{{\rm d}w}{{\rm d}z}=P_2(z) w^2+P_1(z) w+P_0(z),
\label{ricattiPIII}
\end{equation}
by imposing certain conditions on its parameters.
Indeed, differentiation of (\ref{ricattiPIII}) yields:
\begin{equation}
\frac{{\rm d}^2w}{{\rm d}z^2}=2P_2^2 w^3+(P_2^{\prime}+3P_1P_2)w^2+w(2P_2P_0+P_1^{\prime}+P_1^2)+P_1P_0+P_0^{\prime}
\label{ricattidouble}
\end{equation}
Substituting (\ref{ricattiPIII}) and (\ref{ricattidouble}) into the Painlev\'{e} ${\rm P_{III}}$ (\ref{PainleveDREI}) and comparing coefficients of equal powers of $w$ we obtain the following restrictions on the parameters of ${\rm P_{III}}$:
\begin{align}
P_2^2&=\gamma,\;\;P_1(z)=\frac{\alpha-P_2}{P_2}\frac{1}{z},\\
\beta&=\frac{P_0}{P_2}(-\alpha+2P_2),\;\;\delta+P_0^2=0,
\label{sinthikesRic}
\end{align}
so that it reduces to a Ricatti equation.
Now using the transformation $w=-\frac{u^{\prime}}{P_2 u}$ we obtain:
\begin{equation}
u_{zz}-P_1 u_z+P_0P_2u=0\Leftrightarrow u_{zz}+\left(1+\frac{\alpha}{\sqrt{\gamma}}\right)\frac{u_z}{z}+\sqrt{\gamma}\sqrt{-\delta}u=0.
\label{besselHumboldt}
\end{equation}
Applying the transformations:
\begin{equation}
x=\gamma^{1/4}(-\delta)^{1/4}z,\;u=z^{-\varrho}u_1(x),\;\varrho=\alpha/(2\gamma^{1/2}),
\end{equation}
yields the Bessel equation:
\begin{equation}
u_1^{\prime\prime}+\frac{1}{x}u_1^{\prime}+\left(1-\frac{\varrho^2}{x^2}\right)u_1=0.
\end{equation}
Thus the general solution of (\ref{besselHumboldt}) is written as follows:
\begin{equation}
u=Az^{-\varrho}J_{\varrho}(x)+Bz^{-\varrho}Y_{\varrho}(x),
\label{AstroBessel}
\end{equation}
where $J_{\varrho}(x)$ and $Y_{\varrho}(x)$ are Bessel functions. Thus $\rm{P_{III}}$ possesses solutions that are expressed in terms of Bessel functions. As is pointed out in \cite{Milne} the freedom in the choice of sign of the square roots of $\gamma$ and $-\delta$ means that the one-parameter family condition may be regarded as four individual conditions. Indeed, the one-parameter family condition for $\sqrt{\gamma}=\sqrt{-\delta}=1,$ taking into account the third constraint in (\ref{sinthikesRic}) yields $\beta+\alpha+2=0$. This solution is referred to as $y_0^{[1]}(x;\alpha_0,\beta_0,1,-1)$ in \cite{Milne}. For the choice $\gamma^{1/2}=(-\delta)^{1/2}=-1$, i.e. $2-\alpha-\beta=0$, the associated solution is referred as $y_0^{[4]}(x;\alpha_0,\beta_0,1,-1)$ in \cite{Milne}.

These cases are directly applicable to the $\rm{P_{III}}$ equation (\ref{Painlevetria}) in a KN  background spacetime.  If we choose the values $m=\frac{1}{2},\epsilon=+1$ which give:\; $\alpha=\beta=-1$ we get the first solution while if we choose $m=-\frac{1}{2},\epsilon=+1$ i.e. $\alpha=\beta=1$ we get the fourth. Explicitly, in this case the solution of (\ref{Painlevetria}) is expressed in terms of Bessel functions:
\begin{equation}
v=\frac{u_z}{\sqrt{\gamma}u},
\label{recursiveBessel}
\end{equation}
where $u(z)$ satisfies (\ref{AstroBessel}).
In computing (\ref{recursiveBessel}) it is convenient to use the recursive relations for the derivatives of Bessel functions:
\begin{equation}
zJ_{n}^{\prime}=nJ_n(z)-z J_{n+1}(z).
\end{equation}

\subsection{Asymptotic behaviour of the Painlev\'{e} ${\rm P_{III}}$ in terms of Jacobi's elliptic function ${\rm sn}(x)$}\label{PIIIasymptotic}
In the limit $|t|\rightarrow \infty$ eqn.(\ref{Painlevetria}) becomes:
\begin{equation}
v_{tt}=\frac{v_t^2}{v}-\frac{1}{v}+v^3.
\end{equation}

We now prove the following proposition:

\begin{proposition}\label{ellipticfunctionsn1}
A first integral of the differential equation $v_{tt}=\frac{v_t^2}{v}+\frac{\delta}{v}+\gamma v^3$ is:
\begin{equation}
v_t^2=\gamma v^4-\delta+2\mathcal{E}_{\rm{III}}v^2,
\end{equation}
where $\mathcal{E}_{\rm{III}}$ is a constant of integration.
\end{proposition}
\begin{proof}
We have that:
\begin{align}
&v^2\frac{{\rm d}}{{\rm d}t}\left(\frac{v_t}{v}\right)=\delta+\gamma v^4 \Leftrightarrow\nonumber \\
&v^{-3}v^2\frac{{\rm d}}{{\rm d}t}\left(\frac{v_t}{v}\right)=\gamma v+\frac{\delta}{v^3}\Leftrightarrow\nonumber \\
&v^{-1}\frac{{\rm d}}{{\rm d}t}\left(\frac{v_t}{v}\right)=\frac{{\rm d}}{{\rm d}v}\left[\gamma\frac{v^2}{2}-\frac{1}{2}\frac{\delta}{v^2}\right]\Leftrightarrow\nonumber \\
&\frac{v_t}{v}{\rm d}\left(\frac{v_t}{v}\right)={\rm d}[\gamma\frac{v^2}{2}-\frac{\delta}{2v^2}]\Leftrightarrow\nonumber \\
&\int\frac{v_t}{v}{\rm d}\left(\frac{v_t}{v}\right)=\int{\rm d}\left[\gamma\frac{v^2}{2}-\frac{\delta}{2v^2}\right]\Leftrightarrow\nonumber \\
&\frac{\left(\frac{v_t}{v}\right)^2}{2}=\frac{\gamma v^2}{2}-\frac{\delta}{2v^2}+\mathcal{E}_{\rm{III}}\Leftrightarrow\nonumber \\
&v_t^2=\gamma v^4-\delta+2\mathcal{E}_{\rm{III}}v^2.
\label{firstintegralsn}
\end{align}
\end{proof}

Equation (\ref{firstintegralsn}) can be integrated as follows:
\begin{equation}
\frac{{\rm d}v}{\sqrt{\gamma v^4+2\mathcal{E}_{\rm{III}}v^2-\delta}}={\rm d}t\Rightarrow\int\frac{{\rm d}v}{\sqrt{\gamma v^4+2\mathcal{E}_{\rm{III}}v^2-\delta}}=\int{\rm d}t
\label{elleiptikoInteJAC}
\end{equation}
This is an elliptic integral equation and its inversion will involve an elliptic function for discrete roots of the quartic equation: $\gamma v^4+2\mathcal{E}_{\rm{III}}v^2-\delta=0$.
If all the roots of the quartic are real and distinct and arranged in an ascending order of magnitude $v_1>v_2>v_3>v_4$ then we prove the following result:
\begin{proposition}
Equation (\ref{elleiptikoInteJAC}) can be solved in terms of the Jacobi sinus amplitudinus function ${\rm sn}(x,k)$
\begin{equation}
v=\frac{v_1-v_2 \frac{v_4-v_1}{v_4-v_2}{\rm sn}^2\left(\frac{\gamma \sqrt{(v_3-v_1)(v_4-v_2)}}{2}t\right)}
{1-\frac{v_4-v_1}{v_4-v_2}{\rm sn}^2\left(\frac{\gamma \sqrt{(v_3-v_1)(v_4-v_2)}}{2}t\right)}
\end{equation}
\end{proposition}
\begin{proof}
For $\mathcal{E}_{\rm{III}}<0$ and $\gamma=-\delta=1$ the roots of the quartic equation are real and they are given by the following expressions:
\begin{align}
v_1&=+\sqrt{\frac{-\mathcal{E}_{\rm{III}}+\sqrt{\mathcal{E}_{\rm{III}}^2+\gamma\delta}}{\gamma}}, \\
v_2&=+\sqrt{\frac{-\mathcal{E}_{\rm{III}}-\sqrt{\mathcal{E}_{\rm{III}}^2+\gamma\delta}}{\gamma}}, \\
v_3&=-\sqrt{\frac{-\mathcal{E}_{\rm{III}}-\sqrt{\mathcal{E}_{\rm{III}}^2+\gamma\delta}}{\gamma}}, \\
v_4&=-\sqrt{\frac{-\mathcal{E}_{\rm{III}}+\sqrt{\mathcal{E}_{\rm{III}}^2+\gamma\delta}}{\gamma}}.
\end{align}
Applying the transformation:
\begin{equation}
z=\left(\frac{v_4-v_2}{v_4-v_1}\right)\left(\frac{v-v_1}{v-v_2}\right)\equiv\frac{1}{\omega^{\prime}}\left(\frac{v-v_1}{v-v_2}\right)
\end{equation}
yields
\begin{align}
&\int_{v_1}^v \frac{{\rm d}v^{\prime}}{\sqrt{(v^{\prime}-v_1)(v^{\prime}-v_2)(v^{\prime}-v_3)(v^{\prime}-v_1)}}\nonumber \\
&=\frac{\sqrt{\omega^{\prime}}}{\sqrt{(v_3-v_1)(v_4-v_1)}}\int_0^z\frac{{\rm d}z}{\sqrt{z(1-z)(1-k^2z}},
\end{align}
where the modulus $k^2$ is given by:
\begin{equation}
k^2=\frac{v_4-v_1}{v_4-v_2}\frac{v_2-v_3}{v_1-v_3}
\end{equation}
Finally, using the transformation $z=x^2$ yields
\begin{equation}
\int\frac{{\rm d}x}{\sqrt{(1-x^2)(1-k^2 x^2)}}=\frac{\sqrt{(v_4-v_2)(v_3-v_1)}t}{2},
\end{equation}
and inverting
\begin{equation}
x={\rm sn}\left(\frac{\sqrt{(v_3-v_1)(v_4-v_2)}}{2}t,k^2\right).
\end{equation}
\end{proof}

The Jacobian  functions are elliptic functions with two periods $4mK,4niK^{\prime},m,n\in \mathbb{Z}$:
\begin{align}
{\rm sn}(u+4mK+4niK^{\prime})&={\rm sn}u,\\
{\rm cn}(u+4mK+4niK^{\prime})&={\rm cn}u,\\
{\rm dn}(u+4mK+4niK^{\prime})&={\rm dn}u,
\end{align}
where
\begin{equation}
K=\int_0^1 \frac{{\rm d}t}{\sqrt{(1-t^2)(1-k^2t^2)}}, K^{\prime}=\int_1^{1/k} \frac{{\rm d}x}{\sqrt{(x^2-1)(1-k^2 x^2)}}=\int_0^{\cos^{-1}k}\frac{{\rm d}\theta}{\sqrt{\cos^2\theta-k^2}}
\end{equation}

Following \cite{Joshi} and the results in Proposition \ref{ellipticfunctionsn1}, we can define an energy-like quantity for the generic $\rm{P_{III}}$:
\begin{equation}
{\rm{E_{III}}}:=\frac{1}{2v^2}(v_t^2-\gamma v^4+\delta).
\end{equation}
Then applying the transformation $v=\alpha\; \mathcal{U}(\beta t)=:\alpha\; \mathcal{U}(\eta)$ to the previous equation yields:
\begin{equation}
v_t^2=\alpha^2\beta^2\mathcal{U}_{\eta}^2=2\alpha^2{\rm E_{III}}\mathcal{U}^2+\alpha^4\gamma \mathcal{U}^4-\delta.
\label{ellipticCarlJac}
\end{equation}

If we compare the differential equation (\ref{ellipticCarlJac}) that the new dependent variable $\mathcal{U}$ obeys with the equation for the derivatives of the Jacobian elliptic functions:
\begin{align}
\frac{\rm d}{{\rm d}u}{\rm sn}^2(u,k)&=2{\rm sn}(u,k)\frac{{\rm d sn}u}{{\rm d}u}=2{\rm sn}(u,k){\rm cn}(u,k){\rm dn}(u,k)\Leftrightarrow\nonumber \\
\left(\frac{{\rm d sn}}{{\rm d}u}\right)^2&=(1-{\rm sn}^2 u)(1-k^2 {\rm sn}^2 u ).
 \end{align}
 we conclude that:
 \begin{equation}
 v\sim \alpha\mathcal{U}(\eta)=\alpha \;{\rm sn}(\beta t),\; \text{as}\;\;|t|\rightarrow \infty,
 \end{equation}
 where $\frac{2{\rm E_{III}}}{\beta^2}=-(k^2+1),\;\frac{\alpha^2\gamma}{\beta^2}=k^2,\;\frac{-\delta}{\alpha^2\beta^2}=1$ and the Jacobi modulus satisfies the equation:
 \begin{equation}
 2{\rm{E_{III}}}k\pm\sqrt{-\gamma\delta}(k^2+1)=0.
 \end{equation}
 For the specific angular $\rm{P_{III}}$ differential equation (\ref{Painlevetria}) $\alpha^2=\beta^2=4m^2$ and the Jacobi modulus is $k^2=1$. In this case $\rm{sn}(\beta t,1)=\rm{tanh}(\beta t)$ and the leading doubly periodic behaviour becomes singly periodic. Also
 in this degenerate case we have that: $v\sim(-2m){\rm{tanh}}(-2m\epsilon t)\;\; \text{as}\;\; |t|\rightarrow\infty$ and $E^2_{III}=1$.

Solving equation (\ref{eigenasymptot}) for $w$ results in the following equation that relates the eigenvalues of the KN angular equation to the large-$t$ limit of the Painlev\'{e} $\rm{P_{III}}$ in terms of Jacobian elliptic functions (in the degenerate limit in terms of  hyperbolic functions) \footnote{A full analysis of all the regions in the complex plane in which the elliptic asymptotic limit of the generic Painlev\'{e} $\rm{P_{III}}$ is valid in the spirit of Boutroux \cite{Boutroux}, for all possible values of the parameter $\rm{E}_{\rm{III}}$ is beyond the scope of the current publication and will be investigated elsewhere. Some properties of the elliptic integrals involved as functions of $\rm{E}_{\rm{III}}$ are discussed in Appendix \ref{AbelPicard}.}:
\begin{equation}
\lambda(\nu,\xi)=\frac{\dot{v}t}{2v}.
\label{importantasymptot}
\end{equation}
\subsubsection{Exact solutions for $a\omega=\pm a\mu$}
 In the special case $a\omega=\pm a\mu$ we find that the angular GHE (\ref{GONGHE}) reduces to
 \begin{align}
\Biggl\{\frac{\mathrm{d}^2}{\mathrm{d}z^2}&+\left[\frac{2\alpha_1+1}{z}+
\frac{2\alpha_2+1}{z-1}+\frac{-1}{z-z_3}\right]\frac{\mathrm{d}}{\mathrm{d}z}\nonumber \\
&+\frac{\beta_0+\beta_1 z}{z(z-1)(z-z_3)}\Biggr\}\bar{S}(z)=0.
\label{GONHEfalse}
\end{align}
Thus the exact angular eigenfunctions for this particular case are solutions of a Heun equation. Moreover, the problem simplifies even further if the singular point $z_3$ is a \textit{false} singular point. In this case, as has been shown in \cite{Kraniotis1}, the exact solution of  Heun's differential equation with a false singular point is given in terms of Gau$\ss$ hypergeometric function-see eqn (255), page 39 in \cite{Kraniotis1}. The concept of false singularity is described in detail in \cite{Kraniotis1}. This theory directly applies to (\ref{GONHEfalse}). We have therefore obtained a novel result in a top-down approach.

Now with regard to the angular eigenvalues for the special case $a\omega=\pm a\mu$,  the function $w(\nu):=\lambda(\nu,-\psi\nu)$ for some fixed $\psi\in\{-1,+1\}$ satisfies eqn. (\ref{eigenvalueeqn}). Substituting in (\ref{eigenvalueeqn}) dividing by $-\psi\nu$ and integrating the resulting equation we obtain:
\begin{align}
w(\nu)^2-\frac{w(\nu)}{\psi}&=\nu^2+2m\frac{\nu}{\psi}+C\nonumber \\
\Leftrightarrow \left(w(\nu)-\frac{\psi}{2}\right)^2&=\nu^2+2m\nu\psi+\left(w(0)-\frac{\psi}{2}\right)^2
\nonumber \\
\Leftrightarrow w(\nu)\equiv \lambda_j(m;\nu,-\psi\nu)&=\frac{\psi}{2}+{\rm sgn}(j)\sqrt{(\lambda_j(m,0,0)-\frac{\psi}{2})^2+2m\nu\psi+\nu^2},j\in\Bbb{Z}\setminus \{0\},
\label{nixiequaleigenval}
\end{align}
where $\lambda_j(m,0,0)$ are the eigenvalues for the Reissner-Nordstr\"{o}m black hole. Equation (\ref{nixiequaleigenval}) is in agreement with the corresponding results in \cite{SRDolan},\cite{MonikaDB}.

\section{Remarks on the existence of bound states in the KN spacetime}

An interesting issue  is the problem of whether or not bound states exist for the Dirac equation in the presence of a rotating charged black hole. There is a growing literature in the subject \cite{Finster}-\cite{Belgiorno} with various results. Finster \textit{et al} obtained a non-existence theorem regarding normalisable time-periodic solutions for the Dirac equation in the nonextreme Kerr-Newman geometry \cite{Finster}. In their work the authors introduced certain matching conditions for the spinor field across the Cauchy and event horizons, which in turn, enabled them to obtain a \textit{weak} solution of the Dirac equation valid across the horizons. Specifically, by writing the time periodic  wavefunction as a discrete Fourier series,  they exploited the conservation of the Dirac current (flux) to show that, because of the matching conditions, the only way to obtain a normalisable bound state of the Dirac equation is that each term in the Fourier series expansion is identically zero.  Also the authors in \cite{FBelgiornoMKN} computed the essential spectrum of the radial operator in the KN black hole and concluded that it covers the real axis. The same approach was followed in \cite{Belgiorno} with the inclusion of the cosmological term. On the other hand Schmidt investigated the Dirac equation in the extreme Kerr-Newman geometry and derived a set of necessary and sufficient conditions for the existence of bound states in such a background \cite{HSchmidt} \footnote{The set of complicated inequalities and equalities that represent these conditions are presented in appendix \ref{extremeBound}}. However, even in this case  due to the complex way the radial and angular eigenvalues are intertwined no concrete examples of bound states that satisfy these conditions have been found.

Finster \textit{et al} \cite{Finster} interpreted their results as an indication that, in contrast to the classical situation of massive particle orbits, a quantum mechanical Dirac particle must either disappear into the black hole or escape to infinity. Due to the physical significance of these results it is imperative to investigate the problem further via a different method.
The work in \cite{Finster},\cite{FBelgiornoMKN}, lacked the knowledge of the exact analytic solutions  for both the radial and angular components of the Dirac equation in the Kerr-Newman background that we gained and developed in this work.
Thus it is interesting  to apply our fundamental exact mathematical approach in order to explore this important issue further.

\subsection{Investigation of the existence of bound states with $\omega\in\Bbb{R},\;\omega<\mu$, using the four-term recurrence relation for the power-series coefficients of the solution of the radial GHE  in the KN spacetime }\label{boundfermionicstates}

After separation of the variables in the Dirac equation in KN spacetime, $\omega\in \Bbb{R}$  will be an energy eigenvalue of the Dirac equation, if there exists $\lambda\in \Bbb{R}$ and non-trivial solutions:
\begin{equation}
R(r)=\left[\begin{array}{c} R^{(+)}\\
R^{(-)}\end{array}\right],\;\;\;S(\theta)=\left[\begin{array}{c} S^{(+)}\\
S^{(-)}\end{array}\right], \;r>r_{+},\theta\in(0,\pi),
\end{equation}
satisfying the normalisation conditions:
\begin{equation}
\int_{r_+}^{\infty}|R(r)|^2\frac{r^2+a^2}{\Delta^{KN}}{\rm d}r<\infty,\;\;\int_0^{\pi}|S(\theta)|^2\sin\theta\;{\rm d}\theta<\infty.
\end{equation}
From the local solution near the event horizon (\ref{RadialSpinordecay}) and demanding that the radial spinor wavefunctions decay exponentially asympotically away from the KN black hole, we observe that the problem of existence of fermionic bound states reduces to the question of whether or not the radial GHE (\ref{GHERADIALKN}) admits polynomial solutions for $|\omega|<\mu$, since then the local solution near the event horizon would satisfy the radial normalisation condition.
In order to determine if there exist polynomial solutions we will make use of the four-term recursion relation that the coefficients in the power series expansion of the local solution near the event horizon satisfy:
\begin{equation}
T_k=\varphi(k-1)T_{k-1}-\varphi_2(k-2)T_{k-2}+\varphi_3(k-3)T_{k-3}, \;\;T_{-1}=T_{-2}=0,\;(k\in\Bbb{N}),
\end{equation}
where
\begin{align}
\varphi_1(\xi)&=\xi(\xi+1-\mu_1-\mu_0)+\frac{\xi}{\tilde{a}}(\xi+1-\mu_1-\mu_2)+\alpha\xi-\frac{\tilde{\beta}_0}{\tilde{a}},\\
\varphi_2(\xi)&=(\xi+1)(\xi+1-\mu_1)(\frac{\xi}{\tilde{a}}(\xi+2-\mu_1-\mu_0-\mu_2)+
2\alpha\frac{(\tilde{a}+1)}{\tilde{a}}+\frac{\tilde{\beta}_1}{\tilde{a}}),\\
\varphi_3(\xi)&=(\xi+1)(\xi+2)(\xi+1-\mu_1)(\xi+2-\mu_1)(\frac{\alpha\xi}{\tilde{a}}-\frac{\tilde{\beta}_2}{\tilde{a}}).
\label{anadromikoshorizonevent}
\end{align}
We will have a polynomial solution if the local solution near $r_+$ for the radial GHE in KN spacetime terminates at some positive integer $N$:
\begin{equation}
\bar{R}=\sum_{\nu=0}^N\frac{T_{\nu}\zeta^{\nu}}{\nu!\Gamma(1-\mu_1+\nu)}
\end{equation}
We will have a polynomial solution if:
\begin{equation}
T_{k-1}=T_{k-2}=0,\varphi_3(k-3)=0.
\end{equation}
Indeed setting $k-3=N$ in the recursive formula (\ref{anadromikoshorizonevent}) the conditions for a polynomial solution become:
\begin{equation}
T_{N+2}=T_{N+1}=0,\;\;\varphi_3(N)=0
\end{equation}
Bound states will occur for those real values of $\omega$ such that $|\omega|<\mu$ and $\omega$ satisfies simultaneously the above conditions.
The condition $\varphi_3(N)=0$ is satisfied if:
\begin{equation}
N+1-\mu_1=0,\text{or}\; N+2-\mu_1-0,\;\text{or}\;\alpha N-\tilde{\beta}_2=0.
\label{SIGNIFICANTPROOF}
\end{equation}
Interestingly enough, the condition $\varphi_3(N)=0$ cannot be satisfied. For instance,
\begin{align}
&\alpha N-\tilde{\beta}_2=0\Leftrightarrow \nonumber \\ &-2\sqrt{\mu^2-\omega^2}(r_+-r_-)N=C_1+C_2+C_3-2\sqrt{\mu^2-\omega^2}(\boldsymbol{\mu_1}+\boldsymbol{\mu_2})(r_+-r_-) \end{align}
where:
\begin{align}
\boldsymbol{\mu_1}&=\frac{-\left(\frac{M-r_-}{r_--r_+}\right)\pm\sqrt{\left(\frac{M-r_-}{r_--r_+}\right)^2-4[K(r_+)^2+i (r_+-M)K(r_+)]}}{2},\\
\boldsymbol{\mu_2}&=\frac{\left(\frac{M-r_+}{r_--r_+}\right)\pm\sqrt{\left(\frac{M-r_+}{r_--r_+}\right)^2-4[K(r_-)^2+i (r_--M)K(r_-)]}}{2},
\end{align}
and\footnote{Also $C_1+C_2+C_3\in \Bbb{R}.$}:
\begin{equation}
K(r_\pm)=\omega (a^2+r_{\pm}^2)-am+eqr_{\pm}.
\end{equation}
Thus the condition $\varphi_3(N)=0$ demands the equality of a real number with a complex number which is absurd \footnote{For the KN black hole with $r_{+}\not = r_{-}$ no real value of $\omega$ exists that simultaneously leads to the vanishing of both $K(r_+)$ and $K(r_-)$. }.
Likewise, we have checked that none of the conditions in (\ref{SIGNIFICANTPROOF})  can be satisfied if one of the two factors $K(r_+),K(r_-)$ is vanishing.
Suppose that $K(r_+)=0$ or equivalently $\omega=\frac{am-eqr_+}{a^2+r_+^2}$. Then
$1-\mu_1=2\boldsymbol{\mu}_1+\frac{M-r_+}{r_--r_+}=1\pm\left|\frac{M-r_-}{r_--r_+}\right|$.
Then $N+1-\mu_1=0\Leftrightarrow N+1\pm \frac{1}{2}=0$ which is also absurd.

Thus by employing the power series representations of the local exact GHE solutions for the radial components of the Dirac equation in the KN background whose coefficients satisfy a four-term recurrence relation, we have proved the theorem:
\begin{theorem}\label{noboundKN}
In the non-extreme Kerr-Newman geometry there are no fermionic bound states with $\omega^2<\mu^2$, where $\omega$ and $\mu$ are the energy and mass of the fermion respectively.
\end{theorem}
Our exact mathematical method that was culminated in the theorem \ref{noboundKN}, corroborates in the most emphatic way the results in \cite{Finster},\cite{FBelgiornoMKN},obtained by different means, for the absence of physical fermionic bound states in the non-extreme KN geometry .

\section{Conclusions}

In this work we have investigated the massive Dirac equation in the KNdS and the KN black hole backgrounds.  First we derived and separated the Dirac equation in the Kerr-Newman-de Sitter (KNdS) black hole background using a generalised Kinnersley null tetrad in the Newman-Penrose formalism. By using appropriate transformations for the independent variable and appropriate index transformations for the dependent variable we proved the novel result that the resulting second order differential radial  equation for a spin $\frac{1}{2}$ massive charged fermion in KNdS background generalises in a highly non-trivial way the ordinary Heun equation.
With the aid of a Regge-Wheeler-like coordinate and a suitable change of dependent variable we transformed the massive radial equation for the KNdS black hole into a Schr\"{o}dinger-like differential equation. Subsequently we investigated the asymptotics of this novel equation and derived its near-event and near-cosmological horizon limits.

Taking the zero cosmological constant limit we investigated the massive Dirac equation in the Kerr-Newman spacetime. In the Kinnersley tetrad, by suitable transformations of the independent and dependent variables we transformed the separated angular and radial parts in the KN spacetime into generalised Heun equations (GHEs). Such GHEs are characterised by the fact that they possess three regular singularities and an irregular singularity at $\infty$. Heun's differential equation is a special case of such a GHE. We then derived local analytic solutions of these GHEs in which the series coefficients obey a four-term recursion relation. Moreover we investigated the radial solutions near the event horizon and far away from the black hole. The global behaviour of the solutions, along the lines of \cite{RegIrregSIAMJ}, was also studied with emphasis on the connection problem for a regular and an irregular singular point and the computation of the corresponding connection coefficients.
We also performed a detailed analysis for the determination of the separation constant $\lambda$ that appears in both radial and angular equations in the Kerr-Newman background. The procedure involves an eigenvalue matrix problem for the KN angular equations. The angular eigenvalues obey a pde which when integrated with Charpit's method leads to the Painlev\'{e} $\rm{P_{III}}$ ordinary nonlinear differential equation. There is a vast interest in the solutions of the $\rm{P_{III}}$ transcendent. For particular values of the parameters we derived closed form solutions for this angular $\rm{P_{III}}$ nonlinear ODE in terms of Bessel functions.
We then performed a novel asymptotic analysis of the specific angular $\rm{P_{III}}$ transcendent in terms of Jacobian elliptic functions. This is analogous to the scattering theory of ordinary quantum mechanics in which the Bessel functions that solve the radial equation have as an asymptotic limit trigonometric functions.
Our closed form solutions of both the radial and angular components of the Dirac equation in the KN curved background are of physical importance for the theory of general relativistic quantum scattering  by rotating charged black holes as well as for the computation of emission rates from the KN black hole.
It would also be very interesting to obtain further closed form solutions of the angular Painlev\'{e} $\rm{P_{III}}$ transcendent,  by working in the framework of the linearisation of $\rm{P_{III}}$ by the Lax pair \cite{Jimbo1}. A background for the procedure is outlined in appendix  \ref{RiemannHilbertRep}. Such an analysis will be a task for the future.
Using the four-term recurrence relation the coefficients in the power series representation of the holomorphic local solutions of the radial GHE in the KN spacetime obey, we proved that bound fermionic states do not exist  with $\omega^2<\mu^2$.

Future research will also be conducted to obtain analytic solutions of the separated radial and angular ODEs (\ref{ALGMEINEDIR1})-(\ref{ALGEMEINEDIR4}) in the KNdS black hole background.
The construction of solutions of eqn.(\ref{MasDiracGHrKNdS}) is a demanding problem and would lead to new functions beyond the  solutions of the GHE in the Kerr-Newman case and the Heun functions. It would also be very interesting to solve the angular eigenvalue problem in the Kerr-Newman-de Sitter background.
 Such solutions would open the way to computing among other things the effect of the cosmological constant on the quasi-normal modes (QNMs) \footnote{See \cite{FermEliza} for some recent work on quasinormal modes of massive fermions in Kerr spacetime.}as well as the analytic computation of emission rates from a KNdS black hole. They would also shed light on the issue of bound states.   We hope to engage in such  exciting research endeavours and report progress in a future publication.

\section*{Acknowledgements} This work was funded in part by a Teaching Position grant under the auspices of the program ESPA 2014-2020, Code 82028. We dedicate this work to the memory of David Bailin. We are grateful to the referees for the careful reading of the manuscript and their very constructive comments and suggestions.

\appendix{}
\section{The connection problem and the Floquet solutions at the simple singularities $0,1,a=z_3$}
By applying the transformation of the dependent variable
\begin{equation}
\zeta=1-z
\end{equation}
the differential equation (\ref{SCHASCHMI}) transforms into:
\begin{align}
y^{\prime\prime}(\zeta)&+\left(\frac{1-\mu_1}{\zeta}+\frac{1-\mu_0}{\zeta-1}+\frac{1-\mu_2}{\zeta-(1-a)}-\alpha\right)
y^{\prime}(\zeta)\nonumber \\
&+\frac{\beta_0+\beta_1(1-\zeta)+\beta_2(1-\zeta)^2}{(1-\zeta)(-\zeta)(1-\zeta-a)}y(\zeta)=0,
\end{align}
thus we have
\begin{align}
&(\mu_0,\mu_1,\mu_2)\rightarrow (\mu_1,\mu_0,\mu_2) \nonumber \\
&\tilde{a}=1-a \nonumber \\
&\tilde{\alpha}=-\alpha \nonumber \\
&\tilde{\lambda}=-\lambda.\nonumber
\end{align}
\begin{table}[tbp] \centering
\begin{tabular}
[c]{|l|l|l|l|l|l|}\hline
& $\zeta$ & $\tilde{\mu}$ & $\tilde{a}$&$\tilde{\alpha}$ &$\tilde{\lambda}$
\\\hline
$1)$& $z$ & $(\mu_0,\mu_1,\mu_2)$ & $a$ &$\alpha$ &$\lambda$ \\

$2)$ &$1-z$ & $(\mu_1,\mu_0,\mu_2)$ & $1-a$ &$-\alpha$ &$-\lambda$\\
$3)$ &$\frac{z}{a}$ & $(\mu_0,\mu_2,\mu_1)$ &$\frac{1}{a}$ & $a\alpha$ & $a\lambda$\\
$4)$&$\frac{1-z}{1-a}$ &$(\mu_2,\mu_0,\mu_1)$&$\frac{1}{1-a}$ &$\alpha(a-1)$ &$\lambda(a-1)$\\
$5)$&$1-\frac{z}{a}$&$(\mu_1,\mu_2,\mu_0)$& $1-\frac{1}{a}$& $-a\alpha$ &$-a\lambda$\\
$6)$& $\frac{a-z}{a-1}$ &$(\mu_2,\mu_1,\mu_0)$ &$\frac{a}{a-1}$&$(1-a)\alpha$ & $(1-a)\lambda$
\\\hline
\end{tabular}
\caption{The six possible substitutions resulting from the linear transformation (\ref{lineartransf}), which map the simple singularities $(0,1,a)$ to $(0,1,\tilde{a})$ and keep the irregular singularity $\infty$ fixed \cite{RSCHAFKEDSCHMIDT}.}
\end{table}%

In \cite{RSCHAFKEDSCHMIDT} using the Wronskian determinant  for two Floquet solutions and methods of complex analysis such as Watson's lemma for loop integrals and the ratio of two Gamma functions \cite{OLVER}, the following connection formula relating $y_{01}$ to $y_{11}$ and $y_{12}$ was proven:
\begin{proposition}
Let $\left|a\right|>1$. Then there exists a unique function $q=q(\cdot;a)$ holomorphic in  $(\mu,\alpha,\lambda)\in\mathbb{C}^7$, such that the connection formula:
\begin{align}
\frac{\sin(\pi \mu_1)}{\pi}\eta(z,\mu,\alpha,\lambda;a)=q&(\mu_0,-\mu_1,\mu_2,\alpha,\lambda;a)
\eta(1-z,\mu_1,\mu_0,\mu_2,-\alpha,-\lambda;1-a) \nonumber \\
&-q(\mu_0,\mu_1,\mu_2,\alpha,\lambda;a)(1-z)^{\mu_1}\nonumber \\
&\cdot \eta(1-z,-\mu_1,\mu_0,\mu_2,-\alpha,-\lambda;1-a)
\end{align}
is valid for $\left|z\right|<1,\left|z-1\right|<\min(1,\left|a-1\right|), \arg(1-z)\in\left]-\pi,\pi\right[$ and $(\mu,\alpha,\lambda)\in\mathbb{C}^7$.
\end{proposition}
Thus the connection problem between the sets of the Floquet solutions will actually be solved if the function $q$ can be evaluated. Indeed, the following theorem was proved in \cite{RSCHAFKEDSCHMIDT}
\begin{theorem}
Let $\left|a\right|>1$. Then for $(\mu,\alpha,\lambda)\in \mathbb{C}^7$,
\begin{align*}
q(\mu,\alpha,\lambda;a)=&\lim_{k\rightarrow \infty}\frac{\tau_{k}(\mu,\alpha,\lambda;a)}{\Gamma(k+1-\mu_0)\Gamma(k-\mu_1)}\\
&\left(1+\sum_{l=1}^m\frac{\tau_l(-\mu_1,\mu_0,\mu_2,-\alpha,-\lambda;1-a)}{l!}\prod_{\sigma=1}^l(\sigma+\mu_1-k)^{-1}\right)^{-1},
\end{align*}
the convergence being $\mathcal{O}(k^{-m-1})$ as $k\rightarrow \infty$, where $m$ is an arbitrary nonnegative integer.
\end{theorem}
Thus for instance, we obtain the following connection formula:
\begin{align}
\frac{\sin(\pi\mu_0)}{\pi}\xi_0&\eta(1-z,\mu_1,\mu_0,\mu_2,-\alpha,-\lambda;1-a)\nonumber \\
&=-q(-\mu_0,\mu_1,\mu_2,\alpha,\lambda;a)\eta(z,\mu_0,\mu_1,\mu_2,\alpha,\lambda;a) \nonumber \\
&+q(\mu_0,\mu_1,\mu_2,\alpha,\lambda;a)z^{\mu_0}\eta(z,-\mu_0,\mu_1,\mu_2,\alpha,\lambda;a),
\end{align}
relating $y_{11}$ to $y_{01}$ and $y_{02}$.
\section{The Painlev\'{e} $\rm{P_{III}}$ differential equation and its Lax pair isomonodromy representation}\label{RiemannHilbertRep}
The canonical form of the Painlev\'{e} third differential equation is given by:
\begin{equation}
\frac{{\rm d}^2u}{{\rm d}x^2}=\frac{1}{u}\left(\frac{{\rm d}u}{{\rm d}x}\right)^2-\frac{1}{x}\frac{{\rm d}u}{{\rm d}x}
+\frac{\alpha u^2+\beta}{x}+\gamma u^3+\frac{\delta}{u},
\label{PainleveDREI}
\end{equation}
with $\alpha,\beta,\gamma$ and $\delta$ constants.
An important property of $\rm{P_{III}}$ is that it has a number of simple scaling transformations which are denoted by
by $C_1^k,C_2^n$ and $C_3$. They imply that if $u(x;\alpha,\beta,\gamma,\delta)$ is a solution of $\rm{P_{III}}$ with parameters $\alpha,\beta,\gamma$ and $\delta$ then the following are also solutions with the parameter values stated:
\begin{equation}
C_1^k:\;\;\;\tilde{u}(x;\tilde{\alpha},\tilde{\beta},\tilde{\gamma},\tilde{\delta})=ku(x;\alpha,\beta,\gamma,\delta),
\end{equation}
where $\tilde{\alpha}=\alpha/k,\;\tilde{\beta}=k\beta,\;\tilde{\gamma}=\gamma k^{-2}$ and $\tilde{\delta}=\delta k^2$;
\begin{equation}
C_2^n:\;\;\;u(\zeta;\tilde{\alpha},\tilde{\beta},\tilde{\gamma},\tilde{\delta}),
\end{equation}
where $\zeta=x/n, \;\tilde{\alpha}=n\alpha,\;\tilde{\beta}=n\beta\;\tilde{\gamma}=\gamma n^2$ and $\tilde{\delta}=\delta n^2$.

A consequence of the scaling transformations $C_1^k$ and $C_2^n$ is that any version of  $\rm{P_{III}}$ with $\gamma\delta\not =0$ is equivalent within a scaling to:
\begin{equation}
u_{xx}=\frac{(u_x)^2}{u}-\frac{u_x}{x}+\frac{\alpha_s u^2+\beta_s}{x}+u^3-\frac{1}{u},
\end{equation}
an equation that can be regarded as a canonical form of $\rm{P_{III}}$ for $\gamma\delta\not =0$.

A very important property of non-linear differential equations such as  $\rm{P_{III}}$ is that they can be considered in the framework of the inverse problem and can be represented as compatibility conditions for a system of auxiliary linear problems.
Indeed Jimbo \textit{et al} considered the coupled system consisting of the Lax pair \cite{Jimbo1}:
\begin{align}
\frac{\partial Y(x;t)}{\partial x}&=A(x;t)Y(x;t), \\
\frac{\partial Y(x;t)}{\partial t}&=B(x;t)Y(x;t),
\end{align}
where
\begin{equation}
Y_x(x;t)=\Biggl\{\frac{1}{2}\left[\begin{array}{cc}
t & 0\\
0 &-t\end{array}\right]+\frac{1}{x}\left[\begin{array}{cc}
-\theta_{\infty}/2 & u\\
v &\theta_{\infty}/2\end{array}\right]+\frac{1}{x^2}\left[\begin{array}{cc}
z-t/2 & -wz\\
\frac{z-t}{w} &-z+\frac{t}{2}\end{array}\right]\Biggr\}Y(x;t),
\label{Lax1}
\end{equation}
\begin{equation}
Y_t(x;t)=\Biggl\{\frac{1}{2} \left[\begin{array}{cc}
1 & 0\\
0 &-1\end{array}\right]x+\frac{1}{t}\left[\begin{array}{cc}
0 & u\\
v &0\end{array}\right]-\frac{1}{t}\left[\begin{array}{cc}
z-t/2 & -wz\\
\frac{z-t}{w} &-z+\frac{t}{2}\end{array}\right]\frac{1}{x}\Biggr\}Y(x;t)
\label{Lax2}
\end{equation}
Here $x$ and $t$ are taken to be independent complex variables; $u,v,w,z$ are functions of $t$ alone and $\theta_{\infty}$ is a constant. The compatibility condition $Y_{xt}(x;t)=Y_{tx}(x;t)$  is satisfied if and only if:
\begin{align}
\frac{{\rm d}z}{{\rm d}t}&=-2abv-\frac{2yz^2}{t}+\frac{2ytz}{t}+\frac{z}{t}\nonumber \\
&=2wzv+\frac{2u}{w}(z-t)+z,\\
\frac{{\rm d}u}{{\rm d}t}&=\frac{\theta_{\infty}}{t}u+2tab=\frac{\theta_{\infty}}{t}u-2wz\label{ulax}\\
t\frac{{\rm d}\log w}{{\rm d}t}&=\theta_{\infty}-\frac{(\theta_0+\theta_{\infty})t}{z}-2yt\nonumber \\
&=-\theta_{\infty}-2wv+\frac{2u}{w},\\
\frac{{\rm d}v}{{\rm d}t}&=-\frac{\theta_{\infty}}{t}v-\frac{2}{w}(z-t).
\end{align}
If $y=-\frac{u}{zw}$  it then follows that:
\begin{equation}
\frac{{\rm d}y}{{\rm d}t}=\frac{4y^2 z}{t}-\frac{2y^2t}{t}+\frac{(2\theta_{\infty}-1)y}{t}+2
\label{TranscendentDrei}
\end{equation}
and as a result of the compatibility  conditions we conclude that $y(t)$ satisfies the non-linear  Painlev\'{e} $\rm{P_{III}}$ differential equation:
\begin{equation}
\frac{{\rm d}^2y}{{\rm d}t^2}=\frac{1}{y}\left(\frac{{\rm d}y}{{\rm d}t}\right)^2-\frac{1}{t}\frac{{\rm d}y}{{\rm d}t}+
\frac{1}{t}(\alpha y^2+\beta)+\gamma y^3+\frac{\delta}{y},
\end{equation}
with the parameters $\alpha=4\theta_0,\;\beta=4(1-\theta_{\infty}),\;\gamma=4,\;\delta=-4$.
From the equality:
\begin{equation}
-2abv=2\frac{wz}{t}v=\theta_0+\theta_{\infty}-\frac{2yz^2}{t}+\frac{2tyz}{t}-\frac{2\theta_{\infty}z}{t},
\end{equation}
the constant $\theta_0$ is determined as follows:
\begin{equation}
\theta_0=\frac{2wzv}{t}+\theta_{\infty}\left(\frac{2z}{t}-1\right)+\frac{2uz}{zwt}(t-z).
\label{integralofmotion}
\end{equation}
\subsection{The Lax pair for $\theta_0=\theta_{\infty}-1$ and $y(t)=1$}
In the special case that $\theta_0=\theta_{\infty}-1$, the Painlev\'{e} $\rm{P_{III}}$ differential equation that arises as a compatibility condition for the Lax pair (\ref{Lax1})-(\ref{Lax2}) admits the rational (constant) solution $y(t)=1$. If we set $y(t)=1$ in $y=-\frac{u}{zw}$, eqn.(\ref{ulax}) becomes:
\begin{equation}
\frac{{\rm d}u}{{\rm d}t}=\frac{\theta_{\infty}}{t}u+2u,
\end{equation}
which has the solution:
\begin{equation}
u(t)=-\frac{\mathcal{K}}{4}t^{\theta_{\infty}}e^{2t},
\end{equation}
and as a result:
\begin{equation}
wz=\frac{\mathcal{K}}{4}t^{\theta_{\infty}}e^{2t}.
\end{equation}
Equation (\ref{TranscendentDrei}) yields:
\begin{equation}
z=\frac{1-2\theta_{\infty}}{4},
\end{equation}
and as a result:
\begin{equation}
w=\frac{\mathcal{K}t^{\theta_{\infty}}e^{2t}}{1-2\theta_{\infty}}
\end{equation}
If $\theta_0=\theta_{\infty}-1$, then solving (\ref{integralofmotion}) for $v$ yields:
\begin{equation}
v=\frac{1}{4}\mathcal{K}^{-1}t^{-\theta_{\infty}}e^{-2t}\{(2\theta_{\infty}-1)(4t+2\theta_{\infty}+1)\}.
\end{equation}
There are various methods to calculate the necessary monodromy data so that the solution $y(x)=1$ is obtained through the linearisation (\ref{Lax1})-(\ref{Lax2}). One of them is via the construction of an appropriate Riemann-Hilbert problem \cite{solitons}. One can proceed by solving the isomonodromic deformation (\ref{Lax2}) first and then building an additional dependence on $x$ via integration constants to satisfy (\ref{Lax1}) as well.
Applying to (\ref{Lax2}) the transformation:
\begin{equation}
Y=e^{t\sigma_3}t^{\theta_{\infty}\sigma_3/2}t^{-1/2}W,
\end{equation}
where $\sigma_3$ denotes the third Pauli matrix, we find:
\begin{align}
&\frac{1}{\sqrt{t}}\frac{\partial W_1}{\partial t}+\frac{W_1}{\sqrt{t}}+\frac{(\theta_{\infty}-1)}{2t^{3/2}}W_1\nonumber \\
&=\left[\frac{x}{2}-\frac{1}{tx}\frac{(1-2\theta_{\infty})}{4}+\frac{1}{2x}\right]\frac{W_1}{\sqrt{t}}-\frac{\mathcal{K}}{4t^{3/2}}
\left(1-\frac{1}{x}\right)W_2,
\label{edlax1}
\end{align}
and
\begin{align}
&\frac{1}{\sqrt{t}}\frac{\partial W_2}{\partial t}-\frac{W_2}{\sqrt{t}}-\frac{(\theta_{\infty}+1)}{2t^{3/2}}W_2\nonumber \\
&=\frac{\mathcal{K}^{-1}}{4t}(1-2\theta_{\infty})\left[-1-\frac{1}{x}-(4t+2\theta_{\infty})\left(1-\frac{1}{x}\right)\right]
\frac{W_1}{\sqrt{t}}\nonumber \\
&+\left(-\frac{x}{2}+\frac{1}{tx}\frac{(1-2\theta_{\infty})}{4}-\frac{1}{2x}\right)\frac{W_2}{\sqrt{t}}.
\label{edlax2}
\end{align}
Solving (\ref{edlax1}) for $W_2$ and substituting into (\ref{edlax2}) we conclude
that the first-row matrix elements $W_{1j}$ of $W$ satisfy a Whittaker differential equation:
\begin{equation}
\frac{{\rm d}^2 W_1}{{\rm d}\zeta^2}+\left[-\frac{1}{4}+\frac{\kappa}{\zeta}+\frac{1-4m^2}{4\zeta^2}\right]W_1=0,
\label{InverseWhittaker}
\end{equation}
where the parameters of the Whittaker equation are:$\kappa:=\frac{1}{2}(\theta_{\infty}-1),\;m=\frac{1}{4}$ and the independent variable is defined by $\zeta:=-t(2-x-x^{-1})$.
The components $W_{2j}$ are expressed in terms of the matrix elements $W_{1j}$ as follows:
\begin{equation}
W_{2j}=\frac{-4\zeta\left[\frac{\partial W_{1j}}{\partial \zeta}-\frac{1}{2}W_{1j}\right]-4\kappa W_{1j}-x^{-1}(1-2\theta_{\infty})W_{1j}}{\mathcal{K}\left(1-\frac{1}{x}\right)}.
\end{equation}
The Whittaker functions which fall off exponentially as $\zeta\rightarrow \infty$ have the form:
\begin{align}
W_{\kappa,m}(\zeta)&=e^{-\zeta/2}\zeta^{m+\frac{1}{2}}U(m-\kappa+\frac{1}{2},1+2m,\zeta)\nonumber \\ &=\frac{\Gamma(-2m)}{\Gamma(-m-\kappa+\frac{1}{2})}e^{-\zeta/2}\zeta^{m+\frac{1}{2}}M(m-\kappa+\frac{1}{2},1+2m,\zeta)\nonumber \\
&+\frac{\Gamma(2m)}{\Gamma(m-\kappa+\frac{1}{2})}\zeta^{-2m}e^{-\zeta/2}\zeta^{m+\frac{1}{2}}M(-m-\kappa+\frac{1}{2},1-2m,\zeta)\nonumber \\
&=\frac{\Gamma(-2m)}{\Gamma(-m-\kappa+\frac{1}{2})}M_{\kappa,m}(\zeta)+\frac{\Gamma(2m)}{\Gamma(m-\kappa+\frac{1}{2})}M_{\kappa,-m}(\zeta),
\end{align}
where $U(a,c,\xi)$ is the Kummer hypergeometric function of the second kind with asymptotic behaviour:
\begin{align}
U(a,c,\xi)&\sim \xi^{-a}\left[1-\frac{ab}{\xi}+\frac{a(a+1)b(b+1)}{2!\xi^2}-\cdots\right]
\nonumber \\ &=
\frac{1}{\xi^a}\sum_{\nu=0}^{\infty}\frac{(a)_{\nu}(b)_{\nu}}{(1)_{\nu}}
\left(\frac{-1}{\xi}\right)^{\nu}\;\;\xi\rightarrow \infty,
\end{align}
και $b:=1+a-c$.

If we fix a fundamental pair of solutions of (\ref{InverseWhittaker}) that depend on $x$ through the variable $\zeta$ as the first row of the matrix $W$ ($W_{11}:=W_{-\kappa,m}(-\zeta),W_{12}:=W_{\kappa,m}(\zeta)$), then the general solution of the isomonodromic deformation of (\ref{Lax2}) can be written in the form:
\begin{equation}
Y=e^{t\sigma_3}t^{\theta_{\infty}\sigma_3/2}t^{-1/2}W\mathcal{C}(x),
\label{genikilisiLax2}
\end{equation}
where $\mathcal{C}(x)$ cannot depend on $t$ but might depend on $x$. Substituting (\ref{genikilisiLax2}) into (\ref{Lax1}) yields:
\begin{equation}
\mathcal{C}(x)=(x-1)^{-1/2}\mathcal{C},
\label{statherax}
\end{equation}
where $\mathcal{C}$ is a matrix independent of both $t$ and $x$. Thus (\ref{genikilisiLax2}) with $\mathcal{C}$ determined as in (\ref{statherax}) constitutes a simultaneous solution of both equations (\ref{Lax1}),(\ref{Lax2}) in the Lax pair.

\subsubsection{The Picard-Fuchs differential equation that the elliptic integrals satisfy in the asymptotic limit of the generic $\rm{P_{III}}$ with respect to  $\mathcal{E}_{\rm{III}}$ }\label{AbelPicard}

In order to study the dependence of the Abelian elliptic integrals on the energy-like quantity $\rm{E}_{\rm{III}}$, to prove boundedness properties of $\rm{E}_{\rm{III}}$ and to determine further regions in the complex plane in which the elliptic asymptotic limit of $\rm{P_{III}}$ is valid it is convenient to make the change of variables $v=e^u$\cite{Joshi}. Then $\rm{P_{III}}$  is transformed into the differential equation:
\begin{equation}
u_{tt}=e^{2u}-e^{-2u}-\frac{u_t}{t}+\frac{\alpha e^u+\beta \; e^{-u}}{t}
\end{equation}

Taking the limit $|t|\rightarrow\infty,$ gives :
$$u_{tt}\sim e^{2u}-e^{-2u}.$$
Furthermore one can apply the results in \cite{Joshi} (Lemma 2.2), in which the solution of a general $\rm{P_{III}}$ in the transformed variable  was shown to have the asymptotic expansion $u=V+\hat{u}$, where $\hat{u}\ll V$ as $|t|\rightarrow \infty$. The full analysis will be performed in a separate publication \cite{Kraniotis2}. However, we will calculate in this appendix the
$\rm{E}_{\rm{III}}-$   dependence of the Abelian elliptic integrals defined in the elliptic limit of the Painlev\'{e} $\rm{P_{III}}$.

The elliptic functional $\tilde{\omega}:=\oint e^{-V}\sqrt{e^{4V}+2{\mathcal{E}}_{\rm{III}}e^{2V}+ 1}\;{\rm d}V$ satisfies
the differential equation:
\begin{equation}
\tilde{\omega}^{\prime\prime}=-\frac{\tilde{\omega}}{4({\mathcal{E}}_{\rm{III}}^2-1)},
\label{PicardFuchs}
\end{equation}
where $^{\prime}:=\frac{\rm d}{{\rm d}{\mathcal{E}}_{\rm{III}}}$ \cite{Joshi}.
Indeed after integration by parts:
\begin{align}
\tilde{\omega}&=\oint e^{-V}\sqrt{e^{4V}+2{\mathcal{E}}_{\rm{III}}e^{2V}+ 1}\;{\rm d}V\nonumber \\
&=-\oint{\rm d}(e^{-V})\;\sqrt{e^{4V}+2{\mathcal{E}}_{\rm{III}}e^{2V}+ 1}\nonumber \\
&=-e^{-V}\sqrt{e^{4V}+2{\mathcal{E}}_{\rm{III}}e^{2V}+ 1}|_{\gamma}+
\oint e^{-V}\frac{4e^{4V}+4{\mathcal{E}}_{\rm{III}}e^{2V}}{2\sqrt{e^{4V}+2{\mathcal{E}}_{\rm{III}}e^{2V}+ 1}}\;{\rm d}V\nonumber \\
&=\oint e^{-V}\frac{2e^{4V}+2{\mathcal{E}}_{\rm{III}}e^{2V}}{\sqrt{e^{4V}+2{\mathcal{E}}_{\rm{III}}e^{2V}+ 1}}\;{\rm d}V\nonumber \\
&=\oint e^{-V}\frac{2e^{4V}+4{\mathcal{E}}_{\rm{III}}e^{2V}- 2{\mathcal{E}}_{\rm{III}}e^{2V}+2-2}{\sqrt{e^{4V}+2{\mathcal{E}}_{\rm{III}}e^{2V}+ 1}}\;{\rm d}V\nonumber \\
&=2\oint e^{-V}\sqrt{e^{4V}+2{\mathcal{E}}_{\rm{III}}e^{2V}+ 1}\;{\rm d}V-2{\mathcal{E}}_{\rm{III}}\oint \frac{e^V}
{\sqrt{e^{4V}+2{\mathcal{E}}_{\rm{III}}e^{2V}+ 1}}\;{\rm d}V\nonumber \\
&-2\oint \frac{e^{-V}}
{\sqrt{e^{4V}+2{\mathcal{E}}_{\rm{III}}e^{2V}+ 1}}\;{\rm d}V\nonumber\Rightarrow \\
&-\tilde{\omega}=-2{\mathcal{E}}_{\rm{III}}\omega-2\Psi,\label{ELLIntrunE}
\end{align}
where $\Psi:=\oint e^{-V}\frac{1}{\sqrt{e^{4V}+2{\mathcal{E}}_{\rm{III}}e^{2V}+ 1}}\;{\rm d}V$,
$\omega_i=\oint\frac{1}{\sqrt{e^{2V}+2{\mathcal{E}}_{\rm{III}}+e^{-2V}}}{\rm d}V$.
Likewise performing integration by parts on the integral that defines $\Psi$ yields:
\begin{align}
&\oint\frac{ e^{-V}}{\sqrt{e^{4V}+2{\mathcal{E}}_{\rm{III}}e^{2V}+ 1}}\;{\rm d}V=-\oint \frac{{\rm d}e^{-V}}{\sqrt{e^{4V}+2{\mathcal{E}}_{\rm{III}}e^{2V}+ 1}}\nonumber \\
&=-\oint e^{-V}\frac{2e^{4V}+2{\mathcal{E}}_{\rm{III}}e^{2V}}{(e^{4V}+2{\mathcal{E}}_{\rm{III}}e^{2V}+1)^{3/2}}\;{\rm d}V\nonumber \\
&=-2\mathcal{E}_{\rm{III}}\oint \frac{e^V}{(e^{4V}+2{\mathcal{E}}_{\rm{III}}e^{2V}+1)^{3/2}}\;{\rm d}V
-2\oint \frac{e^{3V}}{(e^{4V}+2{\mathcal{E}}_{\rm{III}}e^{2V}+1)^{3/2}}\;{\rm d}V\nonumber \\
&\Leftrightarrow \Psi=2{\mathcal{E}}_{\rm{III}}\Psi^{\prime}+2\tilde{\omega}^{\prime\prime}=2{\mathcal{E}}_{\rm{III}}\Psi^{\prime}+2\omega^{\prime}.
\label{elliptikoOLO}
\end{align}
Taking $\mathcal{E}_{\rm{III}}-$ derivatives of (\ref{ELLIntrunE}) yields:
\begin{equation}
\Psi^{\prime}=\frac{2\tilde{\omega}^{\prime}+2{\mathcal{E}}_{\rm{III}}\tilde{\omega}^{\prime\prime}-\tilde{\omega}^{\prime}}{-2}.
\label{ederiv}
\end{equation}
Using (\ref{ELLIntrunE}),(\ref{elliptikoOLO}),(\ref{ederiv}), we prove  eqn.(\ref{PicardFuchs}).
\subsection{The differential equation for the  eigenvalues $\lambda$ in KN spacetime}\label{pdePIIICharpit}

Following the analysis in \cite{KatoLPT} (Theorem 3.6, VII, \textsection 3),\cite{HSchmidt},\cite{MonikaDB} the partial derivatives of the angular eigenvalues are determined as follows:
\begin{align}
\frac{\partial \lambda}{\partial \nu}&=\left\langle\frac{\partial A}{\partial \nu} \mathcal{S},\mathcal{S}\right\rangle=\int_0^{\pi}\mathcal{S}^*(\theta)\left[\begin{array}{cc} -\cos\theta&0\\
0&\cos\theta\end{array}\right]\left[\begin{array}{c}\mathcal{S}^+(\theta)\\ \mathcal{S}^-(\theta)\end{array}\right]{\rm d}\theta\nonumber \\
&=\int_0^{\pi}(-\cos\theta\mathcal{S}^{+2}(\theta)+\cos\theta\mathcal{S}^{-2}(\theta)){\rm d}\theta=:\int_0^{\pi}\cos\theta V(\theta){\rm d}\theta,
\end{align}

\begin{align}
\frac{\partial \lambda}{\partial \xi}&=\left\langle\frac{\partial A}{\partial \xi} \mathcal{S},\mathcal{S}\right\rangle=\int_0^{\pi}\mathcal{S}^*(\theta)\left[\begin{array}{cc} 0&\sin\theta \\
\sin\theta&0\end{array}\right]\left[\begin{array}{c}\mathcal{S}^+(\theta)\\ \mathcal{S}^-(\theta)\end{array}\right]{\rm d}\theta\nonumber \\
&=\int_0^{\pi}(\sin\theta\mathcal{S}^{+}(\theta)\mathcal{S}^{-}(\theta)+\sin\theta\mathcal{S}^{-}(\theta)\mathcal{S}^{+}(\theta)){\rm d}\theta=:\int_0^{\pi}\sin\theta W(\theta){\rm d}\theta,
\end{align}
where the function $\left\langle\cdot,\cdot\right\rangle$ denotes the inner product in Hilbert space and we require that:
\begin{equation}
\left\langle\mathcal{S},\mathcal{S}\right\rangle=\int_0^{\pi}(\mathcal{S}^{+2}(\theta)+\mathcal{S}^{-2}(\theta)){\rm d}\theta=:\int_0^{\pi}U(\theta){\rm d}\theta=1.
\end{equation}
Integrating by parts:
\begin{align}
\frac{\partial \lambda}{\partial \nu}&=\int_0^{\pi}V(\theta){\rm d}\sin\theta=-\int_0^{\pi}\sin\theta V^{\prime}(\theta){\rm d}\theta\nonumber \\
&=-\int_0^{\pi}\sin\theta\left[\frac{2m}{\sin\theta}U(\theta)-2\xi\sin\theta U(\theta)+2\lambda W(\theta)\right]{\rm d}
\theta\nonumber\\
&=-\left[2m+2\lambda \frac{\partial \lambda}{\partial \xi}-2\xi+\int_0^{\pi}2\xi\cos^2\theta U(\theta){\rm d}\theta\right].
\end{align}
On the other hand from the angular KN equations, $W^{\prime}(\theta)=2\nu\cos\theta U(\theta)-2\lambda V(\theta)$ and integrating by parts:
\begin{align}
\int_0^{\pi}2\nu\cos^2\theta U(\theta){\rm d}\theta&=\int_0^{\pi}\cos\theta W^{\prime}(\theta){\rm d}\theta+\int_0^{\pi}2\lambda \cos\theta V(\theta){\rm d}\theta\nonumber \\
&=\cos\theta W(\theta)\Bigl|_0^{\pi}-\int_0^{\pi}W(\theta)(-\sin\theta){\rm d}\theta+2\lambda\frac{\partial\lambda}{\partial \nu} \nonumber\\
&=+\frac{\partial \lambda}{\partial \xi}+2\lambda \frac{\partial \lambda}{\partial \nu}.
\end{align}
Thus eqn.(\ref{eigenvalueeqn}) is proved.
\subsubsection{Necessary and sufficient conditions for bound states in the extreme Kerr-Newman spacetime}\label{extremeBound}
The following theorem was proven in \cite{HSchmidt}:
\begin{theorem}
A point $\omega\in \Bbb{R}$ is an energy eigenvalue of the Dirac equation in extreme Kerr-Newman geometry (i.e. when $\Delta^{KN}$ has a double root and $ M=\sqrt{a^2+e^2}$) if and only if there exists an eigenvalue $\lambda$ of the angular KN equation such that:
\begin{align}
\omega&=-\frac{-ma+eq M}{a^2+M^2},\;\;\mu^2-\omega^2>0,\;\;\lambda^2+M^2\mu^2-(2M\omega+eq)^2>\frac{1}{4},\label{fbound1}\\
&(r_{+}=r_{-}=M=\sqrt{a^2+e^2})\nonumber\\
&\text{and\; either\;} \beta_1-\sigma\lambda=0,\;\;\alpha_1+\varkappa=0\;\; \text{or} \;\; 1+N+\alpha_1+\varkappa=0\label{fbounds3},
\end{align}
where
\begin{align}
\beta_1&:=\frac{\mu(M|\omega|-\sigma(2M\omega+eq))}{\sqrt{\mu^2-\omega^2}},\;\sigma:={\rm{sign}}\omega\\
\alpha_1&:=\frac{\mu^2M-\omega(2M\omega+eq)}{\sqrt{\mu^2-\omega^2}},\;\varkappa:=\sqrt{\lambda^2+M^2\mu^2-(2M\omega+eq)^2},\;N\in\Bbb{Z}_{\geq0}.
\end{align}
\end{theorem}
Now the condition $\alpha_1+\varkappa=0$ can be written as follows:
\begin{equation}
\alpha_1+\varkappa=0\Leftrightarrow M\sqrt{\mu^2-\omega^2}-\frac{\omega(eq+M\omega)}{\sqrt{\mu^2-\omega^2}}+
\sqrt{\lambda^2+M^2\mu^2-(2M\omega+eq)^2}=0.
\end{equation}
Thus we conclude that the following inequality holds:
\begin{equation}
\omega(eq+M\omega)>0.
\label{fbound4}
\end{equation}


\begin{thebibliography}{99}

\bibitem{Kraniotis1} G. V. Kraniotis, \textit{The Klein-Gordon-Fock equation in the curved spacetime of the Kerr-Newman (anti) de Sitter black hole}  Class. Quantum Grav. \textbf{33} (2016) 225011; G. V. Kraniotis, CQG+ Insight: \textit{The problem of perturbative charged massive scalar field in the Kerr-Newman-(anti) de Sitter black hole background}, 21/11/2016 and references therein

    \bibitem{Bezerra} V.B. Bezerra, H. S. Vieira and A. A. Costa, \textit{The Klein-Gordon equation in the spacetime of a charged and rotating black hole} Clas.Quantum Grav. \textbf{31} (2014) 045003; H.S. Vieira, J. P. Morais Gra\c ca, V.B. Bezerra,\textit{	
Scalar resonant frequencies and Hawking effect of an  f(R)  global monopole} Chin.Phys. C41 (2017) no.9, 095102
\bibitem{WU} S. Wu and X. Cai, J. Math. Phys. \textbf{40} (1999),4538
    \bibitem{GW150914} B. P. Abbott \textit{et al},\textit{Observation of Gravitational Waves from a Binary Black Hole Merger} Phys.Rev.Lett.\textbf{116}, 061102 (2016)
\bibitem{GW151226} B. P. Abbott \textit{et al},\textit{GW151226: Observation of Gravitational Waves from a 22-Solar-Mass Binary} Phys.Rev.Lett.\textbf{116}, 241103 (2016);B. P. Abbott \textit{et al},\textit{GW170104:Observation of a 50-Solar-Mass Binary Black Hole Coalescence at Redshift 0.2} Phys. Rev. Lett.\textbf{118},221101 (2017); B. P. Abbott \textit{et al},\textit{GW170814: A Three-Detector Observation of Gravitational Waves from a Binary Black Hole Coalescence} Phys.Rev.Lett.\textbf{119}, 141101 (2017);  B. P. Abbott \textit{et al},\textit{GW170817: Observation of Gravitational Waves from a Binary Neutron Star Inspiral} Phys.Rev.Lett.\textbf{119}, 161101 (2017)

    \bibitem{KARLHEUNmunich} K. Heun, \textit{Zur Theorie der Riemann'schen Functionen zweiter Ordnung mit vier Verzweipunkten}, Mathematische Annalen 33, (1889) pp 161-179
\bibitem{RONVEAUX} A. Ronveaux. \textit{ed} \textit{Heun's Differential Equations} Oxford Science Publications, OUP (1995)

    \bibitem{Erdelyi} A. Erd\'{e}lyi, \textit{Certain expansions of solutions of the Heun equation}, Q. J. Math.(Oxford), \textbf{15} (1944), pp. 62-69
\bibitem{PMDiracI} P. A. M. Dirac, \textit{The Quantum Theory of the Electron} Proc.R. Soc. Lond.A \textbf{117}(1928) pp 610-624
\bibitem{V. Fock} V. Fock \textit{Geometrisierung der Diracschen Theorie des Elektrons} Zeitschrift f\"{u}r Physik \textbf{57}, (1929), pp 261-277
\bibitem {Stuchlik1}Z. Stuchl\'{\i}k, G. Bao, E. \O stgaard and S.
Hled\'{\i}k,\textit{\ Kerr-Newman-de Sitter black holes with a restricted
repulsive barrier of equatorial photon motion, }Phys. Rev. D. \textbf{58}
(1998) 084003
\bibitem {BCAR}B. Carter, \textit{Global structure of the Kerr family of
gravitational fields }Phys.Rev.\textbf{174 }(1968)1559-71


\bibitem{ZdeStu} Z. Stuchl\'{\i}k and S.Hled\'{\i}k,
\textit{Equatorial photon motion in the Kerr-Newman spacetimes
with a non-zero cosmological constant}, Class. Quantum Grav.
\textbf{17} (2000) 4541-4576

\bibitem {GrifPod}J. B. Griffiths and Ji\v{r}\'{\i} Podolsk\'{y},
\textit{Exact spacetimes in Einstein's General Relativity, }Cambridge
Monographs on Mathematical Physics, Cambirdge University Press (2009)

\bibitem{ZST} Z. Stuchl\'{\i}k, \textit{The motion of test particles in black-hole backgrounds
with non-zero cosmological constant}, Bull. of the Astronomical
Institute of Chechoslovakia \textbf{34} (1983) 129-149

\bibitem{Chandrasekhar} S. Chandrasekhar, \textit{The Mathematical Theory of Black Holes} Oxford Classic Texts in Physical Sciences, 1992
    \bibitem{Saul}S. A. Teukolsky, The Kerr metric, Class. Quantum Grav. (2015) 32 124006

 \bibitem{NPformalism} E. Newman and R. Penrose, J. Math.Phys. \textbf{3}(1962) 566
\bibitem{spinors} R. Penrose and W. Rindler, \textit{Spinors and space-time} Vol.I: \textit{Two-spinor calculus and relativistic fields} Cambridge Monographs on Mathematical Physics,(1984) CUP; Vol.2,\textit{Spinor and twistor methods in space-time geometry} (1986) CUP
\bibitem{page} D. N. Page, \textit{Dirac equation around a charged, rotating black hole}, Phys. Rev. D14 (1976) 1509-1510; C. H. Lee, Phys.Lett.B 68 (1977)152

    \bibitem{WKinnersley} W. Kinnersley, J. Math.Phys.\textbf{10}, (1969) 1195

    \bibitem{Baticr} D. Batic, H. Schmid and M. Winklmeier, J. Phys. A. Math.Gen. \textbf{39}, (2006) 12559-12564
        \bibitem{Khanal}U. Khanal, Phys.Rev.D \textbf{28},(1983) 1291-1297
        \bibitem{Bose} S. K. Bose, J. Math. Phys. \textbf{16}, (1975) 772
        \bibitem{Finster}   F. Finster, N. Kamran, J. Smoller and S.T. Yau, Commun. Pure Appl. Math.(2000), pp.902-929
        \bibitem{HSchmidt} H. Schmid, \textit{Bound state solutions of the Dirac equation in the extreme Kerr geometry}, Math. Nachr.\textbf{274-275} (2004),pp 117-129
            \bibitem{MonikaWinklmeier}M. Winklmeier and O. Yamada, \textit{Spectral analysis of radial Dirac operators in the Kerr-Newman metric and its applications to time-periodic solutions} J.Math.Phys.\textbf{47} (2006)102503,(17 pp)
            \bibitem{FBelgiornoMKN} F. Belgiorno and M. Martellini, Phys.Lett.B \textbf{453}(1999) pp 17-22
            \bibitem{Belgiorno} F. Belgiorno and S. L. Cacciatori, J. Phys.A: Math.Theor.\textbf{42}(2009) 135207 (15pp); J. Math.Phys.\textbf{51} (2010),033517
                \bibitem{MonikaDB}D. Batic, H. Schmid and M. Winkimeier, J. Math.Phys.\textbf{46} (2005),012504
\bibitem{RSCHAFKEDSCHMIDT} R. Sch\"{a}fke and D. Schmidt, \textit{The connection problem for general linear ordinary differential equations at two regular singular points with applications to the theory of special functions}, Siam J. Math. Anal.(1980) Vol.11, 848-862
    \bibitem{PPainleve} P. Painlev\'{e},\textit{Memoire sur les \'{e}quations diff\'{e}rentielles dont l'int\'{e}grale g\'{e}n\'{e}rale est uniforme},  Bull.Soc. Math. France, \textbf{28},(1900) pp 201-261
        \bibitem{Garnier} R. Garnier, \textit{Sur les \'{e}quations diff\'{e}rentielles du troisi\'{e}me ordre dont l'int\'{e}grale g\'{e}n\'{e}rale est uniforme et sur une classe d' \'{e}quations nouvelles d'ordre sup\'{e}rieur ont l'int\'{e}grale g\'{e}n\'{e}rale a ses points critiques fixes }, Ann.Sci.Ec.Norm.Super. \textbf{29} (1912),pp. 1-126
            \bibitem{Donzwei} D. N. Page, \textit{Particle emission rates from a black hole III. Charged leptons from a nonrotating hole}, Phys.Rev.D\textbf{16} (1977), pp. 2402-2411

\bibitem{OLVER} F. W. J. Olver, \textit{Asymptotics and special functions}
Academic Press,  Editor W. Rheinbolt (1974)

\bibitem{RegIrregSIAMJ}R. Sch\"{a}fke, \textit{A connection problem for a regular and an irregular singular point of complex ordinary differential equations}, Siam J. Math.Anal. Vol 15 (1984) 253-271

\bibitem{RegWheeTortoise} T. Regge and J. A. Wheeler, \textit{Stability of a Schwarzschild Singularity}Phys.Rev. \textbf{108} (1957),pp.1063-1069

\bibitem{Erdelyi} A. Erd\'{e}lyi, \textit{Certain expansions of solutions of the Heun equation}, Q. J. Math.(Oxford), \textbf{15} (1944), pp. 62-69
\bibitem{Poincare}H. Poincar\'{e},\textit{Sur les Equations Lin\'{e}aires aux Diff\'{e}rentielles Ordinaires at aux Diff\'{e}rences Finies},Am.J.of Mathematics, \textbf{7} (1885) pp 203-258
            \bibitem{Perron} O. Perron, \textit{\"{U}ber die Poincar\'{e}sche lineare Differenzengleichung} Journal f\"{u}r die reine und angewandte Mathematik \textbf{137}(1909) pp. 6-64
                \bibitem{Svartholm} N. Svartholm,\textit{Die L\"osung der Fuchssehen Differentialgleichung zweiter Ordnung durch hypergeometrische Polynome} Mathematische Annalen \textbf{116} (1939) 413-421
\bibitem{KatoLPT} T. Kato,\textit{Perturbation Theory for Linear Operators}, Classics in Mathematics , Springer Verlag, Berlin Heidelberg (1995), Reprint of the 1980 edition
    \bibitem{Kraniotis2} G. V. Kraniotis, Work in Progress.
    \bibitem{Conway}J. B. Conway, \textit{A course in Functional Analysis}, Springer GTM \textbf{96}(2007)
        \bibitem{Neznamov} V. P. Neznamov and I.I. Safronov, \textit{The effective method to calculate eigenvalues of Chandrasekhar-Page angular equations}, Int.J. Mod. Phys.\textbf{25},(2016) 1650091
            \bibitem{SRDolan} S.R. Dolan and J. R. Gair, \textit{The massive Dirac field on a rotating black hole spacetime: angular solutions} Clas.Quantum.Grav. \textbf{26} (2009) 175020 (26pp)

    \bibitem{Clarkson} P. A. Clarkson, \textit{The third Painlev\'{e} equation and associated special polynomials},(2003) J. Phys. A. Math. Gen. \textbf{36} 9507
        \bibitem{Milne} A. E. Milne, P. A. Clarkson, and P. Bassom, \textit{B\"{a}cklund Transformations and Solution Hierarchies for the Third Painlev\'{e} Equation}, Stud. Appl. Math. \textbf{98}(1997) pp 139-194
            \bibitem{solitons} M. J. Ablowitz and P. A. Clarkson, \textit{Solitons, Nonlinear Evolution Equations and Inverse Scattering}, London Mathematical Society Lecture Note Series \textbf{149},(1991), Cambridge University Press
                \bibitem{Lukashevich} N. A. Lukashevich, \textit{Elementary solutions of certain Painlev\'{e} equations} Differ. Uravn. (1965),\textbf{1}, pp. 731-735

                \bibitem{Boutroux} P. Boutroux, Ann.Ecole.Norm.Super.,\textbf{30} (1913),pp 255-375; P. Boutroux, \textit{Recherches sur les transcendantes de M. Painlev\'{e} et l' \'{e}tude asymptotique des \'{e}quations du second ordre (suite)} Ann.Ecole.Norm.Super.,\textbf{31} (1914),pp. 99-159

            \bibitem{GrattamGuiness} I.Grattan-Guinness and S.Engelsman, \textit{The manuscripts of Paul Charpit}, Historia Mathematica, \textbf{9}, (1982), pp. 65-75
\bibitem{Flaschka} H. Flaschka and A. C. Newell, \textit{Monodromy- and Spectrum-Preserving Deformations I}, Comm. math. Phys.\textbf{76} (1980),pp 65-116
    \bibitem{Jimbo1} M. Jimbo, T. Miwa and K. Ueno, \textit{Monodromy preserving deformation of linear ordinary differential equations with rational coefficients}, Phys.2D (1981),pp. 306-352
        \bibitem{LaxPeter} P. D. Lax, \textit{Integrals of Nonlinear Equations of Evolution and Solitary waves}, Commun. on Pure and Appl.Math., \textbf{XXI}, (1968), pp.467-490

        \bibitem{Joshi} N. Joshi and E. Liu, Nonlinearity \textbf{31} (2018) 3726-3747
            \bibitem{FermEliza} R. A. Konoplya and A. Zhidenko, \textit{Quasinormal modes of massive fermions in Kerr spacetime:Long-lived modes and the fine structure}, Phys.Rev.D\textbf{97} (2018),084034; M. Casals \textit{et al}, \textit{Quantization of fermions in Kerr spacetime}, Phys.Rev.D 87 (2013),064027; A. Coutant and P. Millington,Clas.Quantum Grav.(2019) \textbf{36} 035005,arXiv:1809.03480v2
\end{thebibliography}
\end{document}